\documentclass[12pt,twoside]{article}

\newcommand{\blind}{0}
\usepackage{setspace}

\def\spacingset#1{\renewcommand{\baselinestretch}%
{#1}\small\normalsize} \spacingset{1}

\RequirePackage{amsthm,amsmath,amsfonts,amssymb}
\RequirePackage[authoryear]{natbib}
\RequirePackage[colorlinks,citecolor=blue,urlcolor=blue]{hyperref}
\RequirePackage{graphicx}

\RequirePackage[colorlinks,citecolor=blue,urlcolor=blue]{hyperref}
\usepackage{cleveref}
\usepackage{tikz}
\usetikzlibrary{fit,positioning,calc}
\usepackage[utf8]{inputenc}
\usepackage{graphicx}
\usepackage{colortbl,array} 
\usepackage{multirow,bigdelim}
\usepackage{booktabs}
\newcommand{\M}{\mathcal{M}}

\usepackage{setspace}
\usepackage{enumerate}     
 
\usepackage{multirow}

\theoremstyle{plain}

\newtheorem{thm}{Theorem}[section]
\newtheorem{prop}{Proposition}[section]
\newtheorem{definition}[thm]{Definition}

\theoremstyle{remark}

\usepackage[ruled,vlined]{algorithm2e}

\usepackage[noend]{algpseudocode}

\newcommand{\by}{\boldsymbol{y}}

\newcommand{\bu}{\boldsymbol{u}}
\newcommand{\bb}{\boldsymbol{\beta}}

\newcommand{\bdt}{\boldsymbol{\theta}}

\newcommand{\norm}[1]{\lVert \, #1 \, \rVert}

\newcommand{\commentt}[1]{}

\RequirePackage{xr-hyper} 
\begin{document}

 \date{}
 

\if0\blind
{
\title{\LARGE\bf Fully Bayesian  Spectral Clustering and Benchmarking with  Uncertainty Quantification  for Small Area Estimation}
  \author{Jairo F\'uquene-Pati\~no\thanks{\texttt{E-mail: jafuquenepatino@ucdavis.edu.}}\\ 
Department of Statistics and Data Science,   University of California, Davis}
  \maketitle
} \fi

\if1\blind
{
  \bigskip
  \bigskip
  \bigskip
  \begin{center}
    {\LARGE\bf Fully Bayesian  Spectral Clustering and Benchmarking with  Uncertainty Quantification  for Small Area Estimation}
\end{center}
  \medskip
} \fi

\pagestyle{myheadings}
\markboth{
\small\itshape \textbf{Jairo F\'uquene-Pati\~no}
}{
\small\itshape \textbf{Fully Bayesian  Spectral Clustering and Benchmarking for SAE}
}

\bigskip

\begin{abstract}
In this work, inspired by machine learning techniques, we propose a new Bayesian model for Small Area Estimation (SAE), the Fay-Herriot model with Spectral Clustering (FH-SC). Unlike traditional approaches, clustering in FH-SC is based on spectral clustering algorithms that utilize external covariates, rather than geographical or administrative criteria. A major advantage of the FH-SC model is its flexibility in integrating existing SAE approaches, with or without clustering random effects. To enable benchmarking, we leverage the theoretical framework of posterior projections for constrained Bayesian inference and derive closed form expressions for the new Rao-Blackwell (RB) estimators of the posterior mean under the FH-SC model. Additionally, we introduce a novel measure of uncertainty for the benchmarked estimator, the Conditional Posterior Mean Square Error (CPMSE), which is generalizable to other Bayesian SAE estimators. We conduct model-based and data-based simulation studies to evaluate the frequentist properties of the CPMSE. The proposed methodology is motivated by a real case study involving the estimation of the proportion of households with internet access in the municipalities of Colombia. Finally, we also illustrate the advantages of FH-SC over existing Bayesian and frequentist approaches through our case study.

\end{abstract}

\noindent%


\newpage

\spacingset{1.0} 


\section{Introduction}
\label{sec:intro}

As part of the 2030 agenda for Sustainable Development of the United Nations (UN) \citep{UN2015},  NSOs (National Statistical Offices) in low- and middle-income countries  have advocated for the implementation of  Small Area Estimation (SAE) models  to estimate key official indicators when sample sizes in surveys at subnational (e.g., county, municipality, department, state) or subpopulation (gender, age, socio-economic class) levels are small. Reviews of SAE techniques are discussed  in  \cite{ghosh2020small},  \cite{pfeffermann2002small}, and more recently  in  \cite{molina2023historical}. To obtain accurate estimations of key indicators, practitioners in NSOs use Bayesian and frequentist models based on the Fay-Herriot (FH) model \citep{fay_1979}, with covariates provided from different sources (e.g., administrative sources, population census, surveys). The inclusion of additional information in SAE models using spatial dependence of small areas based on geographical and administrative criteria has been extensively studied  \citep{opsomer2008non, pratesi2009small2, marhuenda2013small, ross2015bayesian,  porter2014spatial, mercer2014comparison, mercer2015space, porter2015multivariate, steorts_etal_2019, chung2020bayesian}. In addition, clustering small areas to improve estimation has been explored in \cite{maiti2014clustering} and \cite{torkashvand2017clustering}. Specifically, small areas are first clustered, followed by the implementation of linear mixed models with clustering effects to account for non-homogeneity of the random effects.

On the other hand, machine learning techniques can be utilized to incorporate spatial/cluster dependence in SAE models.  
In particular, spectral clustering is a popular set of algorithms for clustering in the machine learning literature  \citep{von2007tutorial}. 
Comprehensive reviews of these algorithms are given in \cite{bach2003learning,   bach2006learning, von2007tutorial, hastie2009elements}. The combination of Spectral Clustering (SC)  based on undirected graphs and their estimation via standard regularization  are useful to build semi-supervised algorithms for classification and regression \cite{galway1992spline, hofmann2008kernel}. 
Our first contribution is motivated by the Regularized Task Kernel Learning  (RTKL) criterion \citep{evgeniou2005learning, miller2023additive}, and involves an objective function similar to that used in Laplacian Regularized Least Squares (LapRLS) estimation \citep{evgeniou2000regularization, scholkopf2002learning, belkin2005manifold, belkin2006manifold}. We propose a new SAE model called the \textit{Fay-Herriot model with Spectral Clustering (FH-SC)}.  The FH-SC model includes a Laplacian matrix to ensure smoothness across related small areas obtained with Spectral Clustering (SC) algorithms, as well as a cluster regularization penalty \cite{miller2023additive}. Importantly, our framework is innovative in SAE, as clustering is not based on geographical or administrative criteria but rather on SC algorithms relying on external covariates.


Our proposal shares some similarities with the frequentist approaches in  \cite{maiti2014clustering} and \cite{torkashvand2017clustering},  which attempt to allow non-homogeneity of the random effects to improve SAE. However, our framework is fundamentally different not only in its use of external covariates within the SC algorithm but also in the overall modeling approach. Specifically, we use a regularization method that incorporates a Laplacian matrix within the FH-SC model, allowing for parameter estimation with either smoothness or non-smoothness across small areas. This adaptive feature results in a general FH-SC model that can incorporate existing approaches with clustering random effects, such as those in \cite{maiti2014clustering} and \cite{torkashvand2017clustering}, as well as non-clustering effects, such as the FH model. Consequently, the models of \cite{maiti2014clustering} and \cite{torkashvand2017clustering} can be viewed as complementary to our proposed FH-SC model.

In addition to following the best procedures to include auxiliary information in SAE, practitioners in NSOs typically  need to 
include estimates with higher levels of aggregation in their official statistical reports. For instance, in Colombia, estimates at the
municipality level need to be aggregated according to external estimates at the departmental or national level. This procedure of adjusting small area estimates by imposing benchmarking constraints is known as benchmarking
\cite{datta_2011, fabrizi2012constrained, fabrizi2014mapping, berg2012benchmarked, ghosh_steorts2013, ghosh2015benchmarked}.  Important contributions using posterior projections to provide posterior summaries of benchmarked estimators were proposed by \cite{patra2019constrained, patra2024constrained}. Specifically, the computational and theoretical aspects of Bayesian posterior projections under parameter constrains have been studied in \citep{dunson2003bayesian2, gunn2005transformation, lin2014bayesian}. The recent work of \cite{patra2024constrained} provides a general framework of Bayesian posterior projections with new theoretical and computational advances. Notably, the use of posterior projections in SAE was originally proposed by \cite{patra2019constrained} to obtain a posterior distribution of benchmarked estimators based on projected samples for area level models under the FH model.


In recent years, alternatives to assess the uncertainty of benchmarked estimators have been proposed in the Bayesian context. 
For instance, \cite{you2002benchmarking} propose the Posterior Mean Square Error (PMSE) for the estimation of census coverage. Another contribution is \cite{erciulescu2018benchmarking}, which uses the posterior variance of benchmarked posterior samples to measure estimator uncertainty. More recently, \cite{zhang2020fully} introduce a modified likelihood defined via benchmarking and carried out posterior inference by incorporating the benchmarking constraints in Metropolis-Hasting steps. Although measures of uncertainty for benchmarked estimators under the FH model have been proposed \cite{erciulescu2018benchmarking, you2002benchmarking, zhang2020fully}, and posterior summaries (e.g., credible intervals, quantiles, posterior probabilities) of benchmarked estimates under the FH model can be computed as in \cite{patra2019constrained}, no prior research has explored measures of uncertainty for benchmarked estimators obtained through posterior projections.

Therefore, our second contribution is the proposal of a new measure of uncertainty for benchmarked estimators in SAE problems. To this end, using the theory of posterior projections proposed in \cite{patra2019constrained} and \cite{patra2024constrained}, we build a framework of benchmarked estimators under linear equality constraints for the proposed FH-SC model. Due to the hierarchical structure of the FH-SC model, from a probabilistic perspective, the resulting benchmarked estimators are random variables. Consequently, we are able to compute Rao-Blackwell (RB) \citep{Rao1945, blackwell1947conditional} benchmarked estimators using the conditional expectation of the posterior mean obtained under posterior projections. Crucially, we propose a new measure of uncertainty for these estimators called the \textit{Conditional Posterior Mean Square Error (CPMSE)}.

The CPMSE is useful to approximate the Mean Square Error (MSE) of RB benchmarked estimators obtained under posterior projections and, most notably, has the potential to be applied to other SAE estimators obtained under a Bayesian framework. To compute the RB estimators and CPMSE, posterior samples under the FH-SC model are required. For this purpose, we develop a computational framework with new Markov chain Monte Carlo (MCMC) algorithms to produce posterior samples of the model parameters under the FH-SC model. The MCMC schemes for the FH-SC model are inspired by previous work in Bayesian spatial econometrics for the Simultaneously Autoregressive (SAR) and Spatial Probit (SP) models discussed in \cite{lesage2009introduction} and  \cite{wilhelm2013estimating}, respectively. 

Finally, our third contribution involves the use of the proposed methodology to estimate the Proportion of Households with Internet Access (PHIA) in the municipalities of Colombia in 2015. Universal access to the internet is a priority in low- and middle-income countries \citep{Bank}, aligning with the Human Rights Principles for Connectivity and Development \citep{Human, tully2014human}. Our proposed PHIA indicator closely relates to a Sustainable Development Goal in the UN 2030 Agenda \citep{UN2015}, which emphasizes the need for additional indicators to measure internet access at subnational levels in developing countries. Currently, no municipal-level estimates of the PHIA are available in Colombia. 
To the best of our knowledge, our work is the first to address the estimation of the PHIA at subnational levels in Colombia. Furthermore, we believe the proposed methodology has the potential to be implemented in other Latin American countries.

The remainder of the paper is organized as follows. In Section \ref{sec:FHSC}, we introduce the motivational case study and 
the SC algorithms to obtain a cluster classification based on external covariates. We also present
the new FH-SC model for SAE, find closed form expressions of the conditional expectation and variance (Theorem \ref{expectation}) and describe prior specifications under the FH-SC model. In Section \ref{sec:all_estimation_methods},  we define the Rao-Blackwell (RB) estimators utilized to produce small area estimates of PHIA and consider the theory of benchmarking and posterior projections (Proposition \ref{prop:benchmarking}  and Theorem \ref{thm:PP_problem}) to obtain a full posterior for benchmarking estimation. Importantly, we also propose the CPMSE (Proposition \ref{prop:def3}) and discuss model selection criteria that are useful for evaluating the performance of the various models in our case study.  In Section \ref{sec:application},  we estimate the PHIA in the municipalities of Colombia using the RB estimators and their uncertainty computed with the proposed CPMSE under the FH model, the Bayesian versions of the proposed models in \cite{maiti2014clustering, torkashvand2017clustering}, and three versions of the FH-SC model. 
In Section \ref{sec:application}, we illustrate how the cluster classification of covariates in the proposed FH-SC model allows us to reduce both the CPMSE and the Coefficient of Variation (CV) of posterior estimates of PHIA. Section \ref{sec:discussion} outlines potential future extensions. Simulation studies to evaluate the performance of the proposed RB benchmarked estimator and investigate the frequentist properties of the CPMSE are included in the supplement (Section \ref{sec:simula}).

\section{Model for SAE using Spectral Clustering}
\label{sec:FHSC}

To describe our proposed methodological framework, we first introduce our motivating case study. We are interested in estimating internet connectivity in the municipalities of Colombia using a fully Bayesian approach with spectral clustering and benchmarking.

\subsection{Motivating example}
\label{subsec:Motivating}

The Proportion of Households with Internet Access (PHIA) is a useful measure of internet access and its use. Unfortunately, indicators of internet connectivity at lower geographical levels are typically unavailable in low- and middle-income countries. This is the case in Colombia, where the NSO \citep{internet} and the Ministry of Health \citep{Profamilia} can only compute PHIA estimates at the national or departmental levels. While these estimates are useful, PHIA estimates at finer subnational levels (e.g., municipalities, rural and urban areas, socioeconomic classes) are crucial for effective policy-making.

Our goal is to estimate the PHIA in  $m=294$ municipalities of Colombia with high precision. We consider the  most recent  Demographic and Health Survey (DHS)  implemented in Colombia in 2015 \citep{Profamilia}, which is valuable to estimate important public health indicators of the Colombian population in the most representative municipalities and capital cities. To compute the direct estimate of the PHIA  and the corresponding direct variance for the $i$-th small area denoted as $y_i$ and $D_{i}$, respectively, we use the estimators given by \cite{hajek1971comment} and
the Generalized Variance Function (GVF)  \citep{wolter1985introduction} to smooth the direct variances (Section \ref{sub_applied2}). 

Due to the intrinsic socioeconomic characteristics in Colombia \citep{Profamilia, internet}, we expect the direct estimates of PHIA to correlate with poverty and education indicators as reported in Mexico \citep{mora2021internet} and developing countries in non-Mediterranean Africa \citep{frankfurter2020measuring}. For instance, we expect direct estimates of PHIA to be closely related with external covariates provided by the Education Deficit Index  \citep{roser2014human} and Multidimensional Poverty Index (MPI) \citep{alkire2015multidimensional}. In practice, these covariates are obtained from national censuses and/or administrative sources available through the National Statistical System (SEN, by its Spanish acronym) \cite{DANE2020}. Specifically, the Educational Index and MPI can be obtained from the 2005 Population Census conducted in Colombia \citep{DANE2005} and adjusted using variables from the 2014 Census of Agriculture \citep{CNA2014}. Next, we discuss spectral clustering with these type of external covariates.

\subsection{Spectral Clustering with external covariates for SAE}
\label{subsec:The_SC}

Clustering of small areas to improve precision in SAE has been previously explored. In \cite{maiti2014clustering} and \cite{torkashvand2017clustering}, small areas are clustered using the stochastic search algorithm \cite{booth2008clustering} and hierarchical clustering  \cite{ward1963hierarchical} according to the Euclidean distance between covariates. However, Spectral Clustering (SC) algorithms are popular in the machine learning literature \cite{bach2003learning,   bach2006learning, von2007tutorial, hastie2009elements}. As pointed out by \cite{ng2001spectral}, SC algorithms can be easily implemented by practitioners and often outperform traditional clustering algorithms such as \textit{k}-means \cite{macqueen1967some}, particularly, when the data has irregular shapes \cite{hastie2009elements}. To implement SC algorithms in practice, assumptions on specific mixture component densities \cite{fraley2002model,booth2008clustering} or 
asymptotic distributions for cluster-specific variances \cite{torkashvand2017clustering} are not required.  

Several spectral clustering algorithms with unnormalized or
normalized  Laplacian matrices are discussed in \cite{von2007tutorial}
when two variables are considered. To cluster the municipalities in our case study into $c=1,...,C$ clusters with $C\leq m$, we propose the spectral clustering Algorithm \ref{alg:SC}, detailed in the Supplement. 
Supplementary Algorithm \ref{alg:SC} follows similar steps to the algorithm proposed in \cite{von2007tutorial} but extends it by incorporating more than two variables and introducing a method to select the number of covariates and clusters. Specifically, the input of Algorithm \ref{alg:SC} considers the direct estimates of PHIA, $y_{i}$, and  $k=1,...,p^{*}$ external covariates where $\boldsymbol{x}^{*}_{k}=(x^{*}_{1,k},...,x^{*}_{m,k})^{T}$. In our case, $p^{*}=2$ and $x^{*}_{i,1}$ and $x^{*}_{i,2}$ represent the observed values of the Education Deficit Index  and MPI in the $i$-th municipality. 

\begin{figure}[h!]
\begin{center}
\begin{tabular}{ccc} \scriptsize
(a) Total Within-Cluster Sums of Squares \vspace{-0.2cm}& \hspace{-1cm}  \scriptsize (b)  Weighted adjacency matrix \vspace{-0.2cm}\\
\includegraphics[width=0.4\textwidth]{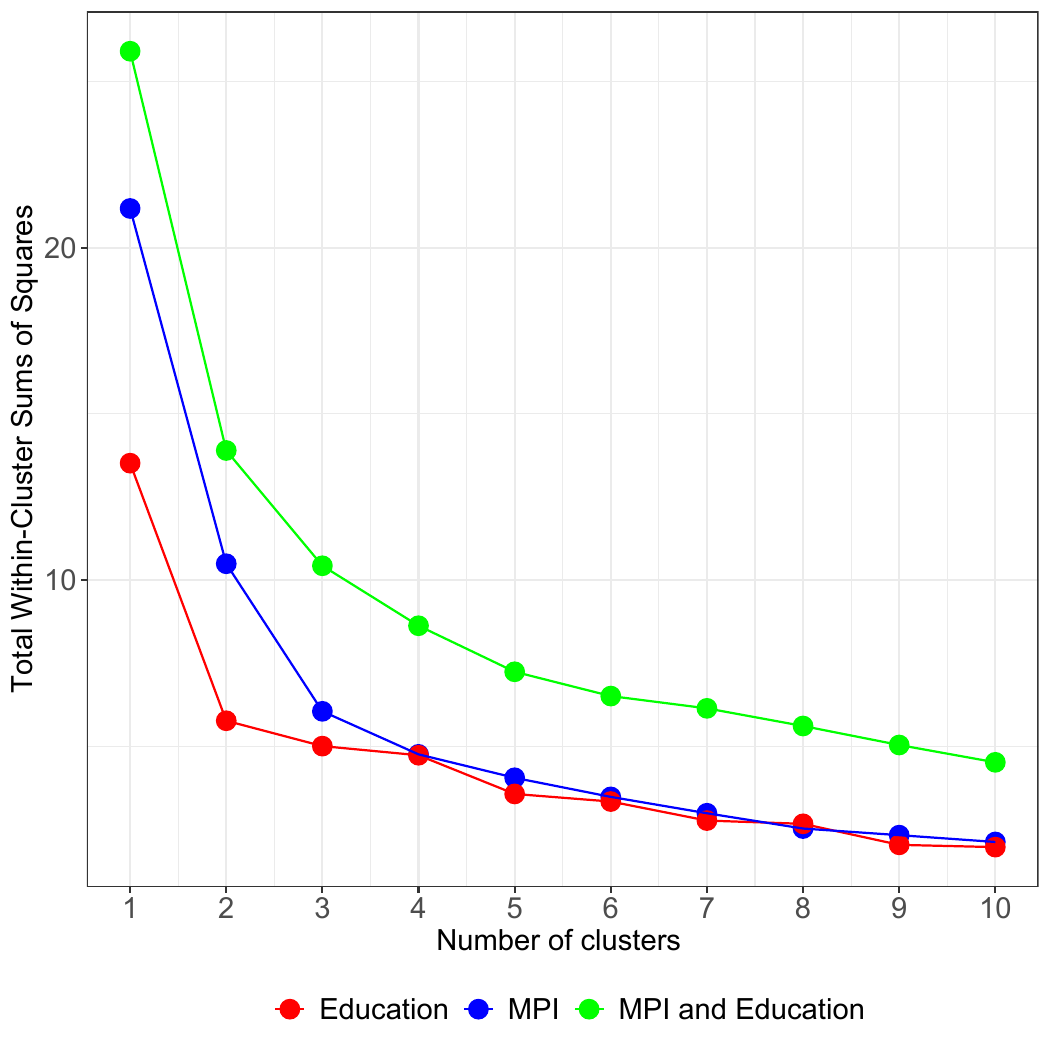} &\hspace{-0.8cm}
\includegraphics[width=0.45\textwidth]{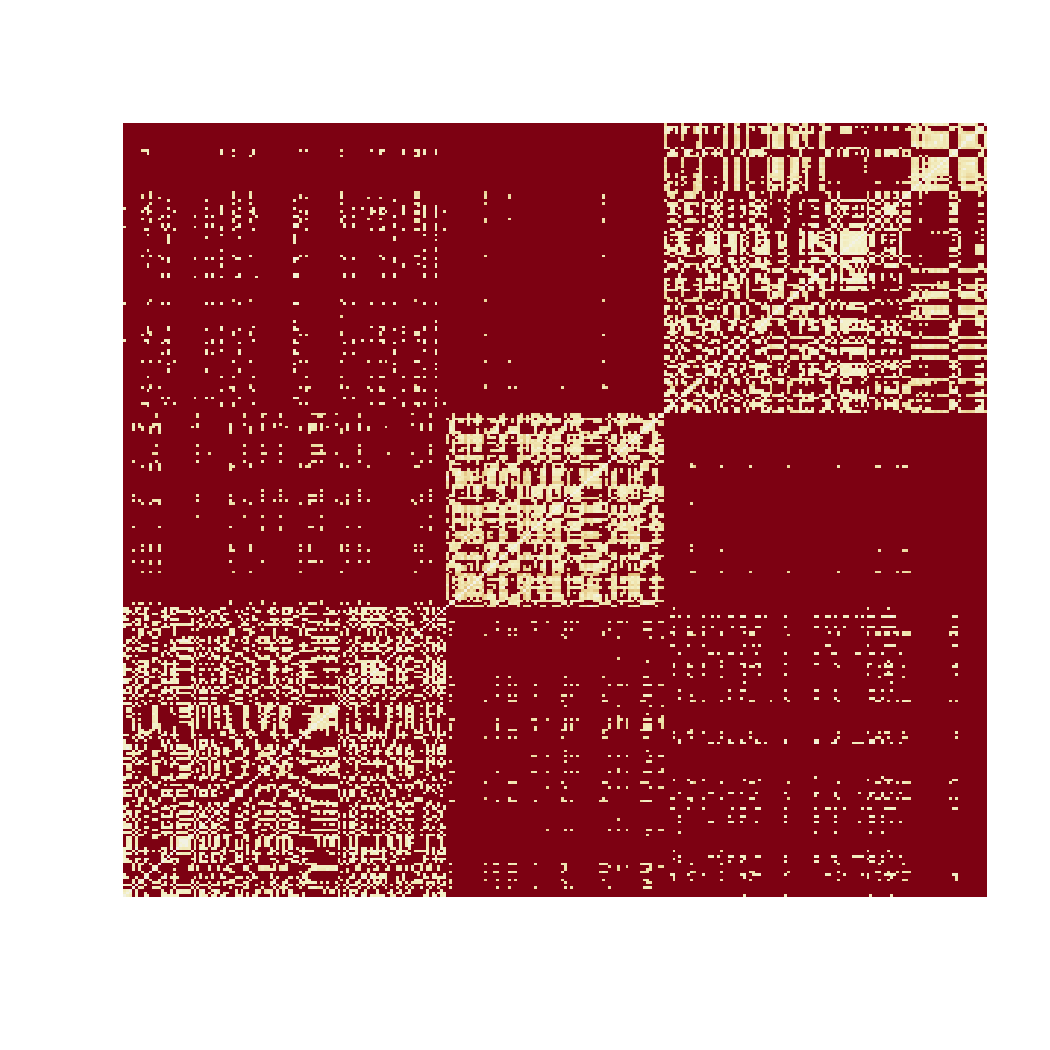} \\
\scriptsize (c)  Direct estimates of PHIA & \hspace{-1.4cm}  \scriptsize (d) Cluster classification  \vspace{-1cm}  \\
\hspace{-1cm} \includegraphics[width=0.55\textwidth]{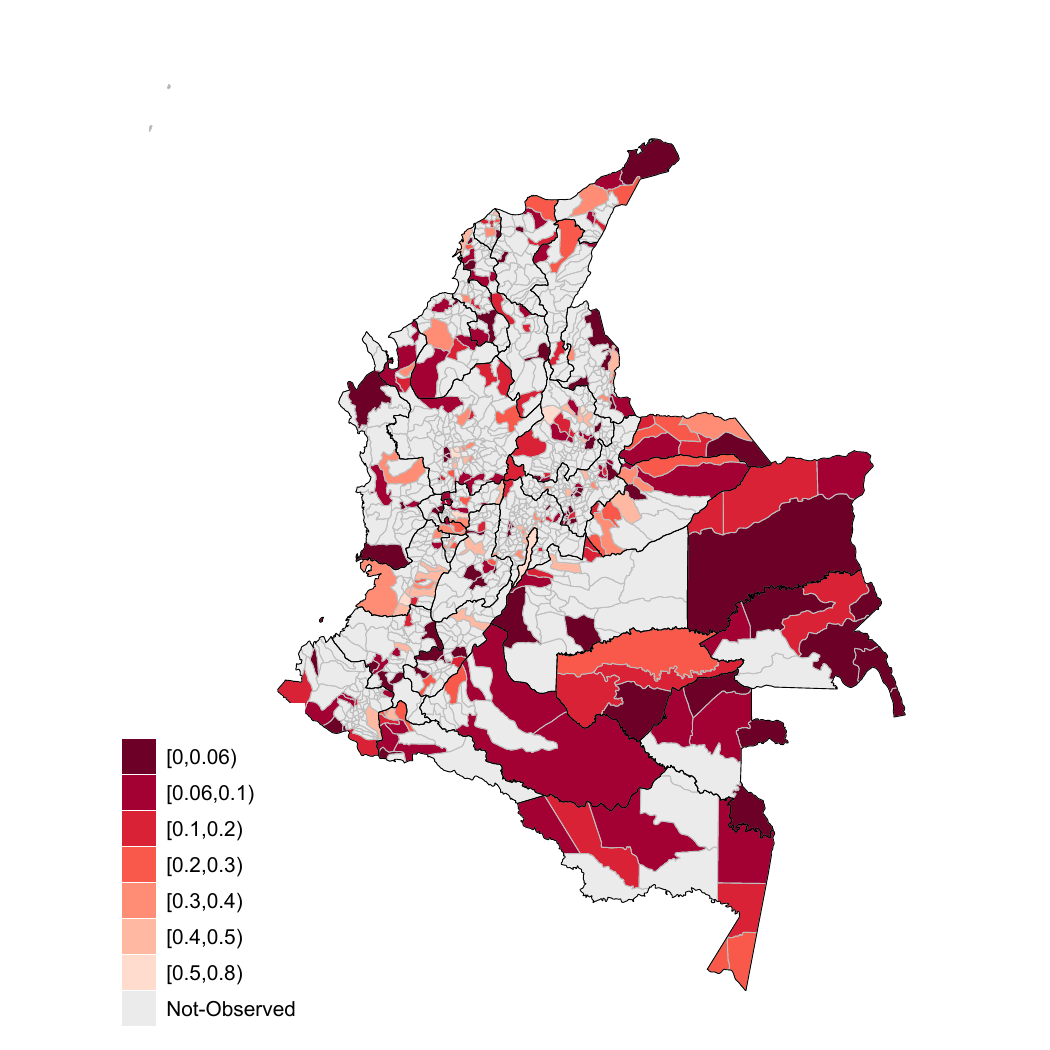} & \hspace{-0.3cm}
\includegraphics[width=0.47\textwidth]{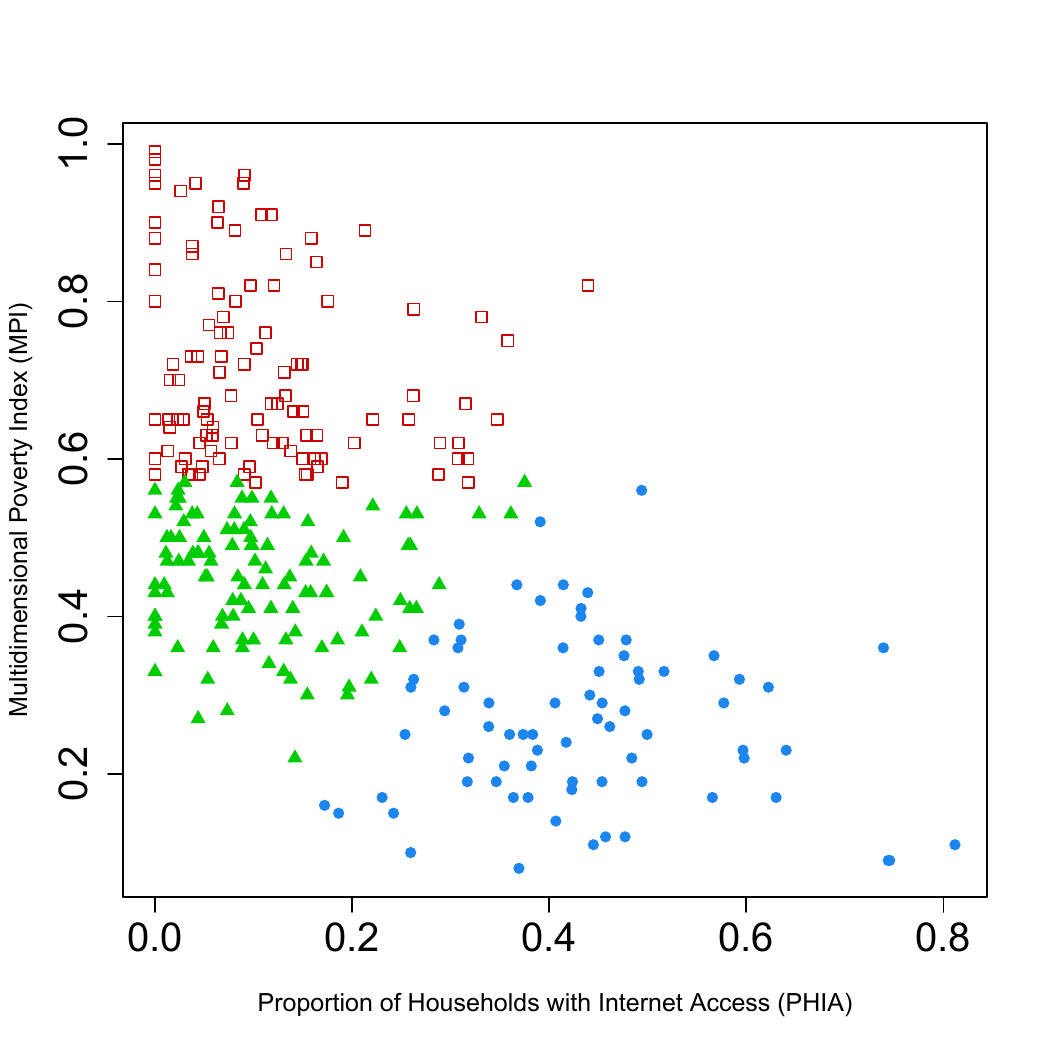}  \\
\end{tabular}
\end{center}
\vspace{-0.5cm}
\caption{ \footnotesize (a) Total Within-Cluster Sums of Squares for each combination of external covariates and clusters.
(b) Reordered weighted adjacency matrix associated to the cluster classification obtained with Algorithm \ref{alg:SC} using as input the direct estimates of PHIA, $C=3$ and the MPI as external covariate. (c) Direct estimates of PHIA in the municipalities of Colombia in 2015. (d) Cluster classification of direct estimates of PHIA and Multidimensional Poverty Index (MPI) obtained with \ref{alg:SC}. The different colors illustrate the classification of the two variables into three clusters.}
\label{fig:maps1}
\end{figure}

The output of Algorithm \ref{alg:SC} provides a cluster classification of the external covariates and the PHIA at the municipality level, along with
the simple graph Laplacian matrix, $\textsf{L}_{SC}= \text{blkdiag}(\{\textsf{L}_{c}\}_{c=1}^{C})$, where $\textsf{L}_{c}=n_c\mathbb I_{n_{c}} - \boldsymbol{1}_{n_c}\boldsymbol{1}_{n_c}^\textsf{T}$ represents the Laplacian matrix for cluster $c$,  $n_{c}$ denotes the number of small areas in cluster $c$,  $\mathbb I_{n_{c}}$ denotes the identity matrix of order $n_c$, 
and $\boldsymbol{1}_{n_c}$ denotes a vector of ones of length  $n_c$. The procedure to include $\textsf{L}_{SC}$ in the proposed SAE model is discussed in the next section. Additionally, to determine the number of clusters and external covariates, Algorithm \ref{alg:SC} computes the total within-cluster sum of squares for each combination of clusters and covariates. For brevity, the details and steps of Algorithm \ref{alg:SC} are provided in Section \ref{SC_algh} of the Supplementary Material.




Figure \ref{fig:maps1}-(a) shows that the MPI and Education Deficit Index  combined have the largest total within-cluster sum of squares. Since MPI is linked to socioeconomic characteristics and shows the largest drop from two to three clusters, we choose it as the external covariate to define $\textsf{L}_{SC}$ in Algorithm \ref{alg:SC}. Importantly, $C=3$ is also associated with the socioeconomic classes division in Colombia into three main categories: low, middle and high income. To graphically assess the clustering structure in our case study, we rearrange the rows and columns of the weighted adjacency matrix obtained from the algorithm. Figure \ref{fig:maps1}-(b) displays the reordered weighted adjacency matrix with dense blocks along the diagonal showing intra-cluster connections. As discussed throughout this paper, setting $C=3$ in our proposed SAE model leads to enhanced precision of the PHIA estimates.

Figure \ref{fig:maps1}-(c) displays the spatial patterns of direct estimates of PHIA and (d) shows the cluster classification of the PHIA and MPI using  Algorithm \ref{alg:SC}.  As we expected, smaller (larger) PHIA direct estimates correspond to larger (smaller) MPI values. We should note that the external covariates $\boldsymbol{x}_k^{*}$ are not directly included in the SAE model as covariates. 


\subsection{Parameter estimation with Spectral Clustering}
\label{subsec:The_FHSC1}

Following the clustering structure illustrated in Figure \ref{fig:maps1}-(b) and Figure \ref{fig:maps1}-(d), we expect the small area parameters associated with PHIA direct estimates to vary smoothly for related small areas in a SAE model. As such, in order to find the small area parameters, we propose the use of machine learning techniques. The Regularized Task Kernel Learning (RTKL) criterion \citep{evgeniou2005learning, miller2023additive} is designed to balance two competing objectives in machine learning and statistical modeling: goodness-of-fit and smoothness/complexity of a learned function. Laplacian Regularized Least Squares (LapRLS) estimation \citep{evgeniou2000regularization, scholkopf2002learning, belkin2005manifold, belkin2006manifold} applies the RTKL criterion by enforcing regularization through a Laplacian matrix, ensuring smoothness across connected data points. Specifically, LapRLS solves an optimization problem where the objective function combines the Squared Loss Function (SLF) with a graph-based Laplacian regularization penalty.

First, consider the following  general linear mixed model definition. 

\begin{definition} \footnotesize
  \label{def:LMM} 
Suppose the direct estimates $\by$ are obtained using the linear mixed model,
\begin{align}
\label{eq:general_LMM}
\begin{split}
    \by & = \bdt + \boldsymbol{e}\\
 \bdt &= \boldsymbol{\mu} + \boldsymbol{Z}\bu,
\end{split}
\end{align}
where $\by$ is the $m \times 1$ vector of direct estimates, $\boldsymbol{Z}$ is a known
 full rank $m \times h$ design matrix, and $\bdt$ is the small area parameter vector. The vectors $\boldsymbol{\mu}$ and $\boldsymbol{e}$ are of dimension $m \times 1$, while $\bu$ is of dimension $h \times 1$. The vectors $\boldsymbol{e}$ and $\bu$ are independently distributed with means $\boldsymbol{0}$ and covariance matrices $\boldsymbol{D}$ and $\boldsymbol{G}_{\boldsymbol{\varphi}}$, respectively. The covariance matrix $\boldsymbol{G}_{\boldsymbol{\varphi}}$ depends on some variance parameters $\boldsymbol{\varphi}$.  
 \end{definition}
 
Our proposal, described in Proposition \ref{prop:LapRLS}, involves finding the minimizer of the objective function, $\bdt^{\M\textsf{-SC}}$, to obtain the small area parameters. Proposition \ref{prop:LapRLS} is motivated by the RTKL criterion and considers an objective function similar to that of LapRLS estimation. Specifically, the objective function includes the SLF, $(\bdt^{\M} -  \bdt)^\textsf{T}    (\bdt^{\M} -  \bdt) $,  
to minimize the error between the new  small area parameter vector $\bdt^{\M}$ and the original parameter $\bdt$, along with the regularization term, $(\bdt^{\M})^\textsf{T}L_{SC}(\bdt^{\M})$, that ensures smoothness of the parameters for related small areas using the Laplacian matrix $L_{SC}$. Inspired by previous work for additive multi-task learning \cite{evgeniou2005learning, miller2023additive}, we also incorporate a Cluster Regularization Penalty (CRP), denoted as $\rho \in (0,1]$. This penalty controls the tradeoff between the closeness of  $\bdt^{\M}$ to $\bdt$  and the desired smoothness induced by the Laplacian matrix. 

 \begin{prop}
 \label{prop:LapRLS}  \footnotesize
Consider the convex differentiable objective function given by
\begin{equation}
\label{eq:LapRLS2}
\bdt^{\M\textsf{-SC}} =    \underset{\bdt^{\M}}{\text{minimize}} \quad    \rho(\bdt^{\M} -  \bdt)^\textsf{T}    (\bdt^{\M} -  \bdt) +  (1-\rho)(\bdt^{\M})^\textsf{T}\textit{L}_{SC}(\bdt^{\M}), 
\end{equation}
where  $\rho \in (0,1]$ is the cluster regularization penalty,  $\bdt  = \boldsymbol{\mu} + \boldsymbol{Z}\bu$ as in model (\ref{eq:general_LMM}), $\textsf{L}_{SC}= \text{blkdiag}(\{\textsf{L}_{c}\}_{c=1}^{C})$  is the  Laplacian matrix obtained from Algorithm \ref{alg:SC}, with $\textsf{L}_{c}=n_c\mathbb I_{n_{c}} - \boldsymbol{1}_{n_c}\boldsymbol{1}_{n_c}^\textsf{T}$ and  $n_{c}$ the number of small areas in cluster $c$.  
The objective function (\ref{eq:LapRLS2}) leads to the following solution 
\begin{align}
\label{eq:LapRLS_2}
\bdt^{\M\textsf{-SC}}&=\boldsymbol{A}_{\rho}^{-1}\bdt, &  \boldsymbol{A}_{\rho}&=( \mathbb I_m +((1-\rho)/\rho)  \textit{L}_{SC}),
\end{align}
 where $\mathbb I_m$ denotes the identity matrix of order $m$ and $\boldsymbol{A}_{\rho} \in \mathbb{R}^{m \times m}$ is a symmetric and positive definite matrix. 
 \label{def:LapRLS_d} 
 \end{prop}
%

\begin{proof}  \small
The proof is in Appendix \ref{app:prop1}.
\end{proof}

\subsection{A new model with Spectral Clustering for SAE }
\label{subsec:The_FHSC}

In Definition \ref{def:FHSC_m1}, we introduce a new area-level model for SAE that incorporates the solution $\bdt^{\M\textsf{-SC}}$ obtained in Proposition \ref{prop:LapRLS} for model (\ref{eq:general_LMM}). 
 \begin{definition}[\footnotesize SAE model with Spectral Clustering ($\M$-SC)]  \footnotesize
 \label{def:FHSC_m1} 
The SAE  model $\M$ with Spectral Clustering ($\M$-SC) can be written as follows:
\begin{align}
\label{eq:model_FHSC_m}
\begin{split}
    \by & =  \bdt^{\textsf{$\M$-SC}} + \boldsymbol{e},\\
    \bdt^{\textsf{$\M$-SC}}  & =  \boldsymbol{A}_{\rho}^{-1}\bdt,\\
\bdt &= \boldsymbol{\mu} + \boldsymbol{Z}\bu,\\
\end{split}
\end{align}
 where $\boldsymbol{A}_{\rho} = ( \mathbb I_m + ((1-\rho)/\rho) \textsf{L}_{SC})$ is a symmetric and positive-definite matrix with $\boldsymbol{A}_{\rho} \in \mathbb{R}^{m \times m}$,  
 $\textsf{L}_{SC}= \text{blkdiag}(\{\textsf{L}_{c}\}_{c=1}^{C})$  is the  Laplacian matrix with $\textsf{L}_{c}=n_c\mathbb I_{n_{c}} - \boldsymbol{1}_{n_c}\boldsymbol{1}_{n_c}^\textsf{T}$, 
$\rho \in (0,1]$, and $\by$, $\bdt$, $\boldsymbol{Z}$, $\boldsymbol{\mu}$, $\boldsymbol{e}$ and $\bu$ are defined in Definition \ref{def:LMM}. The vectors $\boldsymbol{e}$ and $\bu$ are independently distributed with means $\boldsymbol{0}$ and covariance matrices $\boldsymbol{D}$ and $\boldsymbol{G}_{\boldsymbol{\varphi}}$. 
\end{definition}

Note that the $\M$-SC model in Definition \ref{def:FHSC_m1} is constructed without any distributional assumptions on $\boldsymbol{e}$ and $\bu$ and without imposing specific settings on the design matrix $\boldsymbol{Z}$ or the covariance matrices $\boldsymbol{D}$ and  $\boldsymbol{G}_{\boldsymbol{\varphi}}$. This flexible feature allows the incorporation of assumptions and settings of other existing area-level models in Definition  \ref{def:FHSC_m1}. Specifically, we incorporate the structure of the seminal Fay-Herriot model and define the Fay-Herriot model with Spectral Clustering (FH-SC). 



\begin{definition}[\footnotesize Fay-Herriot model with Spectral Clustering (FH-SC)] \footnotesize 
 \label{def:FHSC_m} 
The Fay-Herriot model with Spectral Clustering  (FH-SC) for cluster $c$, where $c=1,...,C$  and $C\leq m$ 
with  $j=1,...,n_c$ for $n_c$ denoting the number of small areas in the  $c$-th  cluster,
can be written as follows:
\begin{align}
\label{fay_herriot_s}
\begin{split}
    \by_c & =  \bdt^{\textsf{FH-SC}}_c + \boldsymbol{e}_c,\\
    \bdt^{\textsf{FH-SC}}_c  & =  A_{\rho,c}^{-1}\bdt_c,\\
\bdt_c &= \boldsymbol{X}_c \boldsymbol{\delta}_c  + \boldsymbol{Z}_c\bu_c,\\
\end{split}
\end{align}
where $ A_{\rho,c} = ( \mathbb I_{n_{c}} + ((1-\rho)/\rho) \textsf{L}_{c})$ is a symmetric and positive-definite matrix
with $ A_{\rho,c} \in \mathbb{R}^{n_c \times n_c}$, where $\textsf{L}_{c}=n_c\mathbb I_{n_{c}} - \boldsymbol{1}_{n_c}\boldsymbol{1}_{n_c}^\textsf{T}$ and $\rho \in (0,1]$.  In model (\ref{fay_herriot_s}), the sampling errors  $\boldsymbol{e}_c=(e_{1,c},...,e_{n_c,c})^\textsf{T}$ and the $h_c \times 1$ vector of random effects $\boldsymbol{u}_c$ are independent
with $\boldsymbol{e}_c \stackrel{ind}{\sim} \text{N}(\boldsymbol{0}, \boldsymbol{D}_c)$ and $\bu_c \stackrel{ind}{\sim} \text{N}(\boldsymbol{0},\boldsymbol{G}_{\boldsymbol{\varphi},c})$.
In addition, $\boldsymbol{X}_c$  and $\boldsymbol{Z}_c$ are known $n_c \times p$  and  $n_c \times h_c$ design matrices,
 $\boldsymbol{\delta}_c$ is a $p \times 1$ vector
of unknown regression coefficients 
 and  $\bdt_c=(\theta_{1,c},...,\theta_{n_c,c})^\textsf{T}$, $\by_c=(y_{1,c},...,y_{n_c,c})^\textsf{T}$,  and $\boldsymbol{D}_{c}=\text{diag}(D_{1,c},,...,D_{n_c,c})$ denote the small area parameter vector, direct estimates and direct variances  for cluster $c$, respectively.  
\end{definition}

Note that the cluster regularization penalty, $\rho$, plays a crucial role in the proposed FH-SC model. Smaller (larger) values of $\rho$ make the  cluster classification of external covariates more (less) important.

\subsubsection{Connections with existing models}
\label{subsec:connectFHSC}

 As discussed in the previous section, \cite{maiti2014clustering} and \cite{torkashvand2017clustering}  use
clustering of small areas to improve precision in SAE.  After clustering, these approaches implement general linear mixed models with clustering effects and use the Empirical Best Linear Unbiased Predictor (EBLUP) to estimate small area means. To evaluate the uncertainty of the EBLUPs, 
different approximations for the Mean-Squared Prediction Error (MSPE) estimate are developed in \cite{maiti2014clustering} and \cite{torkashvand2017clustering}. Since these models follow specifications similar to those of the FH model but incorporate a clustering structure, we refer to them in this work as  Fay-Herriot models with Clustering (FH-C). Crucially, as mentioned, the settings for the fixed and
 random effects in the FH and FH-C models can be incorporated into our proposed FH-SC model. To illustrate this, we refer to Table \ref{tab:models}, which presents the specific settings for the FH and FH-C models, along with three related versions of the proposed FH-SC model.  
 
Table \ref{tab:models} shows that FH and FH-C can be specified by setting $\rho=1$ in Definition \ref{def:FHSC_m}. The first version of the proposed FH-SC model, FH-SC$_1$, considers the settings of the multivariate version of the FH$_{(\boldsymbol{\beta}, \sigma^{2})}$ model with common variance, $\sigma^2$, for all random effects. The second version, FH-SC$_2$, uses the specifications of the FH-C$_{1}$ model proposed by \cite{torkashvand2017clustering}. Here, the variances of the random effects, $\sigma^{2}_c$, are cluster-specific while the fixed effects remain the same across clusters, $\boldsymbol{\delta}_c=\bb$. The last model, FH-SC$_3$, follows the specifications of the FH-C$_{2}$ model of \cite{maiti2014clustering}, where both the variances of the random effects and the fixed effects, $\boldsymbol{\delta}_c=\bb_c$, are cluster specific. In addition, similar to \cite{maiti2014clustering}, the model includes cluster-specific random effects $u_c$  with $u_c \sim N(0, \hat{\gamma}\sigma_c^{2})$ where $\hat{\gamma}$ is known and estimated by analysis of variance \cite{maiti2014clustering,booth2008clustering} and small area random effects $\nu_{i}$ with $\nu_{i} \sim N(0, \sigma_c^{2})$.

\begin{table} 
\begin{center}
\footnotesize
\setlength\extrarowheight{2mm}
\begin{tabular}{lcccccc}  \hline\hline
Model &   $\rho$ & $\boldsymbol{\delta}_c$   & $\boldsymbol{Z}_c$  & $\boldsymbol{G}_{\boldsymbol{\varphi},c}$  &  $\boldsymbol{u}_c$ &  Reference \\ \hline \hline
FH$_{(\boldsymbol{\beta}, \sigma^{2})}$ & 1 & $\boldsymbol{\beta}$ & $\mathbb I_{n_{c}}$  & $\sigma^{2}\mathbb I_{n_{c}}$   & $\boldsymbol{u}_c=(u_{1,c},\ldots,u_{n_c,c})^\textsf{T}$ & \citep{fay_1979}   \\
FH-C$_{1(\boldsymbol{\beta}, \sigma_c^{2})} $  & 1 & $\boldsymbol{\beta}$ & $\boldsymbol{1}_{n_c}$  & $\sigma_c^{2} \mathbb I_{n_{c}}$  &  $\boldsymbol{u}_c=(u_{c},\ldots,u_{c})^\textsf{T}$ &   \cite{torkashvand2017clustering}  \\
FH-C$_{2(\boldsymbol{\beta}_c, \sigma_c^{2}, \nu_{i})}$  & 1 & $\boldsymbol{\beta}_c$  & $[\boldsymbol{1}_{n_c}\; \mathbb I_{n_{c}}]$  & $\text{diag}_{n_c+1}(\hat{\gamma}, \mathbb I_{n_{c}})\sigma_c^{2}$  & $\boldsymbol{u}_c=(u_{c},v_{1,c},\ldots,v_{n_c,c})^\textsf{T}$ &  \cite{maiti2014clustering}  \\ 
FH-SC$_{1(\boldsymbol{\beta}, \sigma^{2}, \rho)}$  & $(0,1]$ & $\boldsymbol{\beta}$ & $\mathbb I_{n_{c}}$  & $\sigma^{2} \mathbb I_{n_{c}}$  &   $\boldsymbol{u}_c=(u_{1,c},\ldots,u_{n_c,c})^\textsf{T}$ &  Def. \ref{def:FHSC_m} \\
FH-SC$_{2(\boldsymbol{\beta}, \sigma_c^{2}, \rho)}$   & $(0,1]$ & $\boldsymbol{\beta}$ & $\boldsymbol{1}_{n_c}$  & $\sigma_c^{2} \mathbb I_{n_{c}}$   & $\boldsymbol{u}_c=(u_{c},\ldots,u_{c})^\textsf{T}$ &  Def. \ref{def:FHSC_m}  \\ 
FH-SC$_{3(\boldsymbol{\beta}_c, \sigma_c^{2}, \nu_{i}, \rho)}$   & (0,1] & $\boldsymbol{\beta}_c$ & $[\boldsymbol{1}_{n_c}\; \mathbb I_{n_{c}}]$  & $\text{diag}_{n_c+1}(\hat{\gamma}, \mathbb I_{n_{c}})\sigma_c^{2}$   &    $\boldsymbol{u}_c=(u_{c},v_{1,c},\ldots,v_{n_c,c})^\textsf{T}$  &  Def.  \ref{def:FHSC_m} \\
\hline\hline
\end{tabular}
\caption{\small  Settings for the FH, FH-C models and three different versions of the proposed FH-SC model according to Definition    \ref{def:FHSC_m}. The subscripts in the names of the models denote the settings of the fixed and random effects. Specifically, $\boldsymbol{\beta}_c$ and $\sigma_c^{2}$ denote fixed effects and variances of random effects changing across clusters,  $\nu_{i}$ represents small area random effects, and $\boldsymbol{\beta}$ and $\sigma^{2}$ denote common fixed effects and variances of the random effects across clusters. }
\label{tab:models}
\end{center}
\end{table} 

\subsubsection{Properties of small area parameters}
\label{subsec:properties}

Theorem \ref{expectation} shows expressions of the conditional posterior expectation and variance of the small area parameter vector
under the FH-SC model in (\ref{fay_herriot_s}). To demonstrate parts (i) and (ii) of Theorem \ref{expectation}, we take advantage of the hierarchical structure and assumptions of the proposed FH-SC model in Definition  \ref{def:FHSC_m}.  

\begin{thm} \footnotesize
\label{expectation}
Consider the FH-SC model in (\ref{fay_herriot_s}) and let  $\boldsymbol{\theta}^{\textsf{FH-SC}}_{c} = (\theta_{1,c}^{\textsf{FH-SC}},...,\theta_{n_{c},c}^{\textsf{FH-SC}})^\textsf{T}$   
be the small area parameter vector under the FH-SC model for cluster $c$ where $\boldsymbol{\theta}^{\textsf{FH-SC}}=(\boldsymbol{\theta}^{\textsf{FH-SC}}_{1},...,\boldsymbol{\theta}^{\textsf{FH-SC}}_{C})^\textsf{T}$.
 \begin{enumerate}[(i)]
 \item The conditional expectation of the posterior small area parameter vector under the FH-SC model for cluster $c$ can be written as follows,
\begin{align}
\label{smooth1}
E(\boldsymbol{\theta}^{\textsf{FH-SC}}_c \mid \boldsymbol{\delta}_c , \boldsymbol{G}_{\boldsymbol{\varphi},c}, \rho, \boldsymbol{Z}_{c}, \boldsymbol{X}_{c}, \by_c) &= \gamma_{c} E(\boldsymbol{\theta}_c\mid \boldsymbol{\delta}_c , \boldsymbol{G}_{\boldsymbol{\varphi},c}, \rho, \boldsymbol{Z}_{c}, \boldsymbol{X}_{c}, \by_c) \\
&\hspace{-0.5cm} 
  + (1-\gamma_{c}) \mathrm{vec}\Big( \sum_{j=1}^{n_c}E(\theta_{j,c}\mid \boldsymbol{\delta}_c , \boldsymbol{G}_{\boldsymbol{\varphi},c}, \rho, \boldsymbol{Z}_{c}, \boldsymbol{X}_{c}, \by_c)/n_c\Big), \notag 
\end{align}
with
\begin{align}
\label{smooth1_a}
E(\boldsymbol{\theta}_c  \mid \boldsymbol{\delta}_c,\boldsymbol{G}_{\boldsymbol{\varphi},c}, \rho, \boldsymbol{Z}_{c}, \boldsymbol{X}_{c}, \by_c) &=   V(\boldsymbol{\theta}_c\mid \boldsymbol{\delta}_c,\boldsymbol{G}_{\boldsymbol{\varphi},c}, \rho,\boldsymbol{Z}_{c}, \boldsymbol{X}_{c}, \by_c)\\  
& \hspace{1.0cm}\times (A_{\rho,c}^{-1} \boldsymbol{D}_c^{-1}  \by_c   +(\boldsymbol{Z}_{c}\boldsymbol{G}_{\boldsymbol{\varphi},c}\boldsymbol{Z}_{c}^\textsf{T} )^{-1}\boldsymbol{X}_c\boldsymbol{\delta}_c), \notag
\end{align}
\begin{align}
\label{smooth1_b}
V(\boldsymbol{\theta}_c\mid \boldsymbol{\delta}_c,\boldsymbol{G}_{\boldsymbol{\varphi},c}, \rho, \boldsymbol{Z}_{c}, \boldsymbol{X}_{c}, \by_c) &=((A_{\rho,c} \boldsymbol{D}_c  A_{\rho,c}^\textsf{T})^{-1} +  (\boldsymbol{Z}_{c}\boldsymbol{G}_{\boldsymbol{\varphi},c}\boldsymbol{Z}_{c}^\textsf{T} )^{-1} )^{-1},
\end{align}
where $A_{\rho,c}^{-1} = \gamma_{c}  \mathbb I_{n_{c}} + ((1-\gamma_{c})/n_c)\boldsymbol{1}_{n_c}\boldsymbol{1}_{n_c}^\textsf{T}$ and 
$ \gamma_{c}=  \rho /((1-\rho)n_{c} + \rho)$ with $ \gamma_{c} \in (0,1]$.


\item   The conditional variance of the posterior small area parameter vector under the FH-SC model for cluster $c$ is, 
\begin{align}
\label{var_smooth1_b}
V(\boldsymbol{\theta}^{\textsf{FH-SC}}_c \mid \boldsymbol{\delta}_c , \boldsymbol{G}_{\boldsymbol{\varphi},c}, \rho, \boldsymbol{Z}_{c}, \boldsymbol{X}_{c}, \by_c)=
\gamma_{c}^2  V(\boldsymbol{\theta}_c\mid \boldsymbol{\delta}_c,\boldsymbol{G}_{\boldsymbol{\varphi},c}, \rho, \boldsymbol{Z}_{c}, \boldsymbol{X}_{c}, \by_c) \\
&\hspace{-10.0cm}  +  (\gamma_{c}(1-\gamma_{c}))/n_c\left( \boldsymbol{1}_{n_c}\boldsymbol{1}_{n_c}^\textsf{T} V(\boldsymbol{\theta}_c\mid \boldsymbol{\delta}_c,\boldsymbol{G}_{\boldsymbol{\varphi},c}, \rho, \boldsymbol{Z}_{c}, \boldsymbol{X}_{c}, \by_c) + V(\boldsymbol{\theta}_c\mid \boldsymbol{\delta}_c,\boldsymbol{G}_{\boldsymbol{\varphi},c}, \rho, \boldsymbol{Z}_{c}, \boldsymbol{X}_{c}, \by_c) \boldsymbol{1}_{n_c}\boldsymbol{1}_{n_c}^\textsf{T} \right) \notag
\\
&\hspace{-7.5cm}  
+ ((1-\gamma_{c})/n_c)^{2}\left( \boldsymbol{1}_{n_c}\boldsymbol{1}_{n_c}^\textsf{T} V(\boldsymbol{\theta}_c\mid \boldsymbol{\delta}_c,\boldsymbol{G}_{\boldsymbol{\varphi},c}, \rho, \boldsymbol{Z}_{c}, \boldsymbol{X}_{c}, \by_c) \boldsymbol{1}_{n_c}\boldsymbol{1}_{n_c}^\textsf{T} \right)  \notag
\end{align}
\end{enumerate}
\end{thm}
\begin{proof}  \footnotesize
The proof is in Appendix \ref{app:prop}.
\end{proof}

In part (i) of Theorem \ref{expectation}, we observe that the posterior conditional expectation is a trade-off between the conditional expectation of the vector of posterior means and the average of  conditional expectations of posterior means for cluster $c$.  In parts (i) and (ii),  $\gamma_{c} \in (0,1]$ is the weight  for cluster $c$ and measures the importance of the spectral clustering classification. More formally, given the size of cluster $c$, $n_c$, if $\rho \to 1$ then $\gamma_{c} \to 1$. Consequently, the posterior conditional expectation and variance  converge to the conditional expectation and variance of $\boldsymbol{\theta}_c$. For instance, when spectral clustering with external covariates is not useful, the conditional expectation (or variance) under the FH-SC model converges to the conditional expectation (or variance) under the FH or FH-C models (see Table \ref{tab:models}). Conversely, if $\rho \to 0$ then $ \gamma_{c} \to 0$, such that $E(\boldsymbol{\theta}^{\textsf{FH-SC}}_c \mid \cdot)$ and $V(\boldsymbol{\theta}^{\textsf{FH-SC}}_c \mid \cdot)$ converge to the corresponding average of conditional expectation and variance under the FH or FH-C models.

\subsubsection{Prior specification and MCMC algorithms}
\label{subsec:The_MCMC}
We consider a joint prior distribution  $\pi(\boldsymbol{\delta}_c,\boldsymbol{G}_{\boldsymbol{\varphi},c},\rho_c) =  \prod_{c=1}^{C} \pi(\boldsymbol{G}_{\boldsymbol{\varphi},c})\pi(\boldsymbol{\delta}_c)\pi(\rho)$, where $\pi(\boldsymbol{G}_{\boldsymbol{\varphi},c})$ denotes the prior for the variance parameters  $\boldsymbol{\varphi}$ in the covariance matrix $\boldsymbol{G}_{\boldsymbol{\varphi},c}$, and $\pi(\boldsymbol{\delta}_c)$ and $\pi(\rho)$ the priors for $\boldsymbol{\delta}_c$ and $\rho$, respectively. Theorem \ref{prop2} establishes the conditions for posterior propriety under the FH-SC models.


\begin{thm}  \footnotesize
\label{prop2}
The posterior probability density $p(\boldsymbol{\theta}, \boldsymbol{\delta},\boldsymbol{G}_{\boldsymbol{\varphi}}, \rho \mid \boldsymbol{y}, \boldsymbol{X}, \boldsymbol{Z}, \boldsymbol{D})$ under the FH-SC$_{1}$, FH-SC$_{2}$  and
 FH-SC$_{3}$ models in Table \ref{tab:models} is proper if the priors
for the variance parameters, $\boldsymbol{\varphi}$, in the covariance matrix, $\boldsymbol{G}_{\boldsymbol{\varphi},c}$, and the prior for the cluster regularization penalty, $\rho$, are proper. 
\end{thm}

\begin{proof} \footnotesize
The proof is in Appendix \ref{app:prop3}.
\end{proof}

Consequently, we consider proper prior distributions for the variance parameters, $\boldsymbol{\varphi}$ in the covariance matrix, $\boldsymbol{G}_{\boldsymbol{\varphi},c}$ and for the cluster regularization penalty $\rho$. Details on the prior specification and resulting posterior distribution, $p(\boldsymbol{\theta}, \boldsymbol{\delta},\boldsymbol{G}_{\boldsymbol{\varphi}}, \rho \mid \boldsymbol{y}, \boldsymbol{X}, \boldsymbol{Z}, \boldsymbol{D})$, are provided in Supplementary Section \ref{prior_spec}. We propose Algorithms \ref{alg:MCMC1} and \ref{alg:MCMC2} in Section \ref{MCM_algh1} of the supplement to generate posterior samples of the model parameters $\boldsymbol{\kappa}=(\boldsymbol{\theta}, \boldsymbol{\delta},\boldsymbol{G}_{\boldsymbol{\varphi}}, \rho )$ for the FH-SC models. 

\section{Estimation and Uncertainty Quantification}
\label{sec:all_estimation_methods}

In this section, we refer to the FH-SC model in a general sense, as the methods discussed also apply to the FH and FH-C models. To obtain small area estimates of PHIA, we can compute the following ergodic average using the posterior samples of $ \theta_{j,c}^{\textsf{FH-SC}(l)}$ for the $j$-th small area in cluster $c$, $c=1,...,C$ (provided by Algorithm \ref{alg:MCMC1}):
\begin{equation}\small
\label{ergodic_mean} 
\bar{\hat{\theta}}_{j,c}^{\textsf{FH-SC}} =\dfrac{1}{(L-T)} \sum_{l=T+1}^{L} \theta_{j,c}^{\textsf{FH-SC}(l)}, 
\end{equation}
where $\bar{\hat{\boldsymbol{\theta}}}^{\textsf{FH-SC}}=(\bar{\hat{\boldsymbol{\theta}}}^{\textsf{FH-SC}}_1,...,\bar{\hat{\boldsymbol{\theta}}}^{\textsf{FH-SC}}_C)^{T}$ denotes the vector of small area estimates of PHIA with $\bar{\hat{\boldsymbol{\theta}}}^{\textsf{FH-SC}}_c=(\bar{\hat{\theta}}_{1,c}^{\textsf{FH-SC}},...,\bar{\hat{\theta}}_{n_c,c}^{\textsf{FH-SC}})^{T}$, and $L$ and $T$ are the total number and the number of discarded MCMC samples, respectively.
However, small area estimators with lower variance can be obtained using a Rao-Blackwell (RB) argument under the conditional posterior mean \citep{Rao1945, blackwell1947conditional,rao2015small}. Note that for the posterior variance it holds that,
\begin{align}\small
\label{var}
V_{\theta_{j,c}^{\textsf{FH-SC}}}(\theta_{j,c}^{\textsf{FH-SC}} \mid  \boldsymbol{X}_{c}, \boldsymbol{Z}_{c}, \boldsymbol{y}_{c}) \geq  V_{\theta_{j,c}^{\textsf{FH-SC}}}(E_{\vartheta_{-\theta_{j,c}}^{\textsf{FH-SC}}}(\theta_{j,c}^{\textsf{FH-SC}}\mid   \boldsymbol{X}_{c}, \boldsymbol{Z}_{c},  \boldsymbol{y}_{c},  \vartheta_{-\theta_{j,c}}^{\textsf{FH-SC}})),
\end{align}
 where $\vartheta_{-\theta_{j,c}}^{\textsf{FH-SC}}$ is the set of parameters under the FH-SC model excluding the parameter $\theta_{j,c}$.
Importantly, as noted by \cite{rao2015small}, RB estimators require closed form expressions of the conditional posterior mean for computation. 
Since such expressions are available for all existing and proposed models in Table \ref{tab:models}, we can compute small area RB estimators of PHIA under the FH, FH-S, and FH-SC models using Definition \ref{def_RB}.
 \begin{definition}[\footnotesize Rao-Blackwell (RB) estimator]  \footnotesize
\label{def_RB}
The Rao-Blackwell (RB) estimator of the small area parameter vector under FH-SC model,  $\boldsymbol{\theta}^{\textsf{FH-SC}}$, 
 is given by
 \begin{align}
   \label{eq:Rao_Blackwell}
    \hat{\boldsymbol{\theta}}^{\textsf{FH-SC}}&= E_{\vartheta_{-\boldsymbol{\theta}}^{\textsf{FH-SC}}}(E(\boldsymbol{\theta}^{\textsf{FH-SC}}\mid  \boldsymbol{X},
     \boldsymbol{Z},  \boldsymbol{y}, \vartheta_{-\boldsymbol{\theta}}^{\textsf{FH-SC}})
) \\ & \approx \frac{1}{L-T} \sum_{\ell=T+1}^{L} E(\boldsymbol{\theta}^{\textsf{FH-SC}(l)} \mid \boldsymbol{X},
     \boldsymbol{Z},  \boldsymbol{y},\vartheta_{-\boldsymbol{\theta^{(l)}}}^{\textsf{FH-SC}(l)}), \notag
\end{align}
where $E(\boldsymbol{\theta}^{\textsf{FH-SC}(l)} \mid \boldsymbol{X},
     \boldsymbol{Z},  \boldsymbol{y},\vartheta_{-\boldsymbol{\theta^{(l)}}}^{\textsf{FH-SC}(l)})$ is the conditional expectation  of the small area parameter vector for the  $l$-th draw computed with the output of Algorithm \ref{alg:MCMC1} for $c=1,...,C$,  and $\vartheta_{-\boldsymbol{\theta}}^{\textsf{FH-SC}(l)}=(\boldsymbol{\delta}^{(l-1)},\boldsymbol{G}_{\boldsymbol{\varphi}^{(l-1)}}, \rho^{(l)})$ is the set of parameters for the $l$-th draw  excluding the small area parameter vector $\boldsymbol{\theta}^{(l)}=(\boldsymbol{\theta}_{1}^{(l)},...,\boldsymbol{\theta}_{C}^{(l)})^\textsf{T}$.
$L$  and $T$ are the total number and the number of discarded MCMC samples, respectively.
\end{definition}

\subsection{Benchmarking estimation using posterior projections}
\label{sec:pp_full}

The process of adjusting estimates while imposing constraints is referred to as benchmarking. Specifically, let $k$ be the number of linear equality constraints $\sum_{i=1}^m w_{j^{'}i} \hat{\theta}_i = p_i$ for $1 \leq j^{'} \leq k$. In matrix notation this can be written as $\boldsymbol{W}\hat{\boldsymbol{\theta}} = \boldsymbol{p}$, where $w_{j^{'}i}$ is the $(j^{'}, i)$-th element of $\boldsymbol{W} \in \mathbb{R}^{k \times m}$. Without loss of generality, we assume $k \leq m$ and the values of $\boldsymbol{W}$ and $\boldsymbol{p}$ are assumed to be obtained from an external data source (e.g., administrative data, survey, population census). However, it is possible for $\boldsymbol{W}$ and $\boldsymbol{p}$ to originate from the same survey, as is described in \cite{bell_2013}. In Colombia, official reports of PHIA at the national level are produced by the NSO \citep{internet}. Specifically, the PHIA was estimated to be 0.418 at the national level in 2015, and we will use this value as our external benchmark.

To enable benchmarking for the FH-SC model, we introduce Proposition \ref{prop:benchmarking}, where external linear benchmarking constraints of the form  $ \boldsymbol{W} \bdt^{\M} = \boldsymbol{p}$ are incorporated into the objective function (\ref{eq:LapRLS2}) of Proposition  \ref{prop:LapRLS}. In Proposition \ref{prop:benchmarking}, we derive the solution to the modified objective function (\ref{eq:LapRLS3}), which defines the small area benchmarked parameter vector under model $\M$. Specifically, since the Laplacian matrix $\textsf{L}_{SC}$ is symmetric and  positive semi-definite \citep{von2007tutorial}, and  $\boldsymbol{W}$ has full row rank $k \leq m$, we can construct a linear system using the Karush–Kuhn–Tucker (KKT)  conditions  \citep{karush1939minima, kuhn1951nonlinear} to solve for $\bdt^{\textsf{$\M$-SC-B}}$.

 \begin{prop}
\footnotesize
 \label{prop:benchmarking} 
Consider the convex differentiable objective function  given by,
\begin{align}
\label{eq:LapRLS3}
\bdt^{\M\textsf{-SC-B}} &=    \underset{\bdt^{\M}}{\text{minimize}} \quad    \rho(\bdt^{\M} -  \bdt)^\textsf{T}    (\bdt^{\M} -  \bdt) +  (1-\rho)(\bdt^{\M})^\textsf{T}L_{SC}(\bdt^{\M}), \\ &\hspace{2cm} \text{subject to} \quad \boldsymbol{W}\bdt^{\M} = \boldsymbol{p}. \notag
\end{align}
where $\boldsymbol{W} \in \mathbb{R}^{k \times m}$  has full row rank $k \leq m$ and  $\rho$,  $\bdt^{\M}$, $\bdt$ and $\textsf{L}_{SC}$ are defined as in Proposition \ref{prop:LapRLS}.
Under the Karush–Kuhn–Tucker (KKT) conditions the objective function (\ref{eq:LapRLS3}) leads to the following solution:
\begin{align}
\label{benchmark_PP_solution}
\bdt^{\M\textsf{-SC-B}}  &= \boldsymbol{\theta}^{\M\textsf{-SC}}  + \boldsymbol{A}_{\rho}^{-1}  \boldsymbol{W}^T ( \boldsymbol{W} \boldsymbol{A}_{\rho}^{-1}  \boldsymbol{W}^T)^{-1} (\boldsymbol{p} -  \boldsymbol{W} \boldsymbol{\theta}^{\M\textsf{-SC}} ).   
\end{align}

\end{prop}

\begin{proof}\footnotesize
The proof is in Appendix \ref{app:prop4}.
\end{proof}

In practice, benchmarked estimates of PHIA are useful for ensuring comparability of PHIA estimates at higher levels of aggregation. These small area benchmarked estimates are obtained by replacing the small area parameter with its ergodic average estimate in the benchmarking equation. 
Specifically, we substitute the estimator based on the ergodic average, $\bar{\hat{\boldsymbol{\theta}}}^{\textsf{FH-SC}}$, from (\ref{ergodic_mean}) into the benchmarked solution in (\ref{benchmark_PP_solution}), as follows
\begin{align} \small
\label{benchmark_PP_est}
\bar{\hat{\bdt}}^{\textsf{FH-SC-B}}  = \bar{\hat{\boldsymbol{\theta}}}^{\textsf{FH-SC}}  + \boldsymbol{A}_{\hat{\rho}}^{-1}  \boldsymbol{W}^T ( \boldsymbol{W} \boldsymbol{A}_{\hat{\rho}}^{-1}  \boldsymbol{W}^T)^{-1} (\boldsymbol{p} - \bar{\hat{\boldsymbol{\theta}}}^{\textsf{FH-SC}}),   
\end{align}
where  $\bar{\hat{\bdt}}^{\textsf{FH-SC-B}} = (\bar{\hat{\bdt}}^{\textsf{FH-SC-B}}_1,...,
\bar{\hat{\bdt}}^{\textsf{FH-SC-B}}_C)^{T}$ is the vector benchmarked estimates with  $\bar{\hat{\bdt}}^{\textsf{FH-SC-B}}_c=(\bar{\hat{\theta}}^{\textsf{FH-SC-B}}_{1,c},...,\bar{\hat{\theta}}^{\textsf{FH-SC-B}}_{n_c,c})^{T}$, and $\hat{\rho}$ is a suitable estimator of $\rho$. However, credible intervals, quantiles or posterior probabilities of benchmarked estimates of PHIA represent a more rich source of information, particularly for policy making in NSOs. Crucially, to compute RB benchmarked estimates of PHIA and their uncertainty using Definition \ref{def_RB}, we require posterior samples of the posterior benchmarked distribution under the FH, FH-C and FH-SC models. To achieve this, we leverage the theory of posterior projections \citep{dunson2003bayesian2, patra2024constrained, patra2019constrained}  to construct a posterior distribution tailored for benchmarking estimation.

The theoretical framework of posterior projections was originally proposed by \citep{dunson2003bayesian2}, extended to arbitrary problems by \cite{patra2024constrained}  and adapted to the context of SAE area-level model in \cite{patra2019constrained}. The theory of projecting samples from the MCMC output into a feasible set to induce a posterior distribution whose support respects certain constraints is useful for obtaining a full posterior of the benchmarked estimator. Following \citep{patra2019constrained}, to derive a full posterior for benchmarking estimation under model $\M$, we need to project the conditional posterior MCMC samples, $\boldsymbol{\theta}^{\M\textsf{-SC}(l)}  \in \Theta^{\textsf{$\M$}}$ for $l=1,...,L$, onto a constrained parameter space $\widetilde{\Theta}^{\textsf{$\M$}} $ where $\Theta^{\textsf{$\M$}} \subset \widetilde{\Theta}^{\textsf{$\M$}}$, through a minimal distance mapping. 

Theorem \ref{thm:PP_problem} in the supplement provides three main results. Part (i) shows that the objective function (\ref{eq:LapRLS3}) in Proposition \ref{prop:benchmarking} leads to a projection problem. Part (ii) provides the following solution for the projection problem:
\begin{align} \small
\label{benchmark_PP}
 \bdt^{\textsf{FH-SC-B(l)}}   &= \bdt^{\textsf{FH-SC(l)}}  + \boldsymbol{A}_{\rho^{(l)}}^{-1} \boldsymbol{W}^T (\boldsymbol{W} \boldsymbol{A}_{\rho^{(l)}}^{-1} \boldsymbol{W}^T)^{-1} (\boldsymbol{p} - \boldsymbol{W}  \bdt^{\textsf{FH-SC(l)}} ),   
\end{align}
where $l=1,...,L$  are posterior samples under the FH-SC model. Notably, this result ensures the availability of posterior samples from the benchmarked posterior distribution. Finally, Part (iii) provides the conditional expectation of the small area benchmarked parameter vector, $\bdt^{\textsf{FH-SC-B}}$, in closed form:
\begin{align} \small
 \label{benchmark}
    E(\bdt^{\textsf{FH-SC-B(l)}}  \mid \boldsymbol{X},
     \boldsymbol{Z},  \boldsymbol{y},\vartheta_{-\boldsymbol{\theta^{(l)}}}^{\textsf{FH-SC}(l)})&=  E(\bdt^{\textsf{FH-SC(l)}}  \mid \boldsymbol{X},
     \boldsymbol{Z},  \boldsymbol{y},\vartheta_{-\boldsymbol{\theta^{(l)}}}^{\textsf{FH-SC}(l)}) \\  & \hspace{-2.5cm} + \boldsymbol{A}_{\rho^{(l)}}^{-1} \boldsymbol{W}^T (\boldsymbol{W} \boldsymbol{A}_{\rho^{(l)}}^{-1}\boldsymbol{W}^T)^{-1}(\boldsymbol{p} - W E(\bdt^{\textsf{FH-SC(l)}}  \mid \boldsymbol{X},
     \boldsymbol{Z},  \boldsymbol{y},\vartheta_{-\boldsymbol{\theta^{(l)}}}^{\textsf{FH-SC}(l)}) ), \notag
\end{align} 
where  $E(\bdt^{\textsf{FH-SC(l)}}  \mid \boldsymbol{X},
     \boldsymbol{Z},  \boldsymbol{y},\vartheta_{-\boldsymbol{\theta^{(l)}}}^{\textsf{FH-SC}(l)})$ is the conditional posterior expectation for the $l$-th posterior sample given in equation (\ref{smooth1}) of Theorem \ref{expectation}, and $\vartheta_{-\boldsymbol{\theta}}^{\textsf{FH-SC}(l)}=(\boldsymbol{\delta}^{(l-1)},\boldsymbol{G}_{\boldsymbol{\varphi}^{(l-1)}}, \rho^{(l)})$. 
To clarify our contribution, the objective function in \citep{patra2019constrained} differs from the objective function (\ref{eq:LapRLS3}) in our Proposition \ref{prop:benchmarking}. Consequently, part (i) of Theorem \ref{thm:PP_problem} is distinct but complementary to Lemma 2 in \citep{patra2019constrained}. However, to the best of our knowledge, parts (ii) and (iii) of Theorem \ref{thm:PP_problem} introduce completely new results. 

An estimator of the benchmarked parameters under the FH-SC model can be computed using the 
posterior projected samples in (\ref{benchmark_PP}), as follows:
\begin{equation} \small
\label{eq:Erg_bench_estimator}
\bar{\hat{\bdt}}^{\textsf{FH-SC-B}} =\dfrac{1}{L-T}\sum_{l=T+1}^{L}  \bdt^{\textsf{FH-SC-B(l)}}.
\end{equation}
However, instead of using the ergodic average of the MCMC samples in (\ref{eq:Erg_bench_estimator}), we propose the use of RB benchmarked estimators to produce small area benchmarked estimates of PHIA. To this end, in Definition \ref{def_RB_becnhmarked}, we use the closed form expression of the conditional expectation in (\ref{benchmark}).
\begin{definition}[\footnotesize Rao-Blackwell (RB) benchmarked estimator]  \footnotesize
\label{def_RB_becnhmarked}
The Rao-Blackwell (RB) benchmarked estimator of the small area benchmarked parameter vector under FH-SC model,  $\boldsymbol{\theta}^{\textsf{FH-SC-B}}$,  is given by
 \begin{align}
   \label{eq:Rao_Blackwell_benchmarked}
    \hat{\boldsymbol{\theta}}^{\textsf{FH-SC-B}}&= E_{\vartheta_{-\boldsymbol{\theta}}^{\textsf{FH-SC}}}(E(\boldsymbol{\theta}^{\textsf{FH-SC-B}}\mid  \boldsymbol{X},
     \boldsymbol{Z},  \boldsymbol{y}, \vartheta_{-\boldsymbol{\theta}}^{\textsf{FH-SC}})      
) \\ & 
\approx \frac{1}{L-T} \sum_{\ell=T+1}^{L} E(\bdt^{\textsf{FH-SC-B(l)}}  \mid \boldsymbol{X},
     \boldsymbol{Z},  \boldsymbol{y},\vartheta_{-\boldsymbol{\theta^{(l)}}}^{\textsf{FH-SC}(l)}), \notag
\end{align}
where $E(\bdt^{\textsf{FH-SC-B(l)}}  \mid \boldsymbol{X}, \boldsymbol{Z},  \boldsymbol{y},\vartheta_{-\boldsymbol{\theta^{(l)}}}^{\textsf{FH-SC}(l)})$ is the conditional expectation of the small area benchmarked parameter vector defined in (\ref{benchmark}), and $\vartheta_{-\boldsymbol{\theta}}^{\textsf{FH-SC}(l)}$, $L$ and $T$ are as defined in Definition \ref{def_RB}.
\end{definition}

\subsection{Uncertainty quantification for benchmarked estimators}
\label{sec:PMSE}

Benchmarked estimates of PHIA are crucial for production of official statistics at different levels of aggregation. Unquestionably, quantification of the  uncertainty associated with these estimates is imperative for sensible policy making. Significant efforts have been made in the SAE literature to measure the uncertainty of benchmarked estimators within a Bayesian framework. For instance,
\cite{erciulescu2018benchmarking} considers the variance of the benchmarked posterior samples to measure the uncertainty of the benchmarked estimator. Meanwhile, \cite{you2002benchmarking} proposes the Posterior Mean Square
Error (PMSE) as an approximation of the MSE for benchmarked estimators.

In the FH-SC model, for cluster $c$ and $j = 1, ..., n_{c}$, the PMSE of the benchmarked estimator proposed by \cite{you2002benchmarking} can be written as follows: 
\begin{align}
\label{eq:PMSE}
\text{PMSE}(\bar{\hat{\theta}}^{\textsf{FH-SC-B}}_{j,c} \mid \boldsymbol{X}_c,
     \boldsymbol{Z}_c, \boldsymbol{y}_{c}) &= E_{\theta_{j,c}^{\textsf{FH-SC}}}((\bar{\hat{\theta}}^{\textsf{FH-SC-B}}_{j,c}- \theta_{j,c}^{\textsf{FH-SC}} )^{2}\mid \boldsymbol{X}_c,
     \boldsymbol{Z}_c, \boldsymbol{y}_{c})\\
&= (\bar{\hat{\theta}}^{\textsf{FH-SC-B}}_{j,c} - \bar{\hat{\theta}}_{j,c}^{\textsf{FH-SC}})^2 + V_{\theta_{j,c}^{\textsf{FH-SC}}}(\theta_{j,c}^{\textsf{FH-SC}} \mid \boldsymbol{X}_c,
     \boldsymbol{Z}_c, \boldsymbol{y}_{c}), \notag
\end{align}
where  
 $\bar{\hat{\theta}}_{j,c}^{\textsf{FH-SC}}$ and  $\bar{\hat{\theta}}^{\textsf{FH-SC-B}}_{j,c}$ are the regular 
and benchmarked small area estimators of $\theta_{j,c}^{\textsf{FH-SC-B}}$ in (\ref{ergodic_mean}) and (\ref{benchmark_PP_est}), 
 and $V_{\theta_{j,c}^{\textsf{FH-SC}}}(\theta_{j,c}^{\textsf{FH-SC}} \mid \boldsymbol{X}_c,
     \boldsymbol{Z}_c, \boldsymbol{y}_{c})$ is the PMSE of $\hat{\theta}_{j,c}^{\textsf{FH-SC}}$. Importantly, to find the PMSE in (\ref{eq:PMSE}), it is assumed that $\bar{\hat{\theta}}_{j,c}^{\textsf{FH-SC}}=E(\theta_{j,c}^{\M} \mid  \boldsymbol{X}_c,
     \boldsymbol{Z}_c, \boldsymbol{y}_{c})$. The proof for deriving the PMSE is provided in the Appendix of \cite{you2002benchmarking}. Additionally, the PMSE as an approximation of the MSE for benchmarked estimators is discussed in \cite{bell_2013}.  As noted by \cite{you2002benchmarking},  the PMSE of the benchmarked estimator increases the posterior mean square error relative to the conditional posterior variance $V_{\theta_{j,c}^{\textsf{FH-SC}}}(\theta_{j,c}^{\textsf{FH-SC}} \mid \boldsymbol{X}_c,
     \boldsymbol{Z}_c, \boldsymbol{y}_{c})$.  As previously discussed, according to equation (\ref{var}), the RB estimator may produce small area estimates with reduced conditional posterior variances compared to those obtained using the ergodic average estimator. This motivates the use of the RB benchmarked estimator in Definition \ref{def_RB_becnhmarked} for estimating PHIA at the municipality level.

 \subsubsection{Conditional Posterior Mean Squared Error for benchmarked estimators }
\label{sec:CPMSE}

In this section, we introduce a new measure for uncertainty quantification of small area benchmarked estimates, called
the \textit{Conditional Posterior Mean Square Error (CPMSE)}. Using a RB argument \citep{Rao1945, blackwell1947conditional} within the PMSE formulation, we define the CPMSE in Proposition \ref{prop:def3}. This measure can be readily applied to other SAE estimators obtained under a Bayesian framework and can be computed using the output of Supplementary Algorithms \ref{alg:MCMC1} and  \ref{alg:MCMC2}.

 \begin{prop}
  \footnotesize
\label{prop:def3}
The Conditional Posterior Mean Square Error (CPMSE) for the RB benchmarked estimator $\hat{\theta}_{j,c}^{\textsf{FH-SC-B}}$  for $c = 1, ..., C$ and $j = 1, ..., n_{c}$ 
is given by
\begin{align}
   \label{CPMSE_projection_estimator}
   \emph{CPMSE}(\hat{\theta}_{j,c}^{\textsf{FH-SC-B}}\mid \boldsymbol{X}_c,
     \boldsymbol{Z}_c, \boldsymbol{y}_{c}) &=  E_{\vartheta^{\textsf{FH-SC}}_{-\theta_{j,c}}}(E_{\theta_{j,c}^{\textsf{FH-SC}}}((\hat{\theta}_{j,c}^{\textsf{FH-SC-B}}- \theta_{j,c}^{\textsf{FH-SC}})^{2}\mid \boldsymbol{X}_c,
     \boldsymbol{Z}_c, \boldsymbol{y}_{c},\vartheta^{\textsf{FH-SC}}_{-\theta_{j,c}}))    \notag \\ \notag
     & \hspace{-3.5cm}=(\hat{\theta}_{j,c}^{\textsf{FH-SC-B}}-\hat{\theta}_{j,c}^{\textsf{FH-SC}})^{2} + \emph{CPMSE}(\hat{\theta}_{j,c}^{\textsf{FH-SC}}\mid \boldsymbol{X}_c,
     \boldsymbol{Z}_c, \boldsymbol{y}_{c}) \\   \notag
    &  \hspace{-3.5cm} \approx (\hat{\theta}_{j,c}^{\textsf{FH-SC-B}}-\hat{\theta}_{j,c}^{\textsf{FH-SC}})^{2}  + \dfrac{1}{L-T}\sum_{l=T+1}^{L-T} 
       \emph{V}(\theta_{j,c}^{\textsf{FH-SC}(l)}\mid \boldsymbol{X}_c,
     \boldsymbol{Z}_c, \boldsymbol{y}_{c},\vartheta^{\textsf{FH-SC}(l)}_{-\theta_{j,c}})
   \\   & \hspace{-0.5cm}   + 
 \dfrac{1}{L-T}\sum_{l=T+1}^{L-T} ( \emph{E}(\theta_{j,c}^{\textsf{FH-SC}(l)}\mid \boldsymbol{X}_c,
     \boldsymbol{Z}_c, \boldsymbol{y}_{c},\vartheta^{\textsf{FH-SC}(l)}_{-\theta_{j,c}})-\hat{\theta}_{j,c}^{\textsf{FH-SC}})^2,  
\end{align}
where $ \emph{E}(\theta_{j,c}^{\textsf{FH-SC}(l)}\mid \boldsymbol{X}_c,
     \boldsymbol{Z}_c, \boldsymbol{y}_{c},\vartheta^{\textsf{FH-SC}(l)}_{-\theta_{j,c}})$ and $ \emph{V}(\theta_{j,c}^{\textsf{FH-SC}(l)}\mid \boldsymbol{X}_c,
     \boldsymbol{Z}_c, \boldsymbol{y}_{c},\vartheta^{\textsf{FH-SC}(l)}_{-\theta_{j,c}})$ are as in  Theorem \ref{expectation} 
     with $\vartheta_{-\boldsymbol{\theta}}^{\textsf{FH-SC}(l)}=(\boldsymbol{\delta}^{(l-1)},\boldsymbol{G}_{\boldsymbol{\varphi}^{(l-1)}}, \rho^{(l)})$ and $L$  and $T$ are as in Definition \ref{def_RB}. 
 \end{prop}
\begin{proof} \footnotesize
The proof is in Appendix \ref{app:prop5}.
\end{proof}

Our proposed CPMSE in Proposition \ref{prop:def3} shares some similarities with frequentist approaches in terms of motivation and development. For instance, the Conditional Mean Squared Error of Prediction (CMSEP) in \citep{booth1998standard} is useful for measuring the variance of predictions in small domains under Generalized Linear Mixed Models, while the Conditional Mean Squared Error (CMSE) in \citep{rivest2000conditional} serves to evaluate the accuracy of small area estimators. Both the CMSEP and CMSE are constructed using conditional expectations within an Unconditional Mean Squared Error (UMSE) framework. However, in CMSEP, the expectation is conditional on the model parameters, whereas in CMSE, the expectation is conditional on the distribution of the data. Similar to the approach used in \cite{booth1998standard} for computing CMSEP, we implement an RB argument \citep{Rao1945, blackwell1947conditional} to derive the CPMSE in Proposition \ref{prop:def3}. Additionally, in computing the CPMSE, we consider the first and second moments of the conditional distribution of $\theta_{j,c}^{\textsf{FH-SC}}$. Therefore, the same principle of incorporating uncertainty through a distribution is also adopted in our proposed CPMSE.

\subsection{Model selection and practical guidance}
\label{sec:model_criteria}

 For model selection and practical implementation, our proposal follows a general three-step process. In the first 
step, small areas are clustered according to the output of Algorithm \ref{alg:SC}, incorporating external covariates. As discussed, to select the number of covariates to be used in Algorithm \ref{alg:SC}, the total within-cluster sums of squares is evaluated for different values of the number of clusters and external covariate combinations.  After defining the clusters, the second step involves producing small area estimates using the existing and proposed models in Table \ref{tab:models}. To this end, Algorithms \ref{alg:MCMC1} and \ref{alg:MCMC2} are used to compute small area RB estimates and benchmarked estimates of PHIA (Definitions \ref{def_RB} and \ref{def_RB_becnhmarked}), along with their associated uncertainty using the estimator of the CPMSE in (\ref{CPMSE_projection_estimator}). Consequently, the last step focuses on evaluating the different models for the purpose of model selection.

Among the different model selection criteria, we consider 
the Expected Predictive Deviance (EPD) proposed in \cite{rao2015small} for model comparison. 
The EPD utilizes samples from the posterior predictive distribution. Typically, model 
 comparison methods, such as 
 EPD, are not implemented in SAE when benchmarking is required. However, the posterior projection theory discussed in this work enables the generation of samples from the benchmarked posterior distribution using equation (\ref{benchmark_PP}).  Moreover, posterior predictive benchmarked samples are also available. For each draw $\theta_{j,c}^{\textsf{FH-SC-B}(l)}$ in (\ref{benchmark_PP}), we can generate a draw, $\tilde{y}_{j,c}^{\textsf{FH-SC-B}(l)}$, from the posterior predictive benchmarking distribution. Therefore, we recommend 
 EPD to evaluate the models in Table \ref{tab:models}. Details for 
 EPD can be found in Supplementary Section \ref{sub_dic}.

\section{Case Study: Estimating Internet Connectivity in the Municipalities of Colombia}
\label{sec:application}

In this section, we present the estimation results of PHIA using our proposed methodology. According to the 2015 National Quality of Life Survey (NQLS) in Colombia \citep{internet}, the proportion of homes with internet access (PHIA) was estimated at 0.418 at the national level. Benchmarking constraints require that the weighted sum of PHIA estimates at the municipality level aligns with the national estimate provided by the NQLS. More formally, the benchmarking constraints assume that $\boldsymbol{w}^T \hat{\bdt}=0.418$,
where $ \boldsymbol{w}$ are the sample weights. 

First, as discussed in Section \ref{sec:FHSC}, we use the estimators proposed in \cite{hajek1971comment} to compute the direct estimates $y_{i}$ and $D_{i}$, and the GVF function \citep{wolter1985introduction} to smooth the direct variances, $D_{i}$. In the GVF method, we explore two different models and select the most appropriate by performing a sensitivity analysis. Details of the sensitivity analysis are provided in Supplementary Section \ref{sub_applied2}. We consider the existing FH and FH-C models and the proposed FH-SC models in Table \ref{tab:models}. The assumptions on the error and random effects are provided in Supplementary Section \ref{sub_applied3}. The two covariates 
are a vector of ones and the index of illiteracy obtained from the 2014 Census of Agriculture \citep{CNA2014}.  As discussed in Section \ref{sec:FHSC} and illustrated in Figure \ref{fig:maps1}, the matrix $\boldsymbol{A}_{\rho}$ is obtained using Supplementary Algorithm \ref{alg:SC} with the direct estimates of PHIA and the MPI as external covariate. We implement Supplementary Algorithms \ref{alg:MCMC1} and \ref{alg:MCMC2} to obtain posterior samples under the various models. Specifically, for each Algorithm, we simulate two chains with $L=50000$ values, discard $T=10000$ and thin the chains by taking one out of every 4 sampled values. Supplementary Section \ref{sub_applied4} is dedicated to discuss the convergence of  Algorithms \ref{alg:MCMC1} and  \ref{alg:MCMC2}. In addition to this case study, Supplementary Section \ref{sec:simula} presents a simulation study. 

Table \ref{tab:model_criteria} presents the 
EPD values for the different models applied to PHIA. We observe that the FH-SC$_{2}$ model exhibits the lowest 
EPD values. Therefore, the FH-SC model with cluster-specific variance of the random effects is the most suitable for generating both small area RB and RB benchmarked estimates of PHIA. However, it is important to note that the best model for RB estimates and RB benchmarked estimates is not always the same. Given that benchmarking is a priority in many practical applications, the most appropriate benchmarking model should be selected in such cases.

\begin{table}[t] 
\begin{center}
\footnotesize
\setlength\extrarowheight{1.5mm}
\begin{tabular}{lclccc}
  \hline
Model  &   Benchmarking   & Estimator    & $\text{EPD}^{\text{ASD}}$  & $\text{EPD}^{\text{AAD}}$ 
\\ 
  \hline
FH$_{(\boldsymbol{\beta}, \sigma^{2})}$ & No & $\hat{\theta}_{j,c}^{\textsf{FH}}$   & 1.16 & 8.49 
\\ 
FH-C$_{1(\boldsymbol{\beta}, \sigma_c^{2})}$  & No & $\hat{\theta}_{j,c}^{\textsf{FH-C}_{1}}$  & 1.20 & 8.57 
\\ 
FH-C$_{2(\boldsymbol{\beta}_c, \sigma_c^{2}, \nu_{i})}$  & No & $\hat{\theta}_{j,c}^{\textsf{FH-C}_{2}}$  & 1.13 & 8.38 
 \\ 
FH-SC$_{1(\boldsymbol{\beta}, \sigma^{2}, \rho)}$ & No & $\hat{\theta}_{j,c}^{\textsf{FH-SC}_{1}}$  & 1.23 & 8.68 
\\  \rowcolor{gray!20} 
FH-SC$_{2(\boldsymbol{\beta}, \sigma_c^{2}, \rho)}$  & No  & $\hat{\theta}_{j,c}^{\textsf{FH-SC}_{2}}$   & 1.11 & 8.33 
\\ 
FH-SC$_{3(\boldsymbol{\beta}_c, \sigma_c^{2}, \nu_{i}, \rho)}$   &  No & $\hat{\theta}_{j,c}^{\textsf{FH-SC}_{3}}$  & 1.24 & 8.75 
\\ 
FH$_{(\boldsymbol{\beta}, \sigma^{2})}$ & Yes & $\hat{\theta}_{j,c}^{\textsf{FH\textsf{-B}}}$  & 1.17 & 8.54 
\\ 
FH-C$_{1(\boldsymbol{\beta}, \sigma_c^{2})}$  & Yes & $\hat{\theta}_{j,c}^{\textsf{FH-C$_1$-B}}$  & 1.36 & 9.24 
\\ 
FH-C$_{2(\boldsymbol{\beta}_c, \sigma_c^{2}, \nu_{i})}$  & Yes & $\hat{\theta}_{j,c}^{\textsf{FH-C$_2$-B}}$  & 1.14 & 8.44 
\\ 
FH-SC$_{1(\boldsymbol{\beta}, \sigma^{2}, \rho)}$ & Yes & $\hat{\theta}_{j,c}^{\textsf{FH-SC$_1$-B}}$  & 1.37 & 9.29 
\\  \rowcolor{gray!20} 
FH-SC$_{2(\boldsymbol{\beta}, \sigma_c^{2}, \rho)}$  & Yes  & $\hat{\theta}_{j,c}^{\textsf{FH-SC$_2$-B}}$   & 1.13 & 8.39 
\\ 
FH-SC$_{3(\boldsymbol{\beta}_c, \sigma_c^{2}, \nu_{i}, \rho)}$   &  Yes & $\hat{\theta}_{j,c}^{\textsf{FH-SC$_3$-B}}$  & 1.39 & 9.35 
\\   \hline
\end{tabular}
\caption{\small EPD 
measures for PHIA under different models with and without benchmarking. In both cases, the smallest EPD(\%) 
values are observed for the FH-SC$_{2}$ model.} \label{tab:model_criteria}
\end{center}
\end{table}

To assess the performance of the estimators in producing PHIA small area estimates at the municipality level, we consider the small area MSE estimates and the Coefficient of Variation (CV) obtained under the various models. Specifically, to compare the uncertainty produced by the various estimators, we consider the direct variances, $D$, and the CV of direct estimates, {\small $\text{CV}(y_{j,c})= \sqrt{D_{j,c}}/y_{j,c}$}. We compare these values with the proposed CPMSE and CV of the RB and RB benchmarked estimators: {\small CPMSE$(\hat{\boldsymbol{\theta}}^{\M})$, CPMSE$(\hat{\boldsymbol{\theta}}^{\M\textsf{-B}})$, $\text{CV}(\hat{\boldsymbol{\theta}}^{\M})=\sqrt{\text{CPMSE}(\hat{\boldsymbol{\theta}}^{\M})}/\hat{\boldsymbol{\theta}}^{\M}$} and {\small $\text{CV}(\hat{\boldsymbol{\theta}}^{\M\textsf{-B}})=\sqrt{\text{CPMSE}(\hat{\boldsymbol{\theta}}^{\M\textsf{-B}})}/\hat{\boldsymbol{\theta}}^{\M\textsf{-B}}$}, where $\M$ refers a the specific model in Table \ref{tab:models}. We also compute the EBLUP and MSPE for the FH-SC$_{1}$ and FH-SC$_{2}$ models using the estimators proposed by \cite{maiti2014clustering} and \cite{torkashvand2017clustering}, and compare them with the posterior RB estimates of PHIA and their corresponding CPMSE. According to Proposition \ref{prop:def3}, benchmarking increases the CPMSE relative to the CPMSE of the RB estimator, {\small CPMSE$(\hat{\theta}_{j,c}^{\M}\mid \boldsymbol{X}_c, \boldsymbol{Z}_c, \boldsymbol{y}_{c})$}. This was also noted by \cite{you2002benchmarking} for the PMSE.

\begin{figure}[t]
\begin{center}
\vspace{0,3cm}
\begin{tabular}{ccc}
 a)  $\sqrt{D}$,  $\sqrt{\text{CPMSE}}$ and  $\sqrt{\text{MSPE}}$ & \hspace{1cm} b)  $\text{CV}$\% \\
 \hspace{-0.2cm} \includegraphics[width=0.50\textwidth]{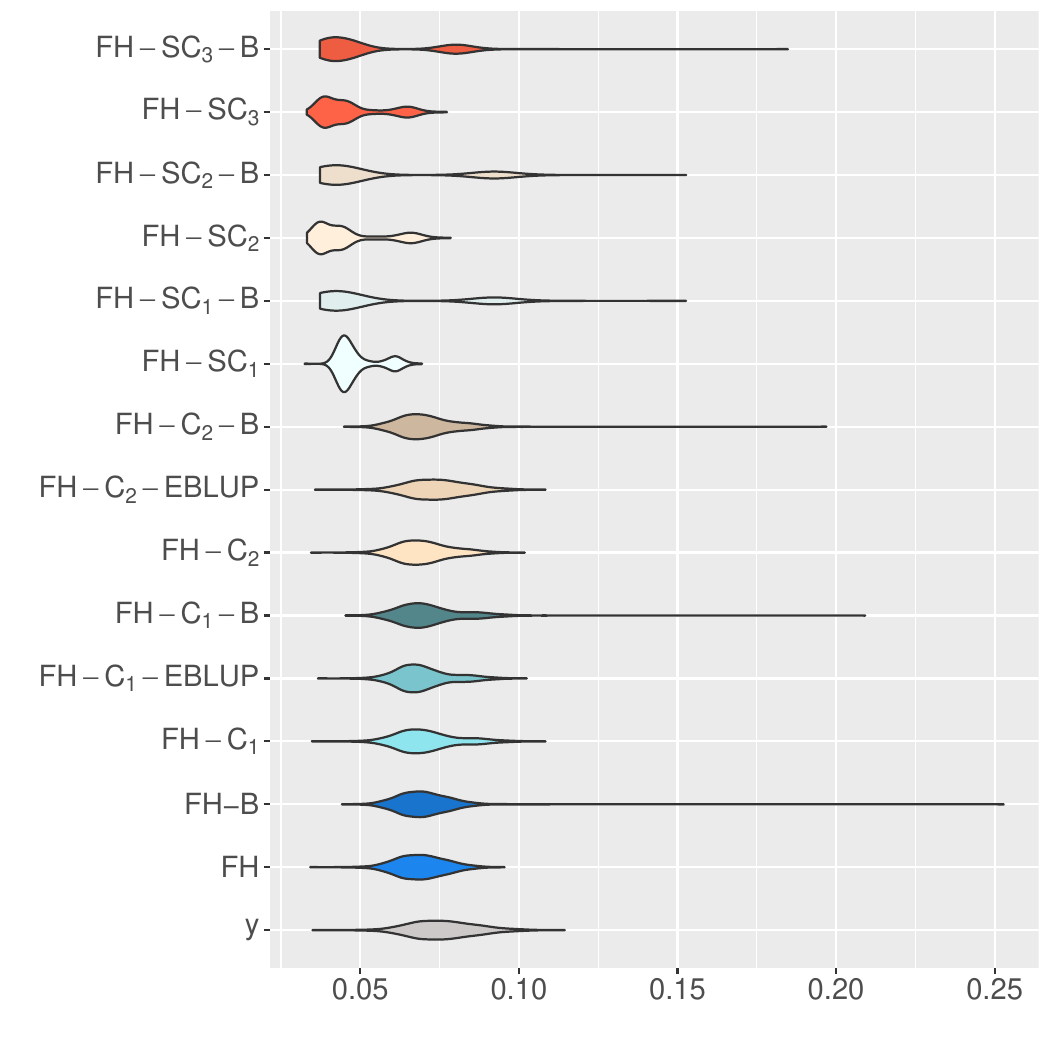} &  \hspace{-0.5cm}
\includegraphics[width=0.50\textwidth]{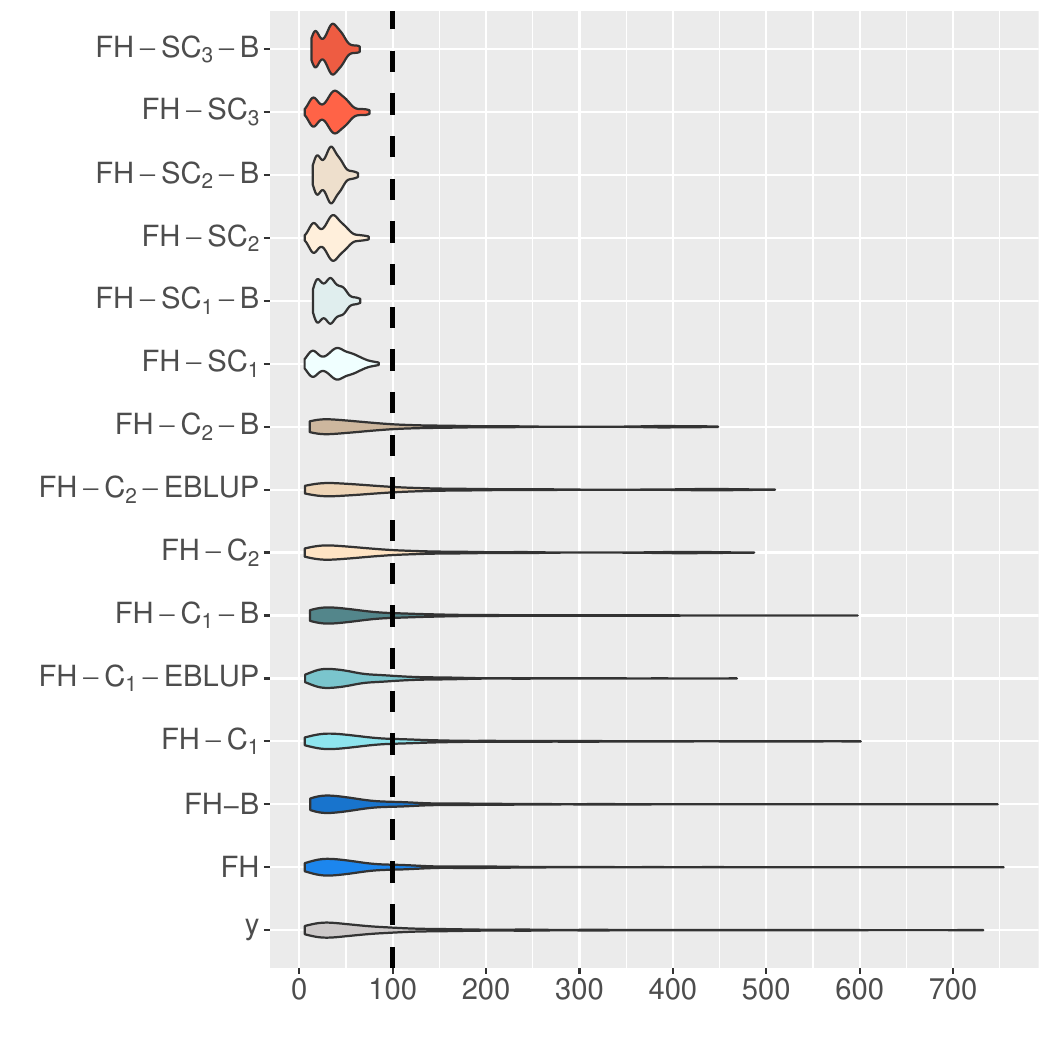} 
\end{tabular}
\end{center}
\caption{a) Standard deviation of direct estimates, square root of 
Mean-Squared Prediction Error (MSPE) of the EBLUPs for FH-C models, and square root of CPMSE for all estimators presented  in Table \ref{tab:model_criteria}. b) Coefficients of Variation (CV) in percentage. The FH-SC models lead to significant reductions in the CVs compared to existing models.}.
\label{fig:figr2}
\end{figure}

Figure \ref{fig:figr2} (left) illustrates this behavior. Under FH, FH-C and FH-SC models the RB benchmarked estimates of PHIA produce larger CPMSE estimates compared to those without benchmarking. While introducing benchmarking generally leads to higher MSE estimates, the RB benchmarked estimator under the FH-SC model achieves significant reductions in MSE estimates and coefficients of variation compared to the existing FH and FH-C models, as shown in Figure \ref{fig:figr2}. As expected, Figure \ref{fig:figr2} also illustrates that the MSPE and CV values under FH-C$_{1}$ and FH-C$_{2}$, obtained with the frequentist estimators proposed in \cite{maiti2014clustering} and \cite{torkashvand2017clustering}, are similar to those obtained under a Bayesian framework. Overall, we found that the FH-SC models produce smaller coefficients of variation even when sample sizes are small. Notably, under the FH-SC$_{2}$ model, more precise PHIA estimates are obtained in around 92.17\% of the 294 municipalities. These findings are illustrated in Supplementary Figure \ref{fig:figcve_ratio}.

Supplementary Table \ref{table:conf_cve} presents PHIA estimates for major capital cities and other relevant municipalities known to have higher poverty and/or education deficit index values in 2018. According to these results, the direct and RB estimates under the FH model are more conservative for some municipalities compared to those produced under the proposed FH-SC$_{2}$ model. Importantly, the PHIA estimates in Supplementary Table  \ref{table:conf_cve} are in accordance with the poverty levels in 2018 at the subnational level.  A key advantage of our proposal is being able to provide the posterior distribution of RB estimators including those with benchmarking. Except for the capital city of Bogot\'a, the posterior distributions of the RB benchmarked estimates under the FH-SC$_{2}$ model are closer to the direct estimates and
within the 95\% confidence interval constructed with the direct estimates and the direct variances (see Supplementary Figures \ref{fig:posterior_estimates1} and 
\ref{fig:posterior_estimates2}). 

\begin{figure}[t!]
\begin{center}
\begin{tabular}{ccc}
FH &  \hspace{2.5cm}  FH-SC$_{2}$ \vspace{-1.0cm}  \\
\hspace{-1.5cm}  \includegraphics[width=0.6\textwidth]{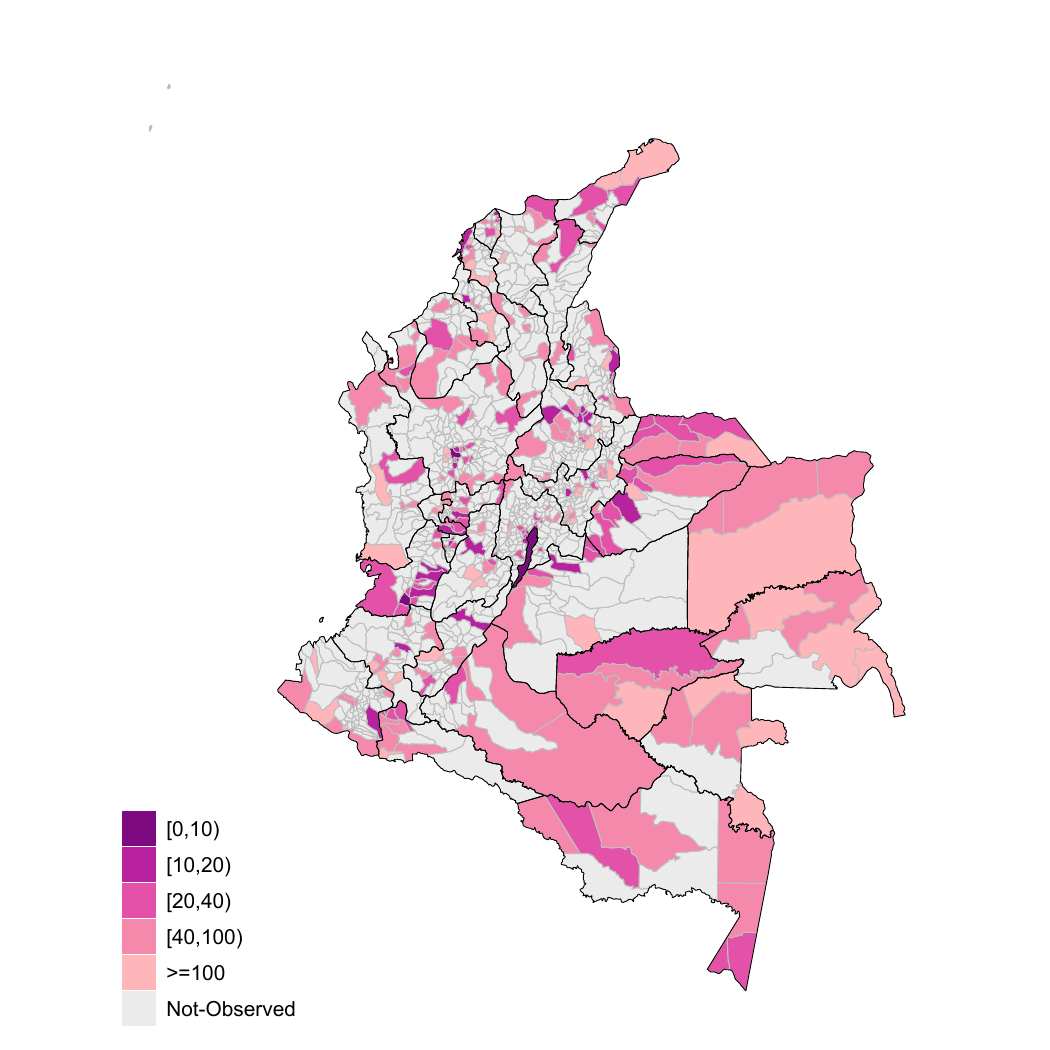} \hspace{-1cm} & \includegraphics[width=0.6\textwidth]{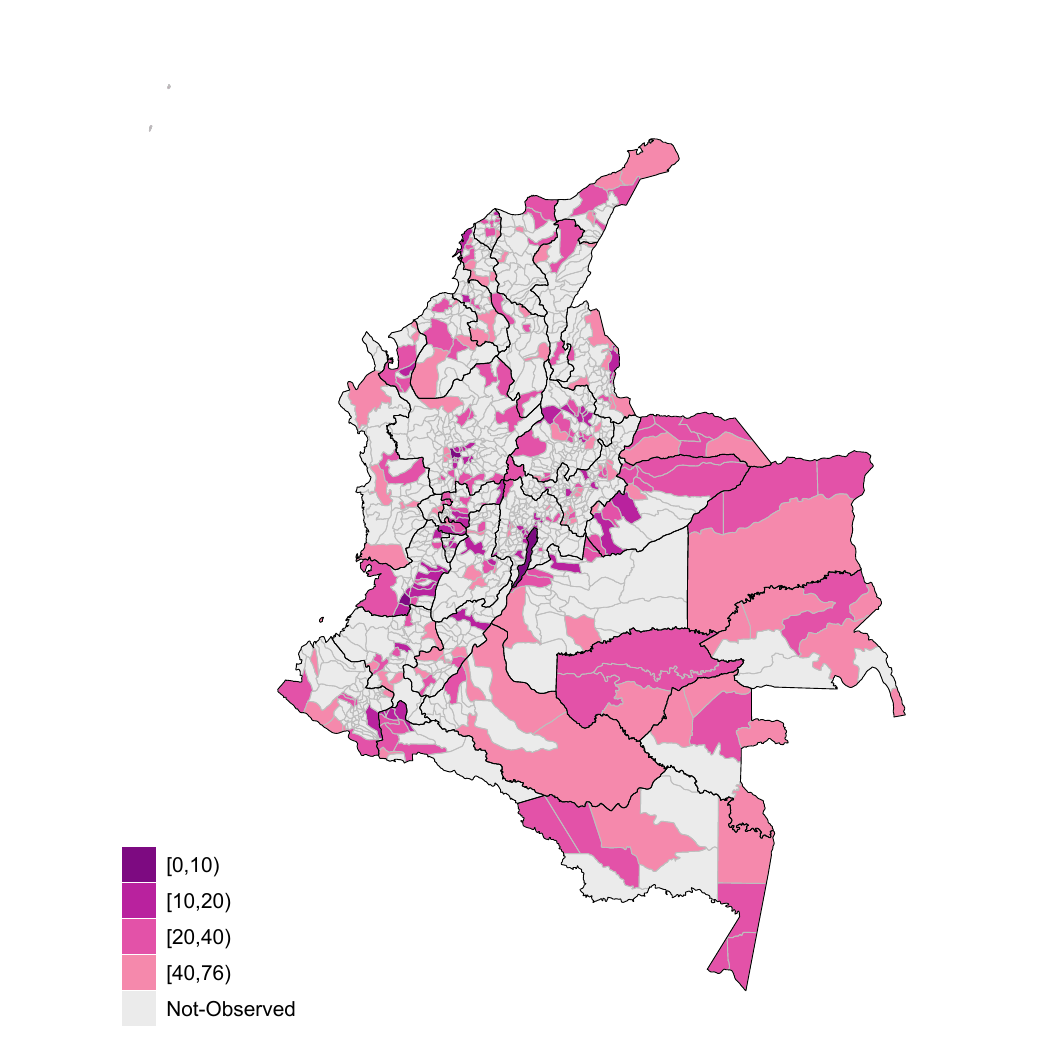}  \hspace{-2cm}  \\
FH-B &  \hspace{2.5cm}  FH-SC$_{2}$-B \vspace{-1.0cm}  \\
\hspace{-1.5cm}  \includegraphics[width=0.6\textwidth]{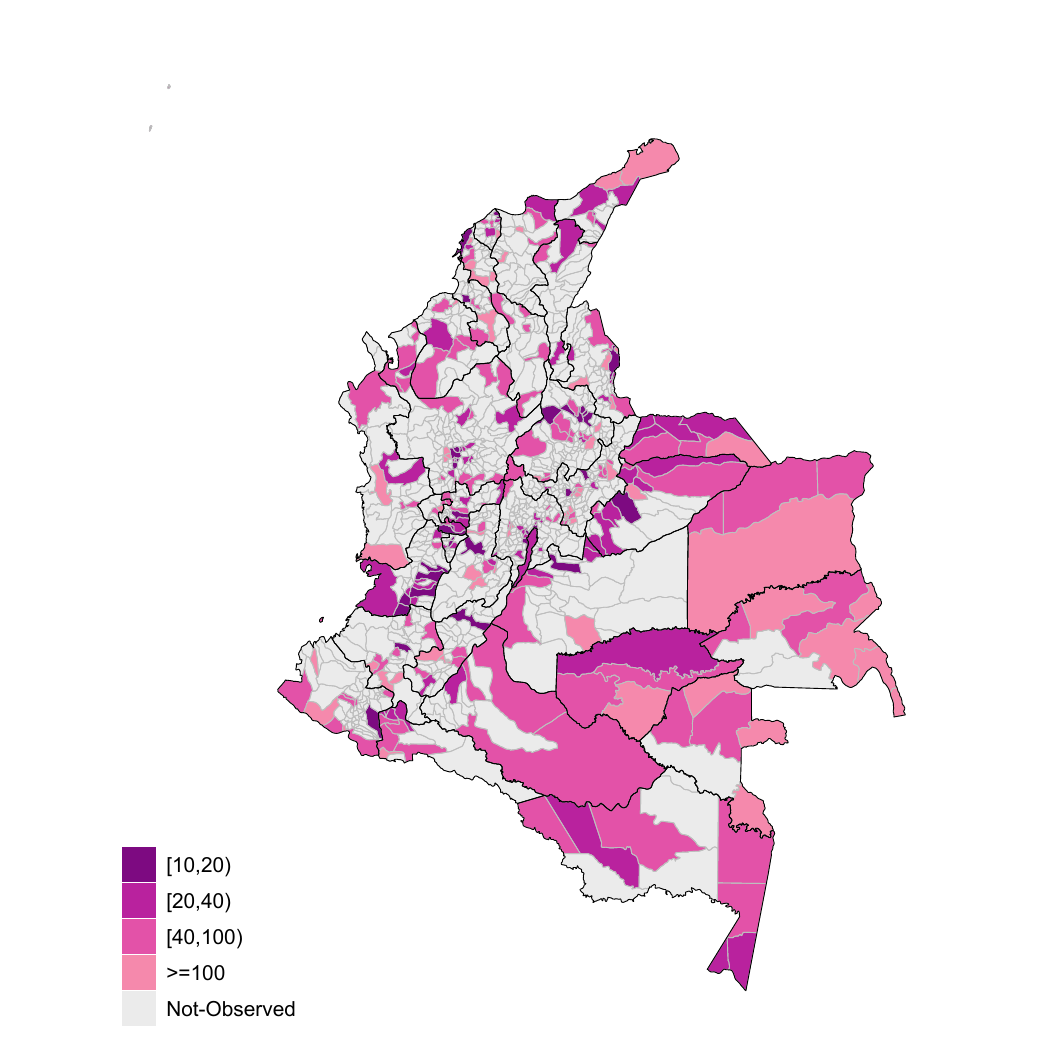} \hspace{-1cm} & \includegraphics[width=0.6\textwidth]{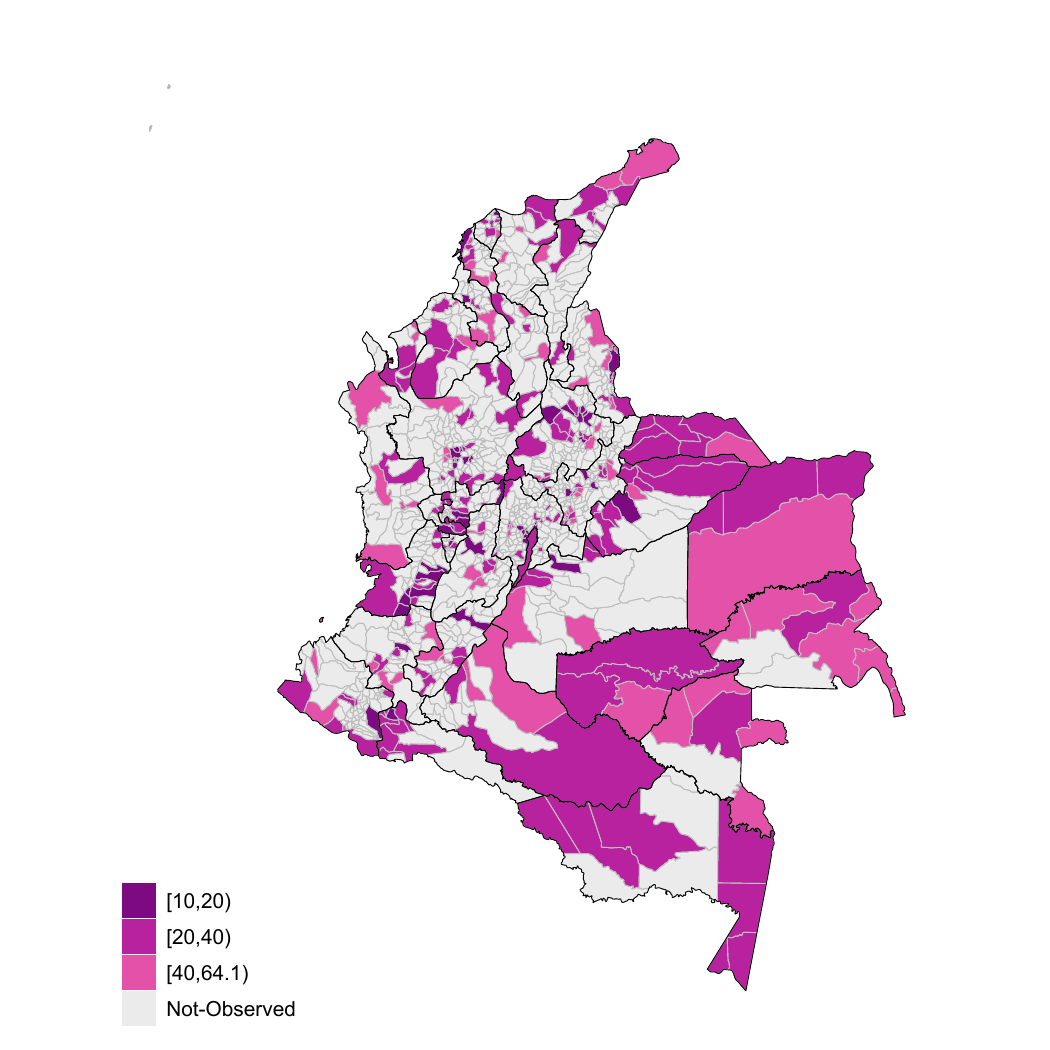}  \hspace{-2cm}  \\
\end{tabular}
\end{center}
\vspace{-0.5cm}
\caption{
Estimates of coefficient of variations of RB estimates produced by a) $\hat{\theta}_{j,c}^{\textsf{FH}}$, b) $\hat{\theta}_{j,c}^{\textsf{FH-SC$_{2}$}}$, c) $\hat{\theta}_{j,c}^{\textsf{FH}\text{-B}}$  and, d) $\hat{\theta}_{j,c}^{\textsf{FH-SC$_{2}$}\text{-B}}$. 
$\hat{\theta}_{j,c}^{\textsf{FH-SC$_{2}$}\text{-B}}$ produces the smallest coefficients of variation of RB benchmarked estimates of  PHIA.}
\label{fig:post2}
\end{figure}

As expected, the main capital cities -- Bogotá, D.C., Medellín, and Cali -- exhibit the highest values of internet connectivity. In contrast, the capital cities of Leticia, San José Del Guaviare, and Riohacha have the lowest posterior estimates of PHIA. These low PHIA estimates are likely correlated with high poverty levels and challenges affecting the education system in these municipalities. Jamundí and Quibdó are representative of many municipalities in Colombia due to their diverse ethnic populations, poverty levels, and education system issues. Notably, Quibdó exhibits lower RB posterior estimates of PHIA compared to Jamundí, possibly reflecting its higher poverty levels, as indicated by the MPI in 2018 (see Supplementary Table \ref{table:conf_cve}). An illustration of the spatial patterns of the different RB estimates of PHIA is presented in Supplementary Figure \ref{fig:post1}. 
Figure \ref{fig:post2} displays the estimated coefficient of variations for the RB estimates. Overall, the coefficients of variation under the selected models, FH-SC$_{2}$ and FH-SC$_{2}$-B, are lower than 20\% and generally smaller compared to those produced by the other RB estimators (see Supplementary Table~\ref{table:conf_cve}). Hence, these models are reliable for generating official statistical reports on PHIA at the municipal level in Colombia. 

\section{Discussion}
\label{sec:discussion}


Our contributions span methodological, computational, and applied domains.  
Methodologically, we propose the Fay-Herriot model with Spectral Clustering (FH-SC), which classifies small areas based on covariates rather than geographical or administrative criteria.  This approach enhances the precision of estimates by leveraging structural similarities among areas. Additionally, we introduce a novel measure of uncertainty, the Conditional Posterior Mean Square Error (CPMSE), specifically designed for benchmarked estimators under a Bayesian framework. This measure provides a more informative assessment of estimator reliability compared to traditional approaches. 
In practical terms, our framework is employed to estimate the Proportion of Households with Internet Access (PHIA) at the municipal level in Colombia. Given the policy importance of PHIA, particularly in low- and middle-income countries, our approach provides a viable alternative for estimating this key indicator in intercensal periods.  
Moreover, the availability of posterior distributions for RB benchmarked estimates allows practitioners to derive comprehensive uncertainty assessments, thereby facilitating data-driven decision-making. Crucially, the results of the case study highlight significant reductions in coefficients of variation and improvements in precision, suggesting that the proposed methodology is a valuable contribution for official statistics. 
Future research directions include extending the FH-SC model to unit-level data, broadening its applicability beyond area-level estimation. Additionally, our methodology could be applied to other key socioeconomic indicators where clustering structures can be informed by auxiliary administrative data.

\spacingset{1.0} 








\clearpage

\if1\blind
{
    {\LARGE\bf Supplementary Material for ``Fully Bayesian  Spectral Clustering and Benchmarking with  Uncertainty Quantification  for Small Area Estimation''}
}
} \fi

\if0\blind
{
  \bigskip
  \bigskip
  \bigskip
  \begin{center}
    {\LARGE\bf Supplementary Material for ``Fully Bayesian  Spectral Clustering and Benchmarking with  Uncertainty Quantification  for Small Area Estimation''}
\end{center}
  \medskip
} \fi

\bigskip


\renewcommand{\thesection}{\Alph{section}}
\setcounter{section}{0}
\counterwithin{figure}{section}
\counterwithin{table}{section}
\setcounter{algocf}{0}
\renewcommand{\thealgocf}{A\arabic{algocf}}

The following sections provide the supplementary material of our work. Section \ref{theory} contains the proofs of 
Propositions \ref{prop:LapRLS}, \ref{prop:benchmarking} and  \ref{prop:def3} and Theorems \ref{expectation}, \ref{prop2} and \ref{thm:PP_problem}. 
Section \ref{MCMC_supp}, provides details of the computational Algorithms. Specifically, Section \ref{SC_algh} describes the proposed Spectral Clustering Algorithm  \ref{alg:SC} and Section \ref{MCM_algh1} contains Algorithms \ref{alg:MCMC1} and \ref{alg:MCMC2}, used to obtain posterior samples of the model parameters $\boldsymbol{\kappa}=\{\boldsymbol{\theta}, \boldsymbol{\delta},\boldsymbol{G}_{\boldsymbol{\varphi}}, \rho\}$ for the FH-SC models presented in Table \ref{tab:models}. To evaluate our methodological proposal, we perform two simulation studies described in Section \ref{sec:simula}. Section \ref{applied} contains the supplementary material for the PHIA case study in the main document.

\section{Proofs of Theorems and Propositions}
\label{theory}


\subsection{Proof of Proposition \ref{prop:LapRLS} }
\label{app:prop1}

\begin{proof}
The derivative of the objective function (\ref{eq:LapRLS2}) vanishes at the minimizer 
\begin{align}
-2\rho(\bdt-\boldsymbol{\theta}^{\M})  + 2 (1-\rho)\textit{L}_{SC}\boldsymbol{\theta}^{\M} = \boldsymbol{0}
\end{align}
which leads to the following solution
\begin{align}
\boldsymbol{\theta}^{\M} = ( \mathbb I_m +((1-\rho)/\rho)  \textsf{L}_{SC})^{-1}\bdt.
\end{align}

\end{proof}

\subsection{Proof of Theorem \ref{prop2}}
\label{app:prop3}

\begin{proof}
To prove  Theorem \ref{prop2} we need to show that the integral in the right-hand side of  (\ref{eq:post1}) 
with respect to $\boldsymbol{\kappa}_c$, is finite.
Consider $ A_{\rho,c} = ( \mathbb I_{n_{c}} + ((1-\rho)/\rho) \textsf{L}_{c})$ symmetric and  positive-definite 
with  $ A_{\rho,c} \in \mathbb{R}^{n_c \times n_c}$ where $\textsf{L}_{c}=n_c\mathbb I_{n_{c}} - \boldsymbol{1}_{n_c}\boldsymbol{1}_{n_c}^\textsf{T}$  and $\rho \in (0,1]$.
 Therefore, we can use the Sherman-Morrison expression \cite{sherman1950adjustment2} to compute
 $A_{\rho}^{-1}= \text{blkdiag}(\{A_{\rho,c}^{-1}\}_{c=1}^{C})$ where $A_{\rho,c}^{-1} = \gamma_{c}  \mathbb I_{n_{c}} + ((1-\gamma_{c})/n_c)\boldsymbol{1}_{n_c}\boldsymbol{1}_{n_c}^\textsf{T}$  and $ \gamma_{c}=  \rho /((1-\rho)n_{c} + \rho)$. Consider the prior  distribution $p(\bdt_c | \boldsymbol{\delta}_c, \boldsymbol{G}_{\boldsymbol{\varphi},c}, \boldsymbol{X}_c, \boldsymbol{Z}_c)$  and likelihood distribution $p( \by_c | \bdt_c, \rho, \boldsymbol{D}_c)$ 
with $c=1,...,C$  as follows,

\begin{align}\small
\label{fay_herriot2_here}
\begin{split}
    \by_c & \sim \text{Normal}(A_{\rho,c}^{-1}\bdt_c,  \boldsymbol{D}_c ),\\ 
\bdt_c &\sim  \text{Normal}\left(\boldsymbol{X}_c\boldsymbol{\delta}_c,  \boldsymbol{Z}_c\boldsymbol{G}_{\boldsymbol{\varphi},c}\boldsymbol{Z}_c^\textsf{T} \right),\\
\end{split}
\end{align}

with  $\bdt_c=(\theta_{1,c},...,\theta_{n_c,c})^\textsf{T}$,  $\by_c=(y_{1,c},...,y_{n_c,c})^\textsf{T}$, $\boldsymbol{D}_{c}=\text{diag}(D_{1,c},,...,D_{n_c,c})$
and  $C\leq m$. Integrating (\ref{fay_herriot2_here}) with respect to $\bdt_c$, we have 
\begin{align*}
p(\boldsymbol{G}_{\boldsymbol{\varphi},c}, \rho \mid \boldsymbol{y}_c)  \leq K  |A_{\rho,c}^{-1}\boldsymbol{Z}_c\boldsymbol{G}_{\boldsymbol{\varphi},c}\boldsymbol{Z}_c^\textsf{T} A_{\rho,c}^{-1} + \boldsymbol{D}_{c}|^{-1/2}  
\pi(\boldsymbol{G}_{\boldsymbol{\varphi},c}) \pi(\rho),
\end{align*}
where $K$ is a generic positive constant. 
 Since $A_{\rho,c}$ and $\boldsymbol{Z}_c$  are full rank matrices then  $A_{\rho,c}^{-1} 
\boldsymbol{Z}_c\boldsymbol{G}_{\boldsymbol{\varphi},c}\boldsymbol{Z}_c^\textsf{T}A_{\rho,c}^{-1} $ is nonsingular.  By monotonicity of the determinant on the cone of positive definite
matrices we have  $|A_{\rho,c}^{-1}\boldsymbol{Z}_c\boldsymbol{G}_{\boldsymbol{\varphi},c}\boldsymbol{Z}_c^\textsf{T} A_{\rho,c}^{-1} + \boldsymbol{D}_{c}| \geq |\boldsymbol{D}_{c}|$ and therefore we have  $|A_{\rho,c}^{-1}\boldsymbol{Z}_c\boldsymbol{G}_{\boldsymbol{\varphi},c}\boldsymbol{Z}_c^\textsf{T} A_{\rho,c}^{-1} + \boldsymbol{D}_{c}|^{-1/2} \leq|\boldsymbol{D}_{c}|^{-1/2}$. Hence, for $j=1,...,n_c$, we obtain
\begin{align}
p(\boldsymbol{G}_{\boldsymbol{\varphi},c}, \rho \mid \boldsymbol{y}_c)  &\leq K  \pi(\boldsymbol{G}_{\boldsymbol{\varphi},c}) \pi(\rho) \prod_{j=1}^{n_c}  D_{j,c}^{-1/2}\\
\end{align}
since  $D_{j,c}>0$, then the posterior distribution under the FH-SC$_{1}$, FH-SC$_{2}$  and
 FH-SC$_{3}$ models in (\ref{eq:post1}) is proper if the priors
for the variance parameters $\boldsymbol{\varphi}$ in the covariance matrix $\boldsymbol{G}_{\boldsymbol{\varphi},c}$ and the prior for the cluster regularization penalty $\rho$ are proper.

\end{proof}

\subsection{Proof of Theorem \ref{expectation}}
\label{app:prop}

\begin{proof}



 Consider (\ref{fay_herriot2_here}) and Compute the first and second moments of the conditional posterior distribution  $p(\bdt_c | \boldsymbol{\delta}_c,\boldsymbol{G}_{\boldsymbol{\varphi},c}, \rho, \boldsymbol{Z}_{c}, \boldsymbol{X}_{c}, \by_c)$ 
using (\ref{fay_herriot2_here}) given by

\begin{align}
\label{smooth1_a_here}
E(\boldsymbol{\theta}_c  \mid \boldsymbol{\delta}_c,\boldsymbol{G}_{\boldsymbol{\varphi},c}, \rho, \boldsymbol{Z}_{c}, \boldsymbol{X}_{c}, \by_c) &=   V(\boldsymbol{\theta}_c\mid \boldsymbol{\delta}_c,\boldsymbol{G}_{\boldsymbol{\varphi},c}, \rho,\boldsymbol{Z}_{c}, \boldsymbol{X}_{c}, \by_c)\\  
& \hspace{1.0cm}\times (A_{\rho,c}^{-1} \boldsymbol{D}_c^{-1}  \by_c   +(\boldsymbol{Z}_{c}\boldsymbol{G}_{\boldsymbol{\varphi},c}\boldsymbol{Z}_{c}^\textsf{T} )^{-1}\boldsymbol{X}_c\boldsymbol{\delta}_c), \notag
\end{align}
\begin{align}
\label{smooth1_b}
V(\boldsymbol{\theta}_c\mid \boldsymbol{\delta}_c,\boldsymbol{G}_{\boldsymbol{\varphi},c}, \rho, \boldsymbol{Z}_{c}, \boldsymbol{X}_{c}, \by_c) &=((A_{\rho,c} \boldsymbol{D}_c  A_{\rho,c}^\textsf{T})^{-1} +  (\boldsymbol{Z}_{c}\boldsymbol{G}_{\boldsymbol{\varphi},c}\boldsymbol{Z}_{c}^\textsf{T} )^{-1} )^{-1}.
\end{align}

Therefore, part $(i)$  and $(ii)$ hold, by taking the expectation and variance on both sides of $\bdt_c^{\textsf{FH-SC}} =  A_{\rho,c}^{-1}\bdt_c$, conditional on $\boldsymbol{\delta}_c,\boldsymbol{G}_{\boldsymbol{\varphi},c}, \rho, \boldsymbol{Z}_{c}, \boldsymbol{X}_{c}$ and $\by_c$.  

\end{proof}


\subsection{Proof of Proposition  \ref{prop:benchmarking} }
\label{app:prop4}

\begin{proof}

To prove Proposition \ref{prop:benchmarking}, we need to find the optimal solution for  $\bdt^{\M\textsf{-SC-B}}$ using the objective function in (\ref{eq:LapRLS2}).
Note that  (\ref{eq:LapRLS3}) in the main document  is equivalent to minimizing the convex differentiable objective function  given by
\begin{align}
\label{eq:LapRLS3_here}
\bdt^{\M\textsf{-SC-B}} &=    \underset{\boldsymbol{\theta}^{\M}}{\text{minimize}} \quad    -2\rho(\boldsymbol{\theta}^{\M})^\textsf{T}\boldsymbol{\theta}+  (\boldsymbol{\theta}^{\M})^\textsf{T}(\rho\mathbb I_m +  (1-\rho)L_{SC})\boldsymbol{\theta}^{\M}, \\ &\hspace{2cm} \text{subject to} \quad \boldsymbol{W}\boldsymbol{\theta}^{\M} = \boldsymbol{p}, \notag
\end{align}
where $\boldsymbol{W} \in \mathbb{R}^{k \times m}$  has full row rank $k \leq m$ and  $\rho$,  $\boldsymbol{\theta}^{\M}$, $\bdt$ and $\textsf{L}_{SC}$ are as defined in Proposition \ref{prop:LapRLS}. We follow similar steps to those  in the proof of Theorem 14 in \cite{patra2019constrained}, and therefore, we require
the two Karush–Kuhn–Tucker (KKT) conditions \citep{karush1939minima, kuhn1951nonlinear} in \cite{patra2019constrained} to hold for the Constrained Quadratic Optimization Problem (CQOP) in (\ref{eq:LapRLS3_here}).
Since the Laplacian matrix $\textsf{L}_{SC}$ is symmetric and  positive semi-definite \citep{von2007tutorial}, $ \rho \mathbb I_m + (1-\rho) \textsf{L}_{SC}$ is positive definite
and the first KKT condition holds. For the second KKT condition,  we assume the equality benchmarking constraints $\boldsymbol{W} \in \mathbb{R}^{k \times m}$, $\boldsymbol{W}$ has full row rank $k \leq m$ as defined in \cite{patra2019constrained}. 

The KKT conditions for the solution $\bdt^{\M\textsf{-SC-B}} \in \mathbb{R}^{m}$ of the CQOP lead to the following linear system: 

\begin{equation}
\label{matrix1}
\begin{pmatrix}
2 \rho \boldsymbol{A}_{\rho} & \boldsymbol{W}^T\\
\boldsymbol{W} & \boldsymbol{0}
\end{pmatrix}
\begin{pmatrix}
\bdt^{\M\textsf{-SC-B}}\\
\boldsymbol{\hat{\lambda}}^{*}
\end{pmatrix} = 
 \begin{pmatrix}
2 \rho \bdt\\
\boldsymbol{p}
\end{pmatrix},
\end{equation}

where $\boldsymbol{\hat{\lambda}}^{*} \in \mathbb{R}^{k}$ is the associated Lagrange multiplier. Using the inverse of the $2 \times 2$
block  matrix in (\ref{matrix1}), we have 
\begin{equation*}
\begin{pmatrix}
\bdt^{\M\textsf{-SC-B}}\\
\boldsymbol{\hat{\lambda}}^{*}
\end{pmatrix} = 
\begin{pmatrix}
(2 \rho \boldsymbol{A}_{\rho})^{-1}   + (2 \rho \boldsymbol{A}_{\rho})^{-1}  \boldsymbol{W}^T \boldsymbol{Z} \boldsymbol{W} (2 \rho \boldsymbol{A}_{\rho})^{-1}  & -(2 \rho \boldsymbol{A}_{\rho})^{-1}  \boldsymbol{W}^T \boldsymbol{Z}\\
-\boldsymbol{Z}\boldsymbol{W} (2 \rho \boldsymbol{A}_{\rho})^{-1}  & \boldsymbol{Z}
\end{pmatrix} \begin{pmatrix}
2 \rho  \bdt\\
\boldsymbol{p}
\end{pmatrix},
\end{equation*}
where $\boldsymbol{Z} = - ( \boldsymbol{W} (2 \rho \boldsymbol{A}_{\rho})^{-1} \boldsymbol{W}^T)^{-1}$ and therefore 
\begin{align}
\label{sol}
\notag \bdt^{\M\textsf{-SC-B}}= \boldsymbol{A}_{\rho}^{-1}\bdt& \\
& \hspace{-1.0cm}
+ (2\rho \boldsymbol{A}_{\rho})^{-1} \boldsymbol{W}^T \boldsymbol{Z} \boldsymbol{W}  \boldsymbol{A}_{\rho}^{-1}\bdt -( 2\rho \boldsymbol{A}_{\rho})^{-1}\boldsymbol{W}^T \boldsymbol{Z} \boldsymbol{p}
\end{align}

\end{proof}

\subsection{Theorem  \ref{thm:PP_problem}: Benchmarking estimation using posterior projections}
\label{app:prop6}

\begin{thm}
 \small
 \label{thm:PP_problem} 
Consider  $\bdt^{\textsf{$\M$-SC(l)}} \in \Theta^{\textsf{$\M$}}$ and $\widetilde{\bdt}^{\textsf{$\M$} } \in \widetilde{\Theta}^{\textsf{$\M$}}$ 
where $\widetilde{\Theta}^{\textsf{$\M$}} = \{ \widetilde{\bdt}^{\textsf{$\M$} } \in \Theta^{\textsf{$\M$}} \mid \boldsymbol{W}\widetilde{\bdt}^{\textsf{$\M$} } = \boldsymbol{p} \}$, we have the following:
\begin{enumerate}[(i)] 
\item The convex differentiable objective function in Proposition  \ref{prop:benchmarking} is a projection problem, where minimizing (\ref{eq:LapRLS3})  is equivalent to 
\begin{align}
\label{eq:LapRLS_PP}  \underset{\widetilde{\bdt}^{\textsf{$\M$} }}{\text{minimize}} \quad  \{ \norm{\bdt^{\textsf{$\M$-SC(l)}} - \widetilde{\bdt}^{\textsf{$\M$} }}_{\boldsymbol{A}_{\rho^{(l)}}^{-1}}: \widetilde{\bdt}^{\textsf{$\M$} } \in \widetilde{\Theta}^{\textsf{$\M$}} \},
\end{align}
where $\widetilde{\bdt}^{\textsf{$\M$}}$ is the projection of $\bdt^{\textsf{$\M$-SC(l)}}$  via the weighted inner product defined by $\boldsymbol{A}_{\rho^{(l)}}^{-1}$. 
\item  When model $\M$ in part (i) is the FH-SC model in Definition \ref{def:FHSC_m}, the solution of (\ref{eq:LapRLS_PP}) is given by the projected samples
\begin{align}
\label{benchmark_PP_sup}
 \bdt^{\textsf{FH-SC-B(l)}}   &= \bdt^{\textsf{FH-SC(l)}}  + \boldsymbol{A}_{\rho^{(l)}}^{-1} \boldsymbol{W}^T (\boldsymbol{W} \boldsymbol{A}_{\rho^{(l)}}^{-1} \boldsymbol{W}^T)^{-1} (\boldsymbol{p} - \boldsymbol{W}  \bdt^{\textsf{FH-SC(l)}} ),   
\end{align}
where $l=1,...,L$  are posterior samples under the FH-SC model.
\item  The conditional expectation of  the small area benchmarked parameter  vector $\bdt^{\textsf{FH-SC-B}}$ for the 
$l$-th posterior sample is given by,
\begin{align} 
 \label{benchmark_sup}
    E(\bdt^{\textsf{FH-SC-B(l)}}  \mid \boldsymbol{X},
     \boldsymbol{Z},  \boldsymbol{y},\vartheta_{-\boldsymbol{\theta^{(l)}}}^{\textsf{FH-SC}(l)})&=  E(\bdt^{\textsf{FH-SC(l)}}  \mid \boldsymbol{X},
     \boldsymbol{Z},  \boldsymbol{y},\vartheta_{-\boldsymbol{\theta^{(l)}}}^{\textsf{FH-SC}(l)}) \\  & \hspace{-2.5cm} + \boldsymbol{A}_{\rho^{(l)}}^{-1} \boldsymbol{W}^T (\boldsymbol{W} A_{\rho^{(l)}}^{-1} \boldsymbol{W}^T)^{-1}(\boldsymbol{p} - W E(\bdt^{\textsf{FH-SC(l)}}  \mid \boldsymbol{X},
     \boldsymbol{Z},  \boldsymbol{y},\vartheta_{-\boldsymbol{\theta^{(l)}}}^{\textsf{FH-SC}(l)}) ), \notag
\end{align} 
where  $E(\bdt^{\textsf{FH-SC(l)}}  \mid \boldsymbol{X},
     \boldsymbol{Z},  \boldsymbol{y},\vartheta_{-\boldsymbol{\theta^{(l)}}}^{\textsf{FH-SC}(l)})$ is the conditional posterior expectation for the $l$-th posterior sample given in (\ref{smooth1}) of Theorem \ref{expectation}, and $\vartheta_{-\boldsymbol{\theta}}^{\textsf{FH-SC}(l)}=(\boldsymbol{\delta}^{(l-1)},\boldsymbol{G}_{\boldsymbol{\varphi}^{(l-1)}}, \rho^{(l)})$. 

When the benchmarking constraints are given by  $\sum_{c=1}^{C} \sum_{j=1}^{n_c} w_{j,c} \theta_{j,c}^{\textsf{FH-SC}(l)} = p$,
the conditional expectation of the small area benchmarked parameter vector under the FH-SC model for
cluster $c$, $c=1,...,C$,  is as follows:
\begin{align}
\label{benchmark2}
    E(\bdt^{\textsf{FH-SC-B(l)}}_c  \mid \boldsymbol{X}_c,
     \boldsymbol{Z}_c,  \boldsymbol{y}_c,\vartheta_{-\boldsymbol{\theta^{(l)}}_c}^{\textsf{FH-SC}(l)})&=  E(\bdt^{\textsf{FH-SC-(l)}}_c  \mid \boldsymbol{X}_c,
     \boldsymbol{Z}_c,  \boldsymbol{y}_c,\vartheta_{-\boldsymbol{\theta^{(l)}}_c}^{\textsf{FH-SC}(l)}) \\  & \hspace{-0.5cm} +  a_{\rho^{(l)},c}(p - \boldsymbol{w}_c  E(\bdt^{\textsf{FH-SC-(l)}}_c  \mid \boldsymbol{X}_c,
     \boldsymbol{Z}_c,  \boldsymbol{y}_c,\vartheta_{-\boldsymbol{\theta^{(l)}}_c}^{\textsf{FH-SC}(l)})),   \notag        
 \end{align}
where  $\,\vartheta_{-\boldsymbol{\theta^{(l)}}_c}^{\textsf{FH-SC}(l)}=(\boldsymbol{\delta}^{(l-1)}_c,\boldsymbol{G}_{\boldsymbol{\varphi}^{(l-1)},c}, \rho^{(l)})$ and $\boldsymbol{w}_c=(w_{1,c},...,w_{n_c,c})$ denotes the benchmarking weights for cluster $c$, 
and $a_{\rho,c}^{(l)}= (\gamma_{c}^{(l)}\boldsymbol{w}_c^\textsf{T} +  (1-\gamma_{c}^{(l)})\bar{\boldsymbol{w}}_c\boldsymbol{1}_c)/(\gamma_{c}^{(l)}\sum_{j=1}^{n_c}w_{j,c}^{2} +  (1-\gamma_{c}^{(l)})n_c\bar{\boldsymbol{w}}_c^{2})$ with $\bar{\boldsymbol{w}}_c=\sum_{j=1}^{n_c}w_{j,c}/n_{c}$ and $\gamma_{c}^{(l)}=  \rho^{(l)} /((1-\rho^{(l)})n_{c} + \rho^{(l)})$.

\end{enumerate}

\end{thm}

\begin{proof}

Note that minimizing the objective function   in Proposition  \ref{prop:benchmarking}  is equivalent to 
\begin{align}
\label{eq:LapRLS3_proof}
\bdt^{\M\textsf{-SC-B}} & =
   \underset{\boldsymbol{\theta}^{\M}}{\text{minimize}} \quad     
   \rho(\boldsymbol{\theta}^{\M})^\textsf{T}( \mathbb I_m +((1-\rho)/\rho)\boldsymbol{\theta}^{\M} -2\rho(\boldsymbol{\theta}^{\M})^\textsf{T}\boldsymbol{\theta} , \\ &\hspace{2cm} \text{subject to} \quad \boldsymbol{W}\boldsymbol{\theta}^{\M} = \boldsymbol{p}, \notag
\end{align}
where $\boldsymbol{W} \in \mathbb{R}^{k \times m}$  has full row rank $k \leq m$ and  $\rho$,  $\boldsymbol{\theta}^{\M}$, $\bdt$ and $\textsf{L}_{SC}$  are as in Proposition  \ref{prop:LapRLS}. Replace  $\bdt$ with $\bdt^{(l)}$  and $\rho$ with $\rho^{(l)}$ for 
$l=1,...,L$ posterior samples, where $L$ is the number of posterior samples, and note that   (\ref{eq:LapRLS3_proof}) is equivalent to 

\begin{align}
\label{eq:LapRLS3_proof2} \notag
\bdt^{\M\textsf{-SC-B}} & =
   \underset{\boldsymbol{\theta}^{\M}}{\text{minimize}} \quad     
   \rho^{(l)}(\boldsymbol{\theta}^{\M}-      \bdt^{\textsf{$\M$-SC}(l)})^\textsf{T} \boldsymbol{A}_{\rho^{(l)}}(\boldsymbol{\theta}^{\M}- \bdt^{\textsf{$\M$-SC}(l)}) -  \rho^{(l)} (\boldsymbol{\theta}^{(l)})^\textsf{T} \boldsymbol{A}_{\rho^{(l)}}^{-1}\boldsymbol{\theta}^{(l)}
  \\ &\hspace{2cm} \text{subject to} \quad \boldsymbol{W}\boldsymbol{\theta}^{\M} = \boldsymbol{p},
\end{align}
where $\boldsymbol{A}_{\rho^{(l)}}=( \mathbb I_m +((1-\rho^{(l)})/\rho^{(l)})  \textit{L}_{SC})$ and  $\bdt^{\textsf{$\M$-SC}(l)}= \boldsymbol{A}_{\rho^{(l)}}^{-1} \boldsymbol{\theta}^{(l)}$. Therefore, $(i)$ follows by noting that
minimizing (\ref{eq:LapRLS3_proof2}) is equivalent to minimizing 
\begin{align}
\label{eq:LapRLS_PP_proof}  \underset{\widetilde{\bdt}^{\textsf{$\M$} }}{\text{minimize}} \quad  \{ \norm{\bdt^{\textsf{$\M$-SC(l)}} - \widetilde{\bdt}^{\textsf{$\M$} }}_{\boldsymbol{A}_{\rho^{(l)}}^{-1}}: \widetilde{\bdt}^{\textsf{$\M$} } \in \widetilde{\Theta}^{\textsf{$\M$}} \},
\end{align}
where $\widetilde{\bdt}^{\textsf{$\M$}}$ is the projection of $\bdt^{\textsf{$\M$-SC}(l)}$  via the weighted inner product defined by $\boldsymbol{A}_{\rho^{(l)}}^{-1}$ where $\widetilde{\Theta}^{\textsf{$\M$}} = \{ \widetilde{\bdt}^{\textsf{$\M$} } \in \Theta^{\textsf{$\M$}} \mid \boldsymbol{W}\widetilde{\bdt}^{\textsf{$\M$} } = \boldsymbol{p} \}$. To find  $\bdt^{\textsf{FH-SC-B(l)}}$ in part $(ii)$ under the FH-SC model, we consider the Lagrangian as follow 
\begin{align}
\mathcal{L}(\widetilde{\bdt}^{\M}, \boldsymbol{\lambda}) =(\widetilde{\bdt}^{\M}-\bdt^{\textsf{FH-SC}(l)})^\textsf{T} \boldsymbol{A}_{\rho^{(l)}}(\widetilde{\bdt}^{\M} - \bdt^{\textsf{FH-SC}(l)})  + \boldsymbol{\lambda}^T ( \boldsymbol{W}\widetilde{\bdt}^{\M} - \boldsymbol{p}),
\end{align}
where  $\boldsymbol{\lambda}$ is the Lagrange multiplier and solve the following equation system for $\widetilde{\bdt}^{\M }$:

\[
    \frac{\partial \mathcal{L}(\widetilde{\bdt}^{\M}, \boldsymbol{\lambda})}{\partial \widetilde{\bdt}^{\M}} = 2\boldsymbol{A}_{\rho^{(l)}}(\widetilde{\bdt}^{\M }-      \bdt^{\textsf{FH-SC}(l)}) + \boldsymbol{W}^T  \boldsymbol{\lambda} = \boldsymbol{0},
\]
\[
    \frac{\partial \mathcal{L}(\widetilde{\bdt}^{\M }, \boldsymbol{\lambda})}{\partial  \boldsymbol{\lambda}} = \boldsymbol{W}\widetilde{\bdt}^{\M} - \boldsymbol{p} = \boldsymbol{0}.
\]


Part $(iii)$ follows by computing the conditional expectation of  $\bdt^{\textsf{FH-SC-B(l)}}$ in $(ii)$ given $\boldsymbol{X}$,
     $\boldsymbol{Z}$,  $\boldsymbol{y}$ and $\vartheta_{-\boldsymbol{\theta^{(l)}}}^{\textsf{FH-SC}(l)}$
and,  $    E(\bdt^{\textsf{FH-SC-B(l)}}_c  \mid \boldsymbol{X}_c,
     \boldsymbol{Z}_c,  \boldsymbol{y}_c,\vartheta_{-\boldsymbol{\theta^{(l)}_c}}^{\textsf{FH-SC}(l)})$, by assuming that $\boldsymbol{W}=\boldsymbol{w}=(\boldsymbol{w}_1,...,\boldsymbol{w}_C)$  with $\boldsymbol{w}_c=(w_{1,c},...,w_{n_c,c})$ in $    E(\bdt^{\textsf{FH-SC-B(l)}}  \mid \boldsymbol{X},
     \boldsymbol{Z},  \boldsymbol{y},\vartheta_{-\boldsymbol{\theta^{(l)}}}^{\textsf{FH-SC}(l)})$.

\end{proof}
\subsection{Proof of Proposition \ref{prop:def3}}
\label{app:prop5}
\begin{proof}
Consider the posterior PMSE of the RB benchmarked estimator given by
\begin{align}
   \label{eq:CPMSE_projection_estimator1}
\text{PMSE}(\hat{\theta}_{j,c}^{\textsf{FH-SC-B}}\mid \boldsymbol{X}_c,
     \boldsymbol{Z}_c, \boldsymbol{y}_{c}) &= E_{\theta_{j,c}^{\textsf{FH-SC}}}((\hat{\theta}_{j,c}^{\textsf{FH-SC-B}}- \theta_{j,c}^{\textsf{FH-SC}} )^{2}\mid \boldsymbol{X}_c,
     \boldsymbol{Z}_c, \boldsymbol{y}_{c})
\end{align}
Using a RB argument  \citep{Rao1945, blackwell1947conditional} under the PMSE in  (\ref{eq:CPMSE_projection_estimator1}) we define the CPMSE  as follows,
\begin{align}
   \label{eq:CPMSE_projection_estimator2} 
\text{CPMSE}(\hat{\theta}_{j,c}^{\textsf{FH-SC-B}}\mid \boldsymbol{X}_c,
     \boldsymbol{Z}_c, \boldsymbol{y}_{c}) &=  \\ & \hspace{-1cm} 
      E_{\vartheta^{\textsf{FH-SC}}_{-\theta_{j,c}}}(E_{\theta_{j,c}^{\textsf{FH-SC}}}((\hat{\theta}_{j,c}^{\textsf{FH-SC-B}}- \theta_{j,c}^{\textsf{FH-SC}} )^{2}\mid \boldsymbol{X}_c,
     \boldsymbol{Z}_c, \boldsymbol{y}_{c},\vartheta^{\textsf{FH-SC}}_{-\theta_{j,c}})),   \notag
\end{align}
 where the first expectation is under a set of parameters denoted $\vartheta^{\textsf{FH-SC}}_{-\theta_{j,c}}$,
which includes the parameters under the FH-SC model except for $\theta_{j,c}$. 
We add and subtract  the RB estimator
of $\theta_{j,c}$ under the FH-SC model,  $\hat{\theta}_{j,c}^{\textsf{FH-SC}}$, on equation  (\ref{eq:CPMSE_projection_estimator2})  to establish the following equality:


\begin{align}
   \label{eq:CPMSE_projection_estimator3}
    E_{\vartheta^{\textsf{FH-SC}}_{-\theta_{j,c}}}(E_{\theta_{j,c}^{\textsf{FH-SC}}}((\hat{\theta}^{\textsf{FH-SC-B}}_{j,c}- \hat{\theta}_{j,c}^{\textsf{FH-SC}} + \hat{\theta}_{j,c}^{\textsf{FH-SC}}-\theta_{j,c}^{\textsf{FH-SC}} )^{2}\mid \boldsymbol{X}_c,
     \boldsymbol{Z}_c, \boldsymbol{y}_{c},\vartheta^{\textsf{FH-SC}}_{-\theta_{j,c}})) &=  \\ &\hspace{-12cm} 
         E_{\vartheta^{\textsf{FH-SC}}_{-\theta_{j,c}}}(E_{\theta_{j,c}^{\textsf{FH-SC}}}((\hat{\theta}^{\textsf{FH-SC-B}}_{j,c}- \hat{\theta}_{j,c}^{\textsf{FH-SC}} )^{2}\mid \boldsymbol{X}_c,
     \boldsymbol{Z}_c, \boldsymbol{y}_{c},\vartheta^{\textsf{FH-SC}}_{-\theta_{j,c}})) \notag \\ &\hspace{-11cm}  -2  
             E_{\vartheta^{\textsf{FH-SC}}_{-\theta_{j,c}}}(E_{\theta_{j,c}^{\textsf{FH-SC}}}((\hat{\theta}^{\textsf{FH-SC-B}}_{j,c}- \hat{\theta}_{j,c}^{\textsf{FH-SC}})(\hat{\theta}_{j,c}^{\textsf{FH-SC}}-\theta_{j,c}^{\textsf{FH-SC}} )\mid \boldsymbol{X}_c,
     \boldsymbol{Z}_c, \boldsymbol{y}_{c},\vartheta^{\textsf{FH-SC}}_{-\theta_{j,c}}))     \notag 
      \\ \notag \ \hspace{-2cm} 
       +   E_{\vartheta^{\textsf{FH-SC}}_{-\theta_{j,c}}}(E_{\theta_{j,c}^{\textsf{FH-SC}}}(( \hat{\theta}_{j,c}^{\textsf{FH-SC}}-\theta_{j,c}^{\textsf{FH-SC}} )^{2}\mid \boldsymbol{X}_c,
     \boldsymbol{Z}_c, \boldsymbol{y}_{c},\vartheta^{\textsf{FH-SC}}_{-\theta_{j,c}}) &=  \\ &\hspace{-12cm} 
     (\hat{\theta}_{j,c}^{\textsf{FH-SC-B}}-\hat{\theta}_{j,c}^{\textsf{FH-SC}})^{2}  \notag \\ &\hspace{-11cm}  -2  
             E_{\vartheta^{\textsf{FH-SC}}_{-\theta_{j,c}}}(E_{\theta_{j,c}^{\textsf{FH-SC}}}((\hat{\theta}^{\textsf{FH-SC-B}}_{j,c}- \hat{\theta}_{j,c}^{\textsf{FH-SC}})(\hat{\theta}_{j,c}^{\textsf{FH-SC}}-\theta_{j,c}^{\textsf{FH-SC}} )\mid \boldsymbol{X}_c,
     \boldsymbol{Z}_c, \boldsymbol{y}_{c},\vartheta^{\textsf{FH-SC}}_{-\theta_{j,c}}))     \notag 
      \\ \notag &  \hspace{-7cm} 
       +   E_{\vartheta^{\textsf{FH-SC}}_{-\theta_{j,c}}}(E_{\theta_{j,c}^{\textsf{FH-SC}}}(( \hat{\theta}_{j,c}^{\textsf{FH-SC}}-\theta_{j,c}^{\textsf{FH-SC}} )^{2}\mid \boldsymbol{X}_c,
     \boldsymbol{Z}_c, \boldsymbol{y}_{c},\vartheta^{\textsf{FH-SC}}_{-\theta_{j,c}})).          
       \end{align}
According to Definition \ref{def_RB}, we have  $\hat{\theta}_{j,c}^{\textsf{FH-SC}}= E_{\vartheta_{-\boldsymbol{\theta}}^{\textsf{FH-SC}}}(E(\theta_{j,c}^{\textsf{FH-SC}}\mid  \boldsymbol{X}_c,
     \boldsymbol{Z}_c,  \boldsymbol{y}_c, \vartheta_{-\boldsymbol{\theta}}^{\textsf{FH-SC}}))$ and noting that
\begin{align}
   \label{eq:CPMSE_projection_estimator4}
\notag E_{\vartheta^{\textsf{FH-SC}}_{-\theta_{j,c}}}(E_{\theta_{j,c}^{\textsf{FH-SC}}}(( \hat{\theta}_{j,c}^{\textsf{FH-SC}}-\theta_{j,c}^{\textsf{FH-SC}} )^{2}\mid \boldsymbol{X}_c,
     \boldsymbol{Z}_c, \boldsymbol{y}_{c},\vartheta^{\textsf{FH-SC}}_{-\theta_{j,c}}))&=
\text{CPMSE}( \hat{\theta}_{j,c}^{\textsf{FH-SC}} \mid \boldsymbol{X}_c,
     \boldsymbol{Z}_c, \boldsymbol{y}_{c})=  \\ &\hspace{-8cm} 
E_{\vartheta^{\textsf{FH-SC}}_{-\theta_{j,c}}}(\text{V}_{\theta_{j,c}^{\textsf{FH-SC}}}(\theta_{j,c}^{\textsf{FH-SC}}\mid \boldsymbol{X}_c,
     \boldsymbol{Z}_c, \boldsymbol{y}_{c}, \vartheta^{\textsf{FH-SC}}_{-\theta_{j,c}}))  \\ &\hspace{-6cm}  +E_{\vartheta^{\textsf{FH-SC}}_{-\theta_{j,c}}}((E_{\theta_{j,c}^{\textsf{FH-SC}}}( (\hat{\theta}_{j,c}^{\textsf{FH-SC}}-\theta_{j,c}^{\textsf{FH-SC}}) \mid\boldsymbol{X}_c,
     \boldsymbol{Z}_c, \boldsymbol{y}_{c},\vartheta^{\textsf{FH-SC}}_{-\theta_{j,c}}))^{2}), \notag
\end{align}
by considering RB estimators in 
(\ref{eq:CPMSE_projection_estimator4}), we obtain Proposition \ref{prop:def3}. 
\end{proof}


\section{Supplementary Algorithms}
\label{MCMC_supp}

To provide the cluster classification of the municipalities in our case study into $C$ clusters, $C\leq m$, we propose the spectral clustering Algorithm \ref{alg:SC} in Section \ref{SC_algh}. Algorithm  \ref{alg:SC} follows similar steps to those in Algorithm of \cite{von2007tutorial} but it incorporates more than two variables and a 
method to select the number of external covariates. The MCMC  Algorithms  \ref{alg:MCMC1} and \ref{alg:MCMC2} describe the  
steps to obtain posterior samples of the vector of model parameters $\boldsymbol{\kappa}=\{\boldsymbol{\theta}, \boldsymbol{\delta},\boldsymbol{G}_{\boldsymbol{\varphi}}, \rho\}$ under the different models in Table \ref{tab:models}.

\subsection{The spectral clustering Algorithm \ref{alg:SC}}
\label{SC_algh}

The proposed spectral clustering Algorithm \ref{alg:SC} is useful to provide a cluster classification of the external covariates and the PHIA at the municipality level and
the simple graph Laplacian matrix  $\textsf{L}_{SC}= \text{blkdiag}(\{\textsf{L}_{c}\}_{c=1}^{C})$. In  the input, Algorithm \ref{alg:SC} considers the direct estimates of PHIA, $y_{i}$, and  $k=1,...,p^{*}$ external covariates where $\boldsymbol{x}^{*}_{k}=(x^{*}_{1,k},...,x^{*}_{m,k})$. In our application  $p^{*}=2$ and $x^{*}_{i,1}$ and $x^{*}_{i,2}$ are the observed values of Educational Index and MPI in the $i$-th municipality.  To build the similarity graph in  step (1) of Algorithm \ref{alg:SC},  we use  the radial-kernel gram matrix  \citep{gonen2011multiple} as the similarity matrix. 
The similarity function  is
$
s_{i,j}^{k} =\exp\{-  ((x^{*}_{i,k}-x^{*}_{j,k})^{2} + (y_{i}-y_{j})^{2})/(2\sigma^2_{s}) \}$
where $\sigma^2_{s}=1$ according to the specifications in the Manifold Regularization Algorithm proposed by \cite{belkin2006manifold} and the data applications in \cite{von2007tutorial}.  By using $s_{i,j}^{k}$, we construct a similarity graph in step (2) of Algorithm \ref{alg:SC}.


 By using $s_{i,j}^{k}$, we construct a similarity graph in step (2) of Algorithm \ref{alg:SC}. More formally, let $\text{G} = \langle \text{V} , \text{E}  \rangle$ be an undirected graph with vertex set $\text{V} = \{\nu_{1}, . . . , \nu_{m}\}$ where pairs of vertices are connected by an edge $E$, if their similarity is positive or exceeds some threshold (See for instance \cite{hastie2009elements}).  The edges are weighted by using the similarity values $s_{i,j}^{k}$. 
Several alternatives to define the similarity matrix and its associated similarity graph are described in detail in \cite{von2007tutorial}. In particular, to construct the similarity graph in step (2) of Algorithm \ref{alg:SC}, we employ the $C-$nearest graph method where a symmetric set of nearby pairs of points $\Upsilon_C$ is assumed. More specifically, a pair $(i,j)$ is in the set $\Upsilon_C$ if point $i$ is among the $C$-nearest neighbors of $j$, or vice-versa. Then, all symmetric nearest neighbors have edge weight $\eta_{i,j} >0$ computed using  $s_{i,j}^{k}$ 
when $(i,j)$ are connected, otherwise the edge weight is zero. Since the proposed Algorithm \ref{alg:SC}  is designed to consider more than one external covariate, $p^{*}>1$,  we follow \cite{gonen2011multiple} to compute the weights for each pair $(i,j)$ using  the weighted sum of the similarities computed for each covariate $k$,  i.e., $\eta_{i,j}= \sum_{k=1}^{p^{*}} \alpha_{k} s_{i,j}^{k}$. 
  
In our application $ \alpha_{1} = \alpha_{2} =0.5$ since both the Educational Index and MPI are computed from the same data source - the 2014 Census of Agriculture \citep{CNA2014} - and may have the same level of accuracy. In step (3) of Algorithm \ref{alg:SC}, we compute the unnormalized weighted Laplacian matrix  \ref{alg:SC}, $\textsf{L}_u=\text{D}_u-\text{W}_{A}$, which is symmetric and positive semi-definite (see for instance Proposition 1 of \cite{von2007tutorial}). To compute $\textsf{L}_{u}$, we need to calculate the degree matrix $\text{D}_u$ and the weighted adjacency matrix $\text{W}_{A}$ using the weights $\eta_{i,j}$ obtained in step (2), where $\eta_{i,j} \ge 0$ and $\eta_{i,j}=\eta_{j,i}$.  In steps (4a)-(4c), we build the matrix $V_{C}$ containing the eigenvectors $v_{1} , . . . , v_{C}$ of $\textsf{L}_{u}$ as columns and cluster the rows of $V_{C}$ to yield a clustering of the direct estimates of PHIA.  
Lastly, we compute the simple graph Laplacian matrix  $\textsf{L}_{SC}= \text{blkdiag}(\{\textsf{L}_{c}\}_{c=1}^{C})$ to be included in the FH-SC model for SAE. To choose the number of clusters and external covariates, we consider the combination that minimizes the total within-cluster sums of squares as described in \cite{ward1963hierarchical}. 


\begin{algorithm}[h] 
\footnotesize
  \begin{algorithmic}
    \State \textbf{Input: direct estimate $y_{i}$ with $i=1,...,m$, external covariates $\boldsymbol{x}^{*}_{k}=(x_{1,k}^{*},...,x_{m,k}^{*})$ for $k=1,...,p^{*}$,   
    and the number of clusters $C$ to be constructed.}
\SetAlgoLined  
\begin{enumerate}[(1)]
        \item Compute  the similarity matrix for each external covariate $\boldsymbol{x}^{*}_{k}$,  $\textsf{S}_{k}=(s_{i,j}^{k})_{i,j=1,...,m} \in \mathbb{R}^{m \times m}$ where         $s_{i,j}^{k} =\exp\{-  ((x^{*}_{i,k}-x^{*}_{j,k})^{2} + (y_{i}-y_{j})^{2})/(2\sigma^2_{s}) \}$.   
    \item Construct the similarity graph using the $C-$nearest graph method  and compute the  weights $\eta_{i,j}= \sum_{k=1}^{p^{*}} \alpha_{k} s_{i,j}^{k}$.
     \item 
      Compute the unnormalized Laplacian $\textsf{L}_{u}=\textsf{D}_u-\text{W}_{A}$. The weighted adjacency matrix of the graph, $\text{W}_{A}$, is calculated using the  weights $\eta_{i,j}$ with $\eta_{i,j} \ge 0$ and $\eta_{i,j}=\eta_{j,i}$. When two vertices $\nu_{i}$ and $\nu_{j}$ are connected $\eta_{i,j}>0$, otherwise $\eta_{i,j}=0$. The diagonal degree matrix, $D_u=\text{diag}(d_1,\dots,d_m)$, uses $d_{i}= \sum_{j=1}^{m} \eta_{i,j}$.  
      \end{enumerate}
      \begin{enumerate}[(4a)]
     \item Compute the first $c$ eigenvectors $v_{1} , . . . , v_{C}$ of $\textsf{L}_{u}$.
     \item Let $V_{C} \in \mathbb{R}^{m \times C}$ be the matrix containing the vectors $v_{1},...,v_{C}$ as columns. 
     \item For $i=1,...,m$, let  $r_{i} \in \mathbb{R}^{C}$ be the vector corresponding to the $i$-th row of $V_{C}$.
     \item Cluster the points $(r_{i})_{i=1,...,m} \in \mathbb{R}^{C}$ with the $k$-means algorithm into clusters $c=1,...,C$.
     \item Create $D_{c} =\{j | r_{j} \in c\}$ clusters with $c=1,...,C$.     
     \item Compute the Laplacian matrix  $\textsf{L}_{SC}$, where the $(i,j)$-th element of $\textsf{L}_{SC}$ is given by 
$$\textsf{L}_{SC}(i,j):=\begin{cases}
n_c-1, & \text{ if }\;\; i=j \;\; \text{ and} \;\;  i \in D_{c},\\
-1, & \text{ if }\;\; i \neq j \;\; \text{ and} \;\;  i \in D_{c},\\
0 & \text{Otherwise}.
\end{cases}
$$
      \end{enumerate}
                \State \textbf{Output: Laplacian matrix $\textsf{L}_{SC}= \text{blkdiag}(\{\textsf{L}_{c}\}_{c=1}^{C})$, where $\textsf{L}_{c}=n_c\mathbb I_{n_{c}} - \boldsymbol{1}_{n_c}\boldsymbol{1}_{n_c}^\textsf{T}$ and   $n_c=\mid \{j | r_{j} \in c\}\mid$. The total within-cluster sum of squares given by $\sum_{c=1}^{C} \sum_{j=1}^{n_c} (y_{j,c}-\bar{y}_c)^2$ where $\bar{y}_c=\sum_{j=1}^{n_c}y_{j,c}$.}  
  \end{algorithmic}
    \caption{\footnotesize Unnormalized spectral clustering \citep{von2007tutorial}  adapted to build the cluster classification in the case study.}
     \label{alg:SC}
\end{algorithm}


\subsection{Prior specification and Posterior distribution}
\label{prior_spec}
For the FH-SC models in Table \ref{tab:models},  we consider a joint prior distribution  $\pi(\boldsymbol{\delta}_c,\boldsymbol{G}_{\boldsymbol{\varphi},c},\rho_c) =  \prod_{c=1}^{C} \pi(\boldsymbol{G}_{\boldsymbol{\varphi},c})\pi(\boldsymbol{\delta}_c)\pi(\rho)$, where $\pi(\boldsymbol{G}_{\boldsymbol{\varphi},c})$ denotes the prior for the variance parameters  $\boldsymbol{\varphi}$ in the covariance matrix $\boldsymbol{G}_{\boldsymbol{\varphi},c}$, and $\pi(\boldsymbol{\delta}_c)$ and $\pi(\rho)$ the priors for $\boldsymbol{\delta}_c$ and $\rho$, respectively. We consider improper Uniform priors for $\boldsymbol{\delta}_c^\textsf{T}=(\delta_{1,c}...,\delta_{p,c})$ for  $c=1,...,C$, 
and a $\text{Beta}(a,b)$ prior distribution on $\rho$ with hyperparameters $a$ and $b$. Our proposed FH-SC model in (\ref{fay_herriot_s}) with these  prior specifications can be written in hierarchical form as follows:
\begin{align}\small
\label{fay_herriot2}
\begin{split}
    \by_c & \sim \text{Normal}(A_{\rho,c}^{-1}\bdt_c,  \boldsymbol{D}_c ),\\ 
\bdt_c &\sim  \text{Normal}\left(\boldsymbol{X}_c\boldsymbol{\delta}_c,  \boldsymbol{Z}_c\boldsymbol{G}_{\boldsymbol{\varphi},c}\boldsymbol{Z}_c^\textsf{T} \right),\\
\rho &\sim \text{Beta}( a,b),\\
\pi(\boldsymbol{\delta}_c,\boldsymbol{G}_c ,\rho) &= \prod_{c=1}^{C} \pi(\boldsymbol{G}_{\boldsymbol{\varphi},c})\pi(\rho),
\end{split}
\end{align}
with   $\bdt^{\textsf{FH-SC}}_c   =  A_{\rho,c}^{-1}  \bdt_c$ and  $A_{\rho,c}^{-1}$ and $\rho \in (0,1]$ as in Definition  \ref{def:FHSC_m}. 

Therefore, we can write the joint posterior density
for $\boldsymbol{\kappa}=\{\boldsymbol{\theta}, \boldsymbol{\delta},\boldsymbol{G}_{\boldsymbol{\varphi}}, \rho\}$ with $\boldsymbol{G}_{\boldsymbol{\varphi}}= \text{blkdiag}(\{\boldsymbol{G}_{\boldsymbol{\varphi},c}\}_{c=1}^{C})$  and $\boldsymbol{\delta}^\textsf{T}=(\boldsymbol{\delta}_1,...,\boldsymbol{\delta}_C)$ as,
\begin{align}\small
\label{eq:post1} 
p(\boldsymbol{\kappa}\mid \boldsymbol{y}, \boldsymbol{X}, \boldsymbol{Z}, \boldsymbol{D}) =  \prod_{c=1}^{C}  p(\boldsymbol{\kappa}_c \mid \boldsymbol{y}_c, \boldsymbol{X}_c, \boldsymbol{Z}_c, \boldsymbol{D}_c),
\end{align}
where $\boldsymbol{\kappa}_c=\{\boldsymbol{\theta}_c, \boldsymbol{\delta}_c,\boldsymbol{G}_{\boldsymbol{\varphi}, c}, \rho\}$ and  
{\small
\begin{align*}
p(\boldsymbol{\kappa}_c \mid \boldsymbol{y}_c, \boldsymbol{X}_c, \boldsymbol{Z}_c, \boldsymbol{D}_c)& \propto   |\boldsymbol{Z}_{c}\boldsymbol{G}_{\boldsymbol{\varphi},c}\boldsymbol{Z}_{c}^\textsf{T}|^{-1/2} 
 \exp\left\{-\frac{1}{2}(A_{\rho,c}\boldsymbol{y}_{c} - \boldsymbol{\theta}_{c})^\textsf{T}A_{\rho,c}^{-1}\boldsymbol{D}_{c}^{-1}A_{\rho,c}^{-1}(A_{\rho,c}\boldsymbol{y}_{c} - \boldsymbol{\theta}_{c})\right\} \\   &  \hspace{-0.5cm}  \times  \exp\left\{-\frac{1}{2}(\boldsymbol{\theta}_{c} - \boldsymbol{X}_{c}\boldsymbol{\delta}_{c})^\textsf{T}( \boldsymbol{Z}_{c}\boldsymbol{G}_{\boldsymbol{\varphi},c}\boldsymbol{Z}_{c}^\textsf{T})^{-1}(\boldsymbol{\theta}_{c} - \boldsymbol{X}_{c}\boldsymbol{\delta}_{c})\right\}
\pi(\boldsymbol{G}_{\boldsymbol{\varphi},c}) \pi(\rho).
\end{align*}
}

Algorithms \ref{alg:MCMC1} and \ref{alg:MCMC2} in Section \ref{MCM_algh1} are used to generate posterior samples of the model parameters $\boldsymbol{\kappa}$ for the FH-SC models presented in Table \ref{tab:models}, according to the joint posterior density in (\ref{eq:post1}). 

\subsection{Supplementary Markov chain Monte Carlo Algorithms \ref{alg:MCMC1} and \ref{alg:MCMC2}}
\label{MCM_algh1}

 Since the conditional distribution for $\rho$  does not have a known form, we use an adaptive Metropolis within Gibbs step in Algorithm \ref{alg:MCMC1}. 
To obtain acceptable rejection rates  we use similar procedures to those implemented for the dependent
parameter in 
Bayesian spatial  models for econometrics \citep{lesage2009introduction, wilhelm2013estimating}. More specifically, to achieve moves over the  entire conditional distribution of $\rho$ in step (1) of Algorithm \ref{alg:MCMC1}, we consider a random-walk procedure where the proposal distribution is Normal and the tuning parameter $\kappa^{\textsf{new}}$ is adjusted to hold acceptance rates between 40\% and 60\%. We implemented  this procedure in our data-based simulation  study in supplementary Section \ref{subsec:simulations} and the motivating application in Section \ref{sec:application}. We observed good mixing for all parameters by adjusting the tuning parameter, multiplying by a factor $\nu$, $\kappa^{\textsf{new}}=\kappa^{\textsf{old}}\nu$. Specifically, we consider $\nu = 1/1.1$ or $\nu = 1.1$ when the acceptance rate during the MCMC falls below 40\% or rises above 60\%, respectively. As in \cite{lesage2009introduction} and \cite{wilhelm2013estimating} for  Bayesian spatial  models, we use a non-informative Beta prior for $\rho$ with $a = b = 1.1$, which induces near zero probability mass on the end points of the interval $(0, 1)$. In addition, by considering $\rho \sim$ Beta(1.1, 1.1) with $\rho \in$ (0,1), we avoid $\rho=1$ in the FH-SC model and facilitate the comparison with FH and FH-C models. The conditional distribution of $\boldsymbol{\theta}_c$  in closed form is provided by Theorem \ref{expectation}  and therefore, we can easily implement two Gibbs sampling steps \citep{gelfand1990sampling2} in  step (2-a) of Algorithm \ref{alg:MCMC2}.
 
Note that  Algorithm  \ref{alg:MCMC1} can be implemented to obtain posterior samples of $\boldsymbol{\theta}^{\textsf{FH-SC}}_{c}$ and $\rho$, and to compute the conditional expectation and variance of $\boldsymbol{\theta}_c$ for the FH-SC models in Table \ref{tab:models}. Importantly, by using $\rho^{(l)}=1$ in Algorithm \ref{alg:MCMC1}, we can obtain posterior samples from the FH and FH-C models in Table \ref{tab:models}. Due to the specific settings for $\boldsymbol{\delta}_c$
 and $\boldsymbol{G}_{\boldsymbol{\varphi},c}$  in Table \ref{tab:models}, we designed Supplementary Algorithm  \ref{alg:MCMC2} 
 to obtain posterior samples for the regression parameters, $\boldsymbol{\delta}_c$, and variance components,  $\boldsymbol{\varphi}_c$,  under the different models in Table \ref{tab:models}. According to Theorem \ref{prop2}, the posterior distribution in (\ref{eq:post1}) is proper if the prior for the cluster regularization penalty $\rho$ and the prior for the variance parameters $\boldsymbol{\varphi}_c$ are proper. We consider independent Gamma($a_{1/\boldsymbol{\varphi}_c}$, $b_{1/\boldsymbol{\varphi}_c}$) priors to sample the scales $1/\boldsymbol{\varphi}_c$ for each cluster $c$. We found that using small hyperparameter values, such as $a_{(1/\boldsymbol{\varphi}_c)}=1$ and $b_{(1/\boldsymbol{\varphi}_c)}=1$, is effective in achieving good mixing for the model parameters, particularly for the parameter $\rho$.

%
\begin{algorithm}[h] 
\footnotesize
  \begin{algorithmic}
    \State  \textbf{Input: $(\boldsymbol{\theta}^{(l-1)}, \boldsymbol{\delta}^{(l-1)}, \boldsymbol{G}_{\boldsymbol{\varphi}^{(l-1)}}= \text{blkdiag}(\{\boldsymbol{G}_{\boldsymbol{\varphi}^{(l-1)},c}\}_{c=1}^{C}), \rho^{(l-1)})$ and compute $A_{\rho^{(l-1)},c}$,  $\boldsymbol{A}_{\rho^{(l-1)}}= \text{blkdiag}(\{A_{\rho^{(l-1)},c}\}_{c=1}^{C})$ and $\boldsymbol{\theta}^{\textsf{FH-SC}(l-1)}=\boldsymbol{A}_{\rho^{(l-1)}}^{-1}\boldsymbol{\theta}^{(l-1)}$.}
\SetAlgoLined
\For{$l=1,...,L$}{ 
\begin{enumerate}[(1)] 
\item  \noindent Generate  
$\log(\rho^{*}) \sim \text{Normal}(\log(\rho^{*}); \log(\rho^{(l-1)}), \kappa^{\textsf{new}})$
and draw $\rho^{(l)}$ with acceptance probability
\begin{align*}
\min\left\{1,
\dfrac{ \text{N}\left(\bdt^{(l-1)},  \boldsymbol{A}_{\rho^{*}}^{-1} \boldsymbol{X}\boldsymbol{\delta}^{(l-1)},  \boldsymbol{A}_{\rho^{*}}^{-1} \boldsymbol{Z}\boldsymbol{G}_{\boldsymbol{\varphi}^{(l-1)}}\boldsymbol{Z}^\textsf{T} \boldsymbol{A}_{\rho^{*}}^{-1} \right) \times \text{Beta}(\rho^{*}, a,b)}{ \text{N}\left(\bdt^{(l-1)},  \boldsymbol{A}_{\rho^{(l-1)}}^{-1} \boldsymbol{X}\boldsymbol{\delta}^{(l-1)},  \boldsymbol{A}_{\rho^{(l-1)}}^{-1} \boldsymbol{Z}\boldsymbol{G}_{\boldsymbol{\varphi}^{(l-1)}}\boldsymbol{Z}^\textsf{T} \boldsymbol{A}_{\rho^{(l-1)}}^{-1} \right) \times \text{Beta}(\rho^{(l-1)}, a,b)} \times
\dfrac{\rho^{*}}{
\rho^{(l-1)}}\right\},
\end{align*}
and update $\boldsymbol{A}_{\rho^{(l)}}= \text{blkdiag}(\{A_{\rho^{(l)},c}\}_{c=1}^{C})$.
\end{enumerate}
\For{$c=1,...,C$}{  
\begin{enumerate}[(2-a)] 
\item   \noindent Draw $\boldsymbol{\theta}^{(l)}_c\mid \boldsymbol{\delta}^{(l-1)}_c,\boldsymbol{G}_{\boldsymbol{\varphi}^{(l-1)},c}, \rho^{(l)}, \by_c$ using the Normal distribution,
\noindent \begin{equation*}
\boldsymbol{\theta}^{(l)}_c \sim \text{Normal}\left(E(\boldsymbol{\theta}_c^{(l)} \mid \boldsymbol{\delta}^{(l-1)}_c,\boldsymbol{G}_{\boldsymbol{\varphi}^{(l-1)},c}, \rho^{(l)}, \boldsymbol{Z}_{c}, \boldsymbol{X}_{c}, \by_c), V(\boldsymbol{\theta}^{(l)}_c\mid \boldsymbol{\delta}^{(l-1)}_c,\boldsymbol{G}_{\boldsymbol{\varphi}^{(l-1)},c}, \rho^{(l)}, \boldsymbol{Z}_{c}, \boldsymbol{X}_{c}, \by_c) \right),
\end{equation*}
where $E(\boldsymbol{\theta}_c^{(l)} \mid \boldsymbol{\delta}^{(l-1)}_c,\boldsymbol{G}_{\boldsymbol{\varphi}^{(l-1)},c}, \rho^{(l)}, \boldsymbol{Z}_{c}, \boldsymbol{X}_{c}, \by_c)$ and $V(\boldsymbol{\theta}_c^{(l)} \mid \boldsymbol{\delta}^{(l-1)}_c,\boldsymbol{G}_{\boldsymbol{\varphi}^{(l-1)},c}, \rho^{(l)}, \boldsymbol{Z}_{c}, \boldsymbol{X}_{c}, \by_c)$ are computed using (\ref{smooth1_a}) and (\ref{smooth1_b}) in Theorem \ref{expectation}, respectively.
\item  Update $ \boldsymbol{\theta}_{c}^{\textsf{FH-SC}(l)}=  A_{\rho^{(l)},c}^{-1} \boldsymbol{\theta}_{c}^{\textsf{FH-SC}(l)}$
\end{enumerate}
}
\begin{enumerate}[(3)] 
\item  Update $\boldsymbol{\theta}^{\textsf{FH-SC}(l)} = (\boldsymbol{\theta}_{1}^{\textsf{FH-SC}(l)},...,\boldsymbol{\theta}_{C}^{\textsf{FH-SC}(l)})^\textsf{T}$ 
\end{enumerate}
}    
\State \textbf{Output: $E(\boldsymbol{\theta}_c^{(l)} \mid \boldsymbol{\delta}^{(l-1)}_c,\boldsymbol{G}_{\boldsymbol{\varphi}^{(l-1)},c}, \rho^{(l)}, \boldsymbol{Z}_{c}, \boldsymbol{X}_{c}, \by_c)$, $V(\boldsymbol{\theta}_c^{(l)} \mid \boldsymbol{\delta}^{(l-1)}_c,\boldsymbol{G}_{\boldsymbol{\varphi}^{(l-1)},c}, \rho^{(l)}, \boldsymbol{Z}_{c}, \boldsymbol{X}_{c}, \by_c)$,   $\rho^{(l)}$,  $\boldsymbol{\theta}^{\textsf{FH-SC}(l)} = (\boldsymbol{\theta}_{1}^{\textsf{FH-SC}(l)},...,\boldsymbol{\theta}_{C}^{\textsf{FH-SC}(l)})^\textsf{T}$  with $\boldsymbol{\theta}_{c}^{\textsf{FH-SC}(l)}=(\theta_{1,c}^{\textsf{FH-SC}(l)},...,\theta_{n_c,c}^{\textsf{FH-SC}(l)})^\textsf{T}$  and  $\boldsymbol{\theta}^{(l)} = (\boldsymbol{\theta}_{1}^{(l)},...,\boldsymbol{\theta}_{C}^{(l)})^\textsf{T}$ with $\boldsymbol{\theta}_{c}^{(l)}=(\theta_{1,c}^{(l)},...,\theta_{n_c,c}^{(l)})^\textsf{T}$.
}      
\end{algorithmic}
 \caption{ \footnotesize Adaptive Metropolis-within-Gibbs sampling  for $\boldsymbol{\theta}^{(l)}_c$, $\boldsymbol{\theta}^{\textsf{FH-SC}(l)}_c$ and $\rho^{(l)}$ for $c=1,...,C$ under the FH-SC model  (\ref{fay_herriot2}). The conditional posterior expectation and variance of  $\boldsymbol{\theta}_c^{(l)}$ and $\boldsymbol{A}_{\rho^{(l)}}$ are also produced in the output.}
 \label{alg:MCMC1}
\end{algorithm}

 

 \begin{algorithm}[h]
\footnotesize
  \begin{algorithmic}
    \State \textbf{Input: ($\boldsymbol{\theta}^{\textsf{FH-SC}(l)}$,  $\boldsymbol{\theta}^{(l)}$, $\boldsymbol{G}_{\boldsymbol{\varphi}^{(l-1)}}= \text{blkdiag}(\{\boldsymbol{G}_{\boldsymbol{\varphi}^{(l-1)},c}\}_{c=1}^{C})$) and compute $A_{\rho^{(l)},c}$,  $\boldsymbol{A}_{\rho^{(l)}}= \text{blkdiag}(\{A_{\rho^{(l)},c}\}_{c=1}^{C})$.}
\begin{enumerate}[(1)] 
\item \eIf{$\boldsymbol{\delta}_c=\boldsymbol{\beta}$}{
Draw 
\begin{equation*}
\boldsymbol{\beta}^{(l)}  \mid \boldsymbol{G}_{\boldsymbol{\varphi}^{(l-1)}}, \boldsymbol{A}_{\rho^{(l)}}, \boldsymbol{\theta}^{\textsf{FH-SC}(l)}, \by \sim \text{Normal} \left(M_{\bb}^{(l)} , V_{\boldsymbol{\bb}}^{(l)}\right),
\end{equation*}
\begin{align*}
V_{\boldsymbol{\bb}}^{(l)}&=(\boldsymbol{X}^\textsf{T} (\boldsymbol{Z}\boldsymbol{G}_{\boldsymbol{\varphi}^{(l-1)}}\boldsymbol{Z}^\textsf{T})^{-1}\boldsymbol{X})^{-1} , &
M_{\bb}^{(l)} &=V_{\boldsymbol{\bb}}^{(l)} ( \boldsymbol{X}^\textsf{T}\boldsymbol{A}_{\rho^{(l)}}(\boldsymbol{Z}\boldsymbol{G}_{\boldsymbol{\varphi}^{(l-1)}}\boldsymbol{Z}^\textsf{T})^{-1}\boldsymbol{\theta}^{\textsf{FH-SC}(l)}),
\end{align*}
and update $\boldsymbol{\delta}^{(l)}_c= \boldsymbol{\beta}^{(l)}$. 
}{
\For{$c=1,...,C$}{  \noindent   Draw 
\begin{equation*}
 \boldsymbol{\beta}^{(l)}_c  \mid \boldsymbol{G}_{\boldsymbol{\varphi}^{(l-1)},c}, A_{\rho^{(l)},c}, \boldsymbol{\theta}^{\textsf{FH-SC}(l)}_{c}, \by \sim \text{Normal} \left(M_{\bb_c^{(l)}}, V_{\boldsymbol{\bb}_c^{(l)}}\right),
\end{equation*}
\begin{align*}
V_{\boldsymbol{\bb}_c^{(l)}}&=(\boldsymbol{X}_c^\textsf{T} (\boldsymbol{Z}_c\boldsymbol{G}_{\boldsymbol{\varphi}^{(l-1)},c}\boldsymbol{Z}_c^\textsf{T})^{-1}\boldsymbol{X}_c)^{-1} , &
M_{\bb_c^{(l)}} &=V_{\boldsymbol{\bb}_c^{(l)}} ( \boldsymbol{X}_c^\textsf{T}A_{\rho^{(l)},c}(\boldsymbol{Z}_c\boldsymbol{G}_{\boldsymbol{\varphi}^{(l-1)},c}\boldsymbol{Z}_c^\textsf{T})^{-1}\boldsymbol{\theta}^{\textsf{FH-SC}(l)}_c),
\end{align*}
and update $\boldsymbol{\delta}^{(l)}_c= \boldsymbol{\beta}^{(l)}_c$. 
}
}
\end{enumerate}
\begin{enumerate}[(2)] 
\item \eIf{$\boldsymbol{G}_{\boldsymbol{\varphi},c}=\sigma^{2} \mathbb I_{n_{c}}$}{
Draw
\begin{equation*}
1/\sigma^{2(l)} \mid  \boldsymbol{G}_{\boldsymbol{\varphi}^{(l-1)},c}, \boldsymbol{A}_{\rho^{(l)}}, \boldsymbol{\delta}^{(l)}_c, \boldsymbol{\theta}^{\textsf{FH-SC}(l)}, \by \sim \text{Gamma} \left(\frac{1}{2}m + a_{(1/\sigma^{2})}, \frac{1}{2}\text{SSB}_{u}^{(l)}+ b_{(1/\sigma^{2})}  \right) ,
\end{equation*}
 where 
$ \text{SSB}_{u}^{(l)} = (A_{\rho^{(l)}}\boldsymbol{\theta}^{\textsf{FH-SC}(l)} - \boldsymbol{X}\boldsymbol{\delta}^{(l)}_c)^\textsf{T}(\boldsymbol{Z}\boldsymbol{Z}^\textsf{T})^{-1}(A_{\rho^{(l)}}\boldsymbol{\theta}^{\textsf{FH-SC}(l)} - \boldsymbol{X}\boldsymbol{\delta}^{(l)}_c)$ and update $\boldsymbol{G}_{\boldsymbol{\varphi}^{(l)},c}=\sigma^{2(l)}\mathbb I_{n_{c}}$. 
}{
\eIf{$\boldsymbol{G}_{\boldsymbol{\varphi},c}=\sigma^{2}_c\mathbb I_{n_{c}}$}{
\For{$c=1,...,C$}{  \noindent  
Draw
\begin{equation*}
1/\sigma_c^{2(l)} \mid  \boldsymbol{G}_{\boldsymbol{\varphi}^{(l-1)},c}, A_{\rho^{(l)},c}, \boldsymbol{\delta}^{(l)}_c, \boldsymbol{\theta}_c^{\textsf{FH-SC}(l)}, \by \sim \text{Gamma} \left(\frac{1}{2}n_c + a_{(1/\sigma^{2}_c)}, \frac{1}{2}\text{SSB}_{u,c}^{(l)}+ b_{(1/\sigma^{2}_c)} \right),
\end{equation*}
 where 
$ \text{SSB}_{u,c}^{(l)} = (A_{\rho^{(l)},c}\boldsymbol{\theta}_c^{\textsf{FH-SC}(l)} - \boldsymbol{X}_c\boldsymbol{\delta}^{(l)}_c)^\textsf{T}(\boldsymbol{Z}_c\boldsymbol{Z}_c^\textsf{T})^{-1}(A_{\rho^{(l)}}\boldsymbol{\theta}_c^{\textsf{FH-SC}(l)} - \boldsymbol{X}_c\boldsymbol{\delta}^{(l)}_c)$ and update $\boldsymbol{G}_{\boldsymbol{\varphi}^{(l)},c}=\sigma_c^{2(l)}\mathbb I_{n_{c}}$. }
}{
\If{$\boldsymbol{G}_{\boldsymbol{\varphi},c}=\text{diag}_{n_c+1}(\hat{\gamma},\boldsymbol{1}_{n_c})\sigma_c^{2}$}{
\For{$c=1,...,C$}{  \noindent   
Draw
\begin{equation*}
1/\sigma_c^{2(l)} \mid  \boldsymbol{G}_{\boldsymbol{\varphi}^{(l-1)},c}, A_{\rho^{(l)},c}, \boldsymbol{\delta}^{(l)}_c, \boldsymbol{\theta}_c^{\textsf{FH-SC}(l)}, \by \sim \text{Gamma} \left(\frac{1}{2}n_c + \frac{1}{2}+  a_{(1/\sigma^{2}_c)}, \frac{1}{2}\text{SSB}_{u,c}^{(l)}+  b_{(1/\sigma^{2}_c)}\right),
\end{equation*}
where 
$ \text{SSB}_{u,c}^{(l)} = (A_{\rho^{(l)},c}\boldsymbol{\theta}_c^{\textsf{FH-SC}(l)} - \boldsymbol{X}_c\boldsymbol{\delta}^{(l)}_c)^\textsf{T}(\boldsymbol{Z}_c
\boldsymbol{H}_{\hat{\gamma}}
\boldsymbol{Z}_c^\textsf{T})^{-1}(A_{\rho^{(l)},c}\boldsymbol{\theta}_c^{\textsf{FH-SC}(l)} - \boldsymbol{X}_c\boldsymbol{\delta}^{(l)}_c)$ 
where $\boldsymbol{H}_{\hat{\gamma}}= \text{diag}_{n_c+1}(\hat{\gamma},1,...,1)$ and update $\boldsymbol{G}_{\boldsymbol{\varphi}^{(l)},c}=\text{diag}_{n_c+1}(\hat{\gamma},\boldsymbol{1}_{n_c})\sigma_c^{2(l)}$. }}}}
\end{enumerate}
\State \textbf{Output: ($\boldsymbol{G}_{\boldsymbol{\varphi}^{(l)}}= \text{blkdiag}(\{\boldsymbol{G}_{\boldsymbol{\varphi}^{(l)},c}\}_{c=1}^{C})$, $\boldsymbol{\delta}^{(l)}_c)$}      
\end{algorithmic}
 \caption{\small Gibbs sampling steps for $\boldsymbol{G}_{\boldsymbol{\varphi}^{(l)},c}$ and $\boldsymbol{\delta}^{(l)}_c$  for the FH-SC models in Table \ref{tab:models}.}
 \label{alg:MCMC2}
\end{algorithm}

\clearpage


\section{Simulation studies}
\label{sec:simula}

To illustrate the performance of the proposed FH-SC model given in   (\ref{fay_herriot_s}) and the CPMSE given in Proposition  \ref{prop:def3}, we implement two simulation studies. In Section \ref{sec:simulations_CPMSE},  we consider
a model-based simulation study using the 
the FH model
\citep{fay_1979}, and compare the results using different values of correlations 
between direct estimates and covariates. In Section \ref{subsec:simulations}, we consider a data-based
simulation study using the same covariates from our motivating case study where the generating model is the proposed FH-SC$_{1}$  model. 
In Section \ref{sec:simulations_CPMSE}, we show that the proposed measure of uncertainty for benchmarked estimators, CPMSE, works well under the FH model and its precision increases with the values of correlation between
direct estimates and covariates. Importantly, in Section \ref{subsec:simulations}, we illustrate situations where 
the FH-SC$_{1}$ model performs better than the FH model and the proposed CPMSE is useful to approximate the MSE of the RB benchmarked
estimator under the FH-SC$_{1}$ model. 


To evaluate the performance of the CPMSE  given in Proposition \ref{prop:def3}  in both simulation studies,  
we calculate the average of the CPMSE  of $\hat{\theta}_i^{\mathcal{M}\textsf{-B}}$ 
 across the simulated data sets for  each small area, as follows
 \begin{align}
   \label{eq:CPMSE}
 \widehat{\text{CPMSE}}(\hat{\theta}_i^{\mathcal{M}\textsf{-B}}) &= \dfrac{1}{T}\sum_{t=1}^{T}(\text{CPMSE}( \hat{\theta}_i^{\mathcal{M}\textsf{-B}}))^{(t)}, 
\end{align}
and the MSE of $\hat{\theta}_i^{\mathcal{M}\textsf{-B}}$ given by 
 \begin{align}
   \label{eq:MSE}
 \widehat{\text{MSE}}(\hat{\theta}_i^{\mathcal{M}\textsf{-B}}) &= \dfrac{1}{T}\sum_{t=1}^{T}(\hat{\theta}_i^{\mathcal{M}\textsf{-B}}-  \theta_{i}^{(t)})^{2},
\end{align}
where $t=1,...,T$ are the simulated data sets. Also, to measure the quality of the approximation of the MSE provided by our proposed CPMSE, we consider the average  across the small areas of  (\ref{eq:CPMSE}) and (\ref{eq:MSE}), given by  
 \begin{align}
   \label{eq:CPMSE_MSE}
\widehat{\text{CPMSE}}^{\textsf{Avg}-\mathcal{M}} &= \dfrac{1}{m}\sum_{i=1}^{m}  \widehat{\text{CPMSE}}(\hat{\theta}_i^{\mathcal{M}\textsf{-B}}), &
\widehat{\text{MSE}}^{\textsf{Avg}-\mathcal{M}} &= \dfrac{1}{m}\sum_{i=1}^{m} \widehat{\text{MSE}}(\hat{\theta}_i^{\mathcal{M}\textsf{-B}}).
\end{align}


\subsection{Model-based simulation study under the Fay-Herriot model}

 \label{sec:simulations_CPMSE} 
 
In this section, we consider the following \cite{fay_1979} model for our simulations:
\begin{align}
   \label{eq:FH_model}
   y_{i} &= \theta_{i} + e_{i}, \\
   \theta_{i} &= \beta_{0} + \beta_{1}x_{i} + u_{i}, \notag
\end{align}
where $u_{i} \overset{\mathrm{iid}}{\sim}  N(0,\sigma^2)$ and $e_{i} \overset{\mathrm{iid}}{\sim} N(0,D_{i})$ for $i=1,...,m$. We consider different numbers of small areas
$m = \{50, 100, 500\}$ and the design matrix $\boldsymbol{X}$ has a column of ones and
one explanatory variable, $x_{i}  \overset{\mathrm{iid}}{\sim}  \text{Uniform}(0,1)$. We set equally spaced values of $D_{i}$ from 0.1 to 1. We specify $\beta_{0} =1$ and 
$\sigma^2=0.25$, and consider different values of correlation between $\boldsymbol{x}$ and $\boldsymbol{y}$, $\text{cor}(\boldsymbol{x}, \boldsymbol{y})=\{0.2, 0.4, 0.8 \}$, by setting $\beta_{1} = \sqrt{(12(\bar{D}+\sigma^2))/(1-\text{cor}(\boldsymbol{x}, \boldsymbol{y})^2)}$ with $\bar{D}=(1/m)\sum_{i=1}^{m}D_{i}$. We simulate $T=100$ data sets generating $ \theta_{i}$ and $ y_{i}$ from model  (\ref{eq:FH_model}) 
with the same  specification of $\boldsymbol{X}$ for a given $m$. 
 
 \begin{figure}[ht]
\begin{center}
\begin{tabular}{ccc}
 \hspace{-0.5cm} $\text{cor}(\boldsymbol{x}, \boldsymbol{y})=0.2$  &   \hspace{0.5cm} $\text{cor}(\boldsymbol{x}, \boldsymbol{y})=0.4$ &  \hspace{0.5cm}  $\text{cor}(\boldsymbol{x}, \boldsymbol{y})=0.8$ \\
  \hspace{-1.0cm}  \includegraphics[width=0.35\textwidth]{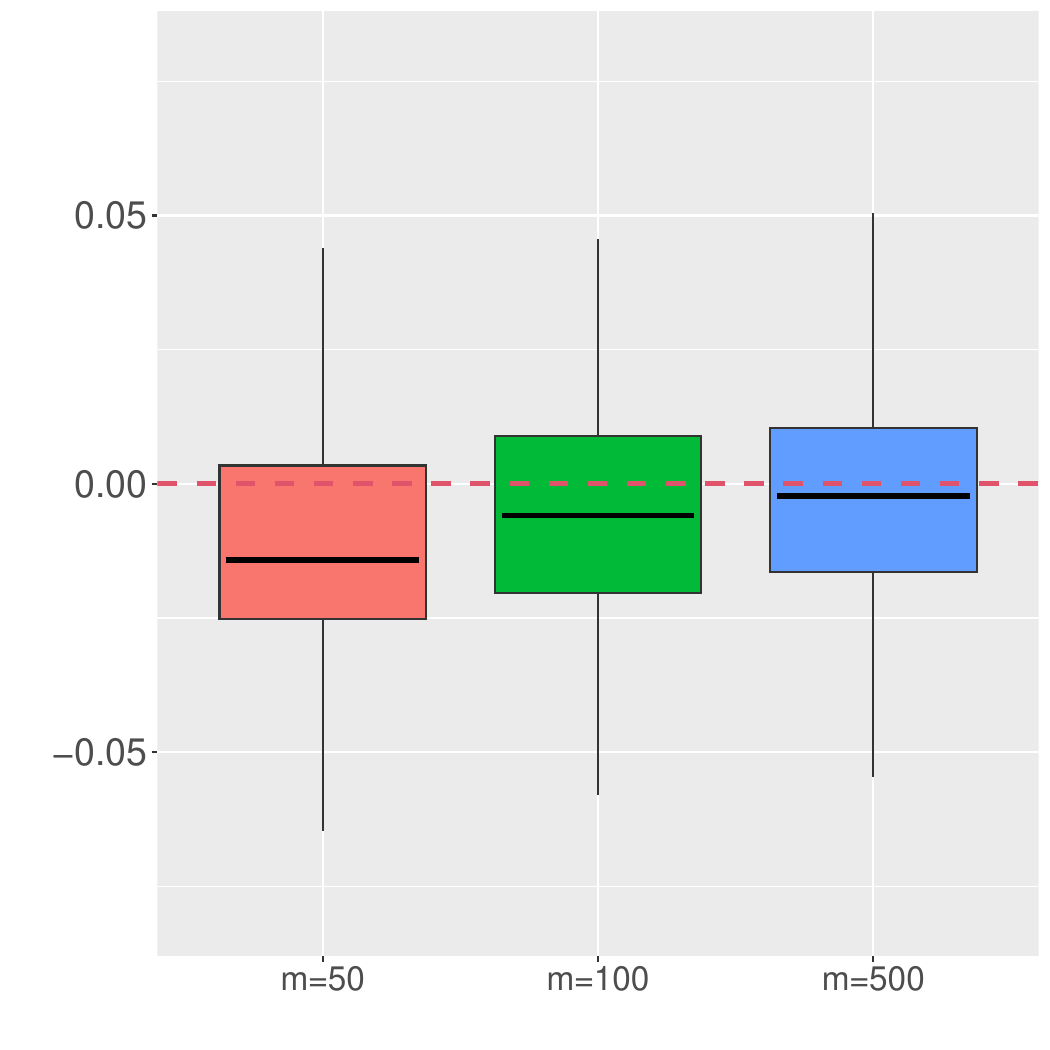}   & \includegraphics[width=0.35\textwidth]{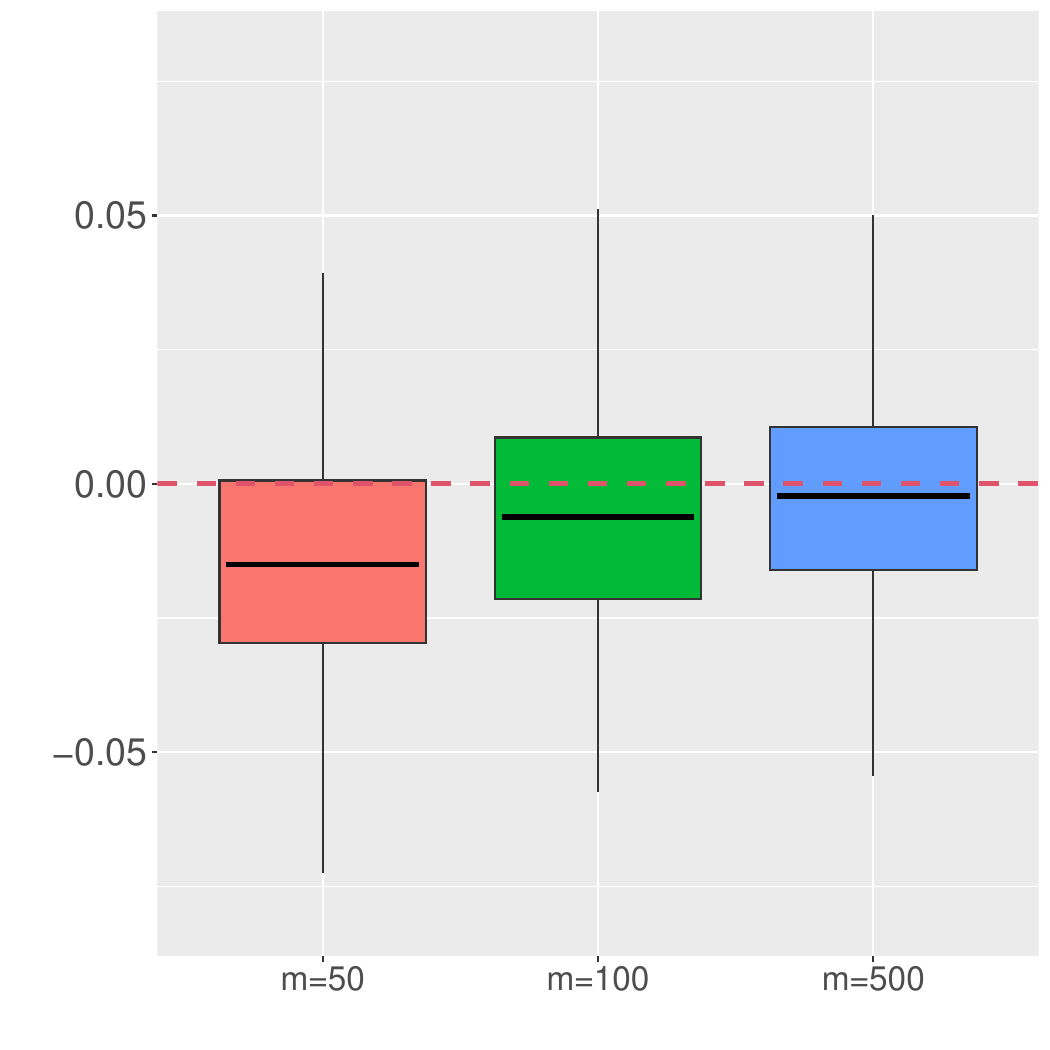}   &
\includegraphics[width=0.35\textwidth]{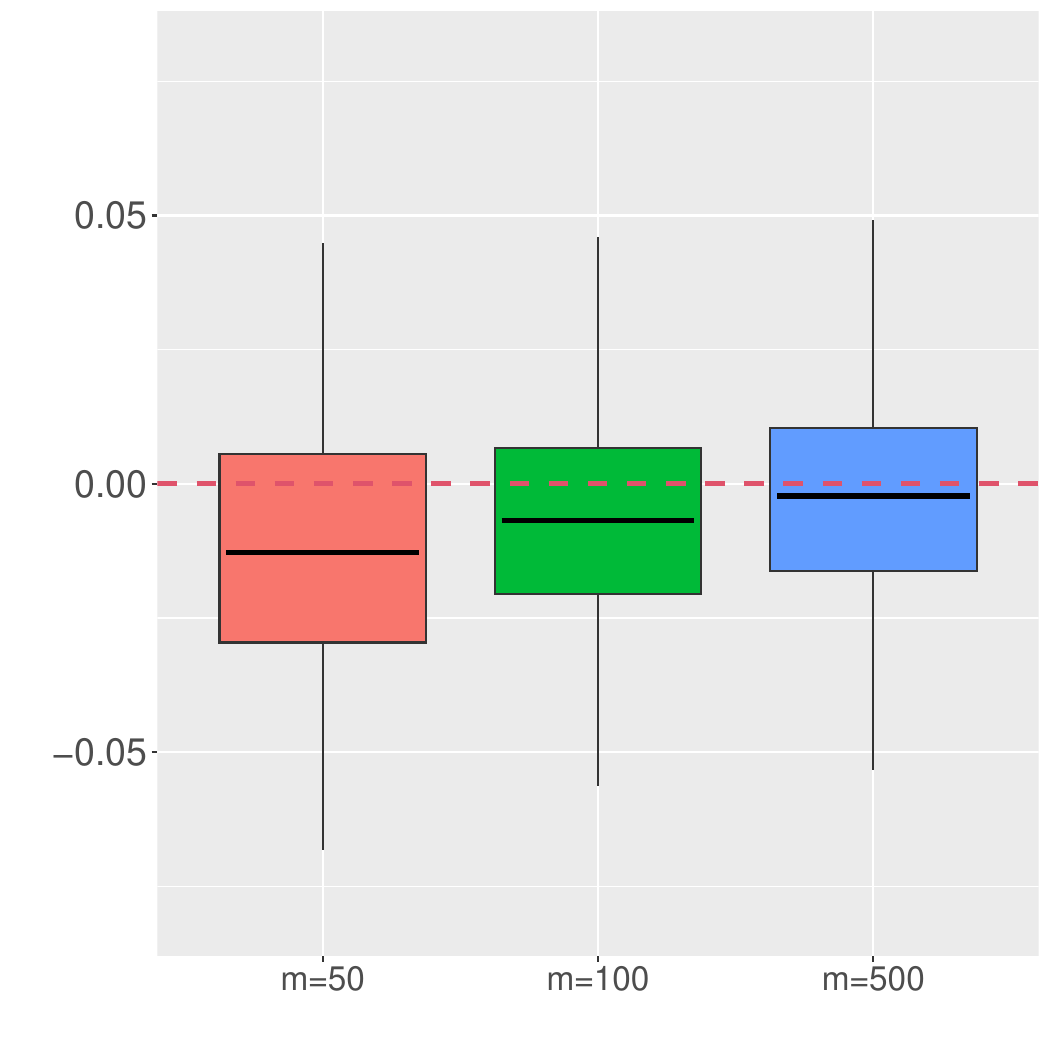} \\
\end{tabular}
\end{center}
\vspace{-0.5cm}
\textit{\caption{
 $\widehat{ \emph{CPMSE}}(\hat{\theta}_{i}^{\emph{FH-B}}) - \widehat{\emph{MSE}}(\hat{\theta}_{i}^{\emph{FH-B}})$  for correlation values $\text{cor}(\boldsymbol{x}, \boldsymbol{y})=\{0.2, 0.4, 0.8 \}$ and number of small areas $m = \{50, 100, 500\}$. The red lines display the average of the differences,  $\widehat{\emph{CPMSE}}(\hat{\theta}_{i}^{\emph{FH-B}}) - \widehat{\emph{MSE}}(\hat{\theta}_{i}^{\emph{FH-B}})$, across the small areas.
\label{fig:CPMSE}}}
\end{figure}

{\renewcommand{\arraystretch}{4}
\begin{table}[h]  
\resizebox{1.1\textwidth}{!}{
\begin{tabular}{||cc||c||c|c|c||c||c|c|c||c||c|c|c||ccc} 
  \hline 
 & $\text{cor}(\boldsymbol{x}, \boldsymbol{y})$ & $m$ & $\widehat{\text{MSE}}^{\textsf{Avg}}$ & $\widehat{\text{CPMSE}}^{\textsf{Avg}}$ & $\mid$  $\text{Diff}^{\textsf{Avg}}$ $\mid$ & $m$ &  $\widehat{\text{MSE}}^{\textsf{Avg}}$ & $\widehat{\text{CPMSE}}^{\textsf{Avg}}$ & $\mid$  $\text{Diff}^{\textsf{Avg}}$ $\mid$ & $m$ &  $\widehat{\text{MSE}}^{\textsf{Avg}}$ & $\widehat{\text{CPMSE}}^{\textsf{Avg}}$ & $\mid$  $\text{Diff}^{\textsf{Avg}}$ $\mid$ \\ 
  \hline \hline 
 & 0.2 & 50 & 0.1852 & 0.1994 & 0.0142 & 100 & 0.1727 & 0.1793 & 0.0066 & 500 & 0.1635 & 0.1658 & 0.0023 \\   \hline
   & 0.4 & 50 & 0.1836 & 0.1986 & 0.0150 & 100 & 0.1722 & 0.1792 & 0.0070 & 500 & 0.1636 & 0.1659 & 0.0023 \\   \hline
   & 0.8 & 50 & 0.1863 & 0.1991 & 0.0128 & 100 & 0.1721 & 0.1798 & 0.0077 & 500 & 0.1636 & 0.1658 & 0.0022 \\ 
  \hline
  \end{tabular}
}
\caption{$\widehat{\text{CPMSE}}^{\textsf{Avg}}$, $\widehat{\text{MSE}}^{\textsf{Avg}}$ and the differences, $\text{Diff}^{\textsf{Avg}}=\widehat{\text{CPMSE}}^{\textsf{Avg}} - \widehat{\text{MSE}}^{\textsf{Avg}}$,  in absolute value under the FH model of the benchmarked estimator,  for  $\text{cor}(\boldsymbol{x}, \boldsymbol{y})=\{0.2, 0.4, 0.8 \}$ and $m = \{50, 100, 500\}$.} 
\label{table:CPMSE1}
\end{table}
}

 Figure \ref{fig:CPMSE} illustrates the difference between $\widehat{ \text{CPMSE}}(\hat{\theta}_{i}^{\text{FH-B}})$ and $ \widehat{\text{MSE}}(\hat{\theta}_{i}^{\text{FH-B}})$ for different values of the number of small areas (50, 100 and 500) and correlations (0.2, 0.4 and 0.8). We observe that $\widehat{ \text{CPMSE}}(\hat{\theta}_{i}^{\text{FH-B}})$ provides a good approximation of the MSE for the benchmarked estimators under FH model, as the difference between $\widehat{ \text{CPMSE}}(\hat{\theta}_{i}^{\text{FH-B}})$  and $ \widehat{\text{MSE}}(\hat{\theta}_{i}^{\text{FH-B}})$ 
is small for $i=1,...,m$. More importantly, the quality of the  approximation increases with the number of small areas and remains unaffected for the different values of correlation.  Table \ref{table:CPMSE1}  
displays, for the different number of small areas and  values of correlation, the average across the small areas of  (\ref{eq:CPMSE}) and (\ref{eq:MSE}) under FH model, $\widehat{\text{CPMSE}}^{\textsf{Avg-}\text{FH-B}}$ and $\widehat{\text{MSE}}^{\textsf{Avg-}\text{FH-B}}$, and the difference between those measures in absolute value.
The results in Table \ref{table:CPMSE1} support the results in Figure \ref{fig:CPMSE},  showing that CPMSE produces accurate estimates of the MSE and their precision improves as the number of small areas increases.

\subsection{Data-based simulation study under the FHSC$_{1}$ model}

\label{subsec:simulations}


In this section, we consider the data-based simulation study. We use the proposed FH-SC$_{1}$ given in  Table \ref{tab:models}  with the prior specification discussed in Section \ref{subsec:The_MCMC}. The generating model is given by 

\begin{align}
\label{fay_herriot2_sim}
\begin{split}
    \by_c & \sim \text{Normal}(A_{\rho,c}^{-1}\bdt_c,  \boldsymbol{D}_c ),\\ 
\bdt_c &\sim  \text{Normal}\left(\boldsymbol{X}_c\boldsymbol{\delta}_c,  \boldsymbol{Z}_c\boldsymbol{G}_{\boldsymbol{\varphi},c}\boldsymbol{Z}_c^\textsf{T} \right),\\
\rho &\sim \text{Beta}( a,b),\\
\pi(\boldsymbol{\delta}_c,\boldsymbol{G}_c ,\rho) &= \prod_{c=1}^{C} \pi(\boldsymbol{G}_{\boldsymbol{\varphi},c})\pi(\rho).
\end{split}
\end{align}

We use the same covariates of the case study, a vector of ones and  the index of illiteracy obtained from the 2014 Census of Agriculture \citep{CNA2014}.
We consider the matrix  $A_{\rho}$ as defined in Section \ref{sec:application}.
To set the parameters in this simulation study, we consider the posterior mean estimates of the regression parameters obtained
under the FH model. Specifically, we set $\beta_{0}=0.5$ , $\beta_{1}=-0.01$ and  $\sigma^2_{u}=0.2$. Similar to our case study, we assume that 
the proportion of households with internet connection at municipality-levels needs to be agregated with the national level value, according to  \cite{internet}. More formally, the benchmarking constraints are given by $\sum_{c=1}^{C}\sum_{j=1}^{n_c} w_{j,c} \hat{\theta}_{j,c} = 0.418$.  We use three different values for the cluster regularization penalty given by $\rho=\{0.8,0.9,0.99\}$. To evaluate the performance of the FH-SC$_{1}$ and  FH models, we simulate 100 data sets and compute the average absolute deviation (AAD) and average square deviation (ASD) as follows,
\begin{align}
\text{AAD} &= \dfrac{1}{m} \sum_{c=1}^{C} \sum_{j=1}^{n_c} \mid \hat{\theta}_{j,c}^{\M}
 - \theta_{j,c} \mid, 
& \text{ASD} &= \dfrac{1}{m} \sum_{c=1}^{C} \sum_{j=1}^{n_c} ( \hat{\theta}_{j,c}^{\M}
 - \theta_{j,c} )^{2},
\end{align}
where $\theta_{j,c}$ is estimated by its posterior mean $\hat{\theta}_{j,c}^{\M}$.
\begin{figure}[ht]
\begin{center}
\begin{tabular}{ccc}
\includegraphics[width=0.8\textwidth]{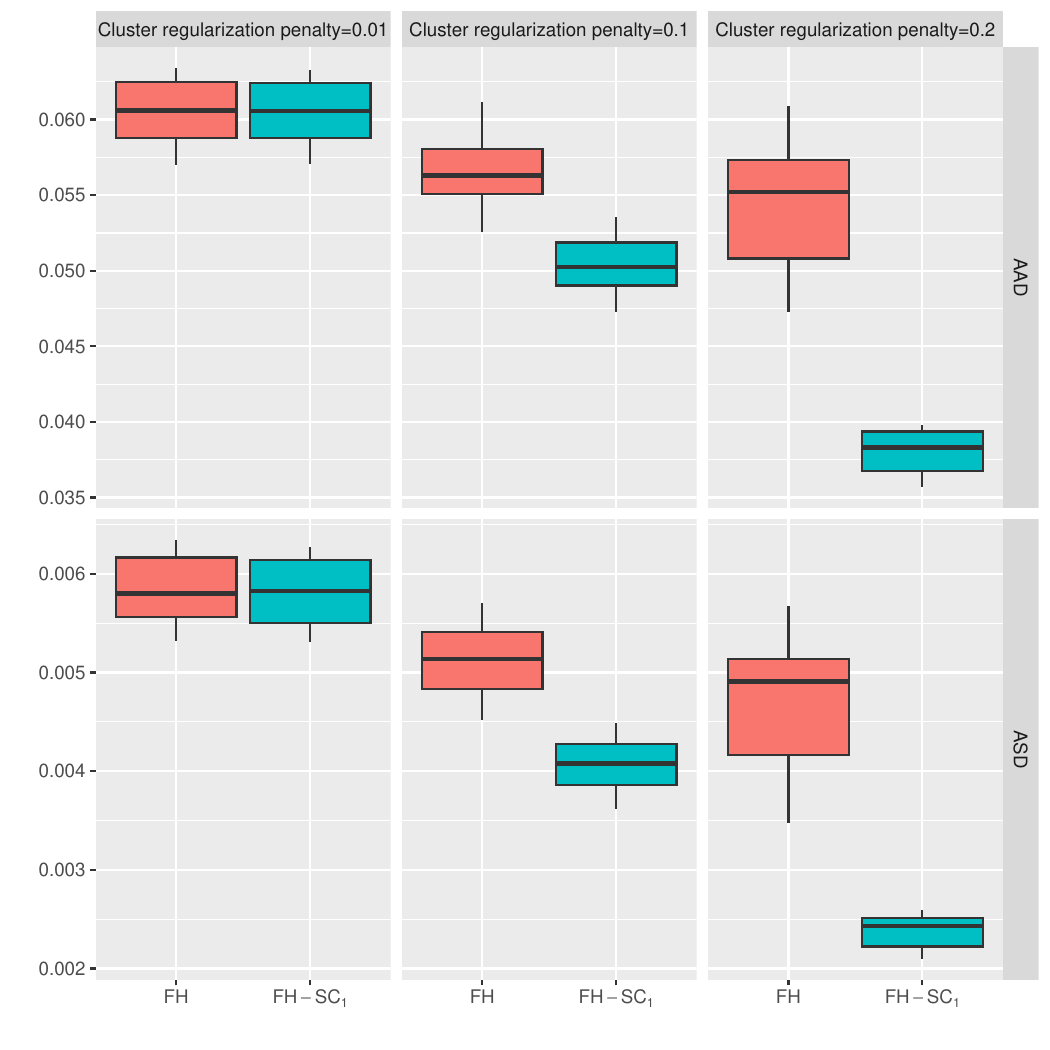} 
\end{tabular}
\end{center}
\caption{AAD and ASD measures  for the RB estimators under the FH  and FH-SC$_{1}$ models.  \label{fig:fig1}}
\end{figure}

Figure \ref{fig:fig1} illustrates that the posterior  estimates under FH-SC$_{1}$ produce smaller AAD and  ASD values compared to those produced under the FH model. 
More specifically, the efficacy of the FH-SC$_{1}$ model increases as $\rho$ decreases, yielding RB estimates
that are closer to the true mean, $\theta_{i}$. In other words, when the cluster classification of covariates is important, the RB estimates under FH-SC$_1$ model outperform the RB estimates under the FH model.  To assess the performance of  the CPMSE, we consider  $ \widehat{\text{CPMSE}}(\hat{\theta}_{j,c}^{\mathcal{M}\textsf{-B}})$,  $\widehat{\text{MSE}}(\hat{\theta}_{j,c}^{\mathcal{M}\textsf{-B}})$ as in (\ref{eq:CPMSE}) and (\ref{eq:MSE}), and $\widehat{\text{CPMSE}}^{\textsf{Avg}-\mathcal{M}}$ and $\widehat{\text{CPMSE}}^{\textsf{Avg}-\mathcal{M}}$  as in  (\ref{eq:CPMSE_MSE}).  When the generating model is the proposed FH-SC$_{1}$, according to Figure \ref{fig:CPMSE1},
the CPMSE provides a good approximation of the MSE of the RB estimator under the FH-SC$_{1}$ model and the MSE
of the RB benchmarked estimator under the FH and FH-SC$_{1}$ models, for $i=1,...,m$. The good approximations of the MSE given by the proposed CPMSE are also confirmed with the results of Table  \ref{table:CPMSE2}, where the average across the small areas of the differences in absolute value are very small and consistent  for the different cluster regularization penalty values, $\rho=\{0.8,0.9,0.99\}$.

{\renewcommand{\arraystretch}{4.5}
\begin{table}[h]  
\resizebox{1.0\textwidth}{!}{
\begin{tabular}{||cc||c||ccc||c||ccc||c||ccc||ccc} 
  \hline  \hline
 & $\rho$ & Estimator & $\widehat{MSE}^{\textsf{Avg}}$ & $\widehat{CPMSE}^{\textsf{Avg}}$ & $\mid$ $\text{Diff}^{\textsf{Avg}}$ $\mid$ & Estimator & $\widehat{MSE}^{\textsf{Avg}}$ & $\widehat{CPMSE}^{\textsf{Avg}}$ & $\mid$ $\text{Diff}^{\textsf{Avg}}$ $\mid$ & Estimator & $\widehat{MSE}^{\textsf{Avg}}$ & $\widehat{CPMSE}^{\textsf{Avg}}$ & $\mid$ $\text{Diff}^{\textsf{Avg}}$ $\mid$ \\ 
  \hline 
 & 0.01& $\hat{\theta}_{j,c}^{\textsf{FH-SC}_{1}}$ & 0.0058 & 0.0060 & 0.0002 & $\hat{\theta}_{j,c}^{\textsf{FH-B}}$ & 0.0111 & 0.0113 & 0.0002 & $\hat{\theta}_{j,c}^{\textsf{FH-SC}_{1}-\textsf{B}}$ & 0.0126 & 0.0129 & 0.0003 \\ 
 & 0.1 & $\hat{\theta}_{j,c}^{\textsf{FH-SC}_{1}}$ & 0.0041 & 0.0042 & 0.0001 & $\hat{\theta}_{j,c}^{\textsf{FH-B}}$ & 0.0083 & 0.0086 & 0.0003 & $\hat{\theta}_{j,c}^{\textsf{FH-SC}_{1}-\textsf{B}}$ & 0.0177 & 0.0181 & 0.0004 \\ 
 & 0.2 & $\hat{\theta}_{j,c}^{\textsf{FH-SC}_{1}}$ & 0.0023 & 0.0025 & 0.0002 & $\hat{\theta}_{j,c}^{\textsf{FH-B}}$ & 0.0078 & 0.0082 & 0.0004 & $\hat{\theta}_{j,c}^{\textsf{FH-SC}_{1}-\textsf{B}}$ & 0.0198 & 0.0203 & 0.0005 \\ 
    \hline
  \end{tabular}
}
\caption{$\widehat{\text{CPMSE}}^{\textsf{Avg}}$ and $\widehat{\text{MSE}}^{\textsf{Avg}}$  and the differences, $\text{Diff}^{\textsf{Avg}}=\widehat{\text{CPMSE}}^{\textsf{Avg}}-\widehat{\text{MSE}}^{\textsf{Avg}}$, in absolute value  for $\rho=\{0.01,0.1,0.2\}$.}
\label{table:CPMSE2}
\end{table}
}
  
\begin{figure}[t]
\begin{center}
\begin{tabular}{ccc}
 \hspace{-0.5cm} $\rho=0.01$  &   \hspace{0.5cm} $\rho=0.1$ &  \hspace{0.5cm}  $\rho=0.2$ \\
  \hspace{-1.0cm}  \includegraphics[width=0.35\textwidth]{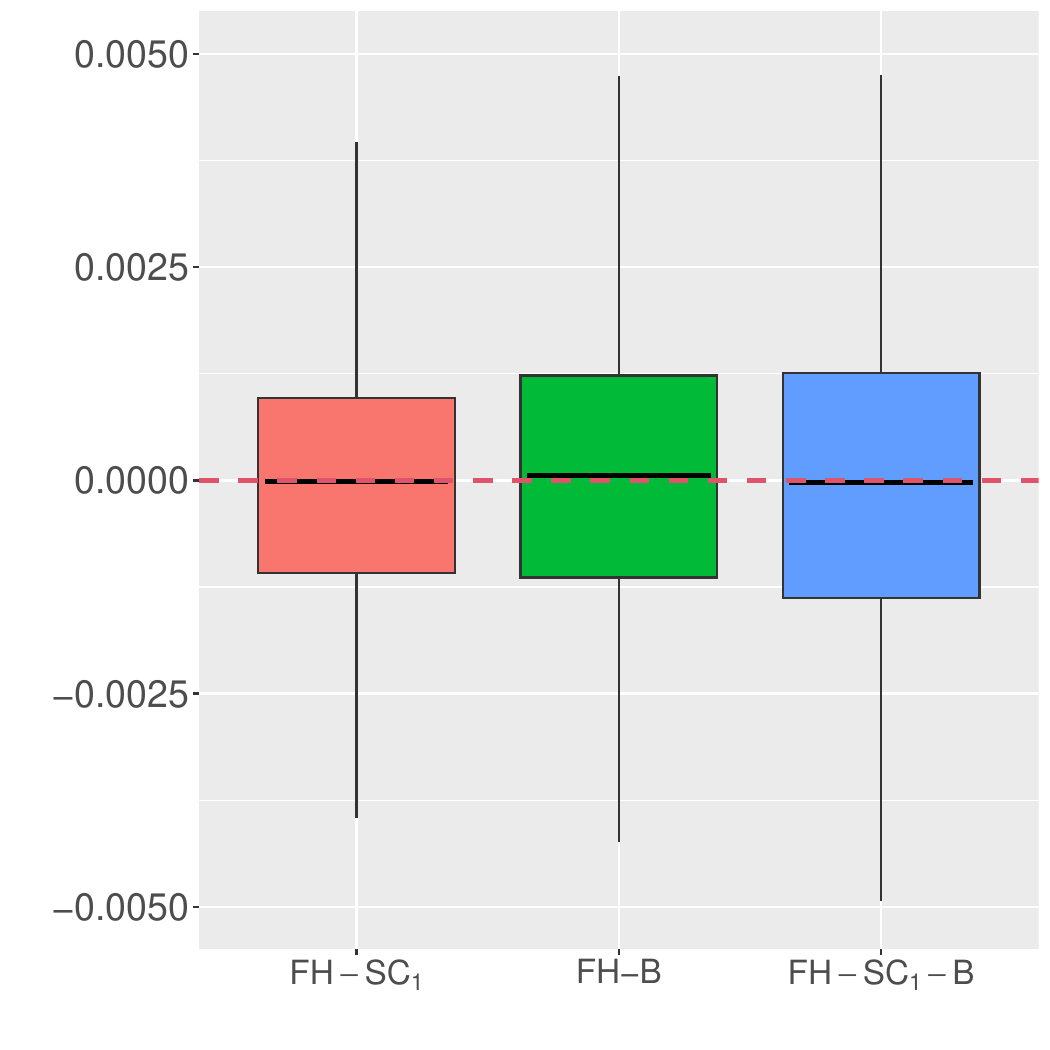}   & \includegraphics[width=0.35\textwidth]{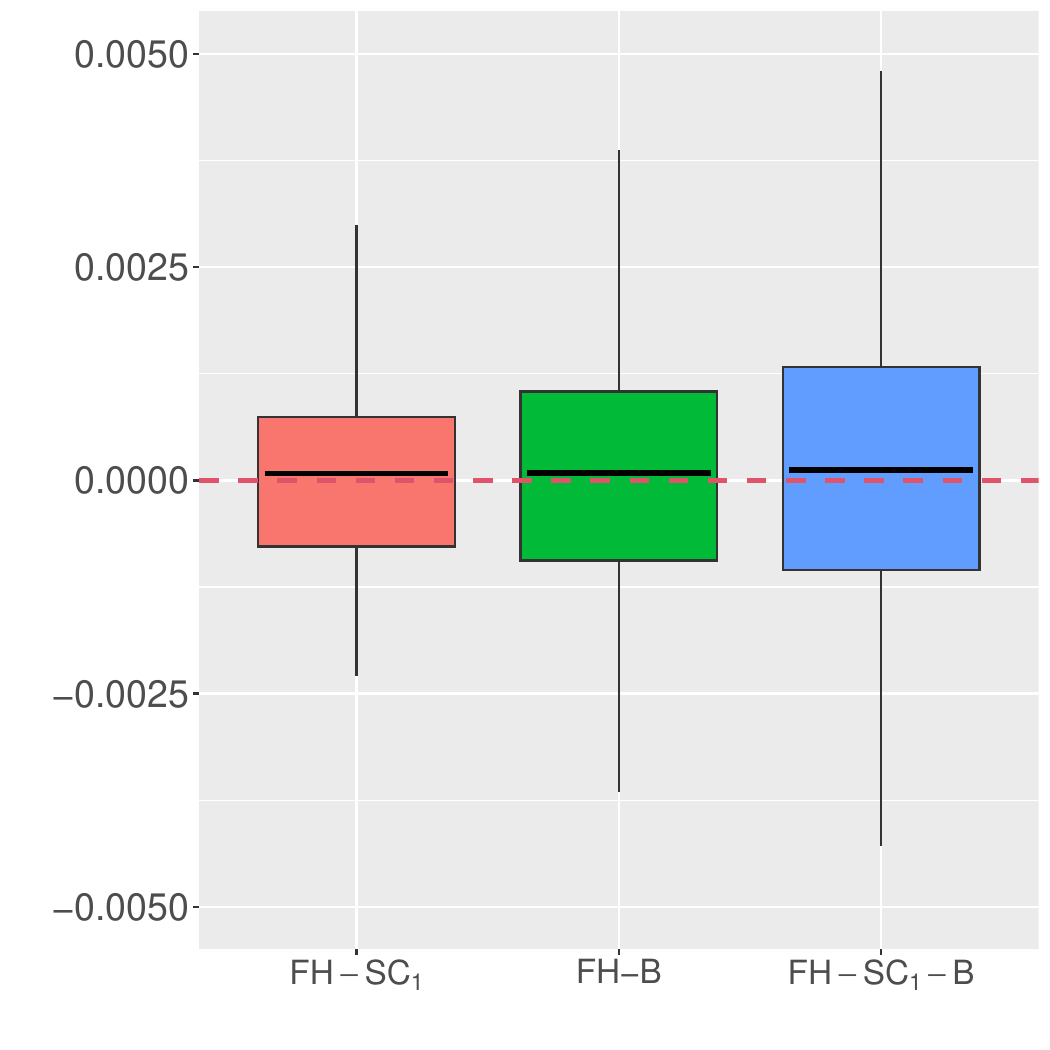}   &
\includegraphics[width=0.35\textwidth]{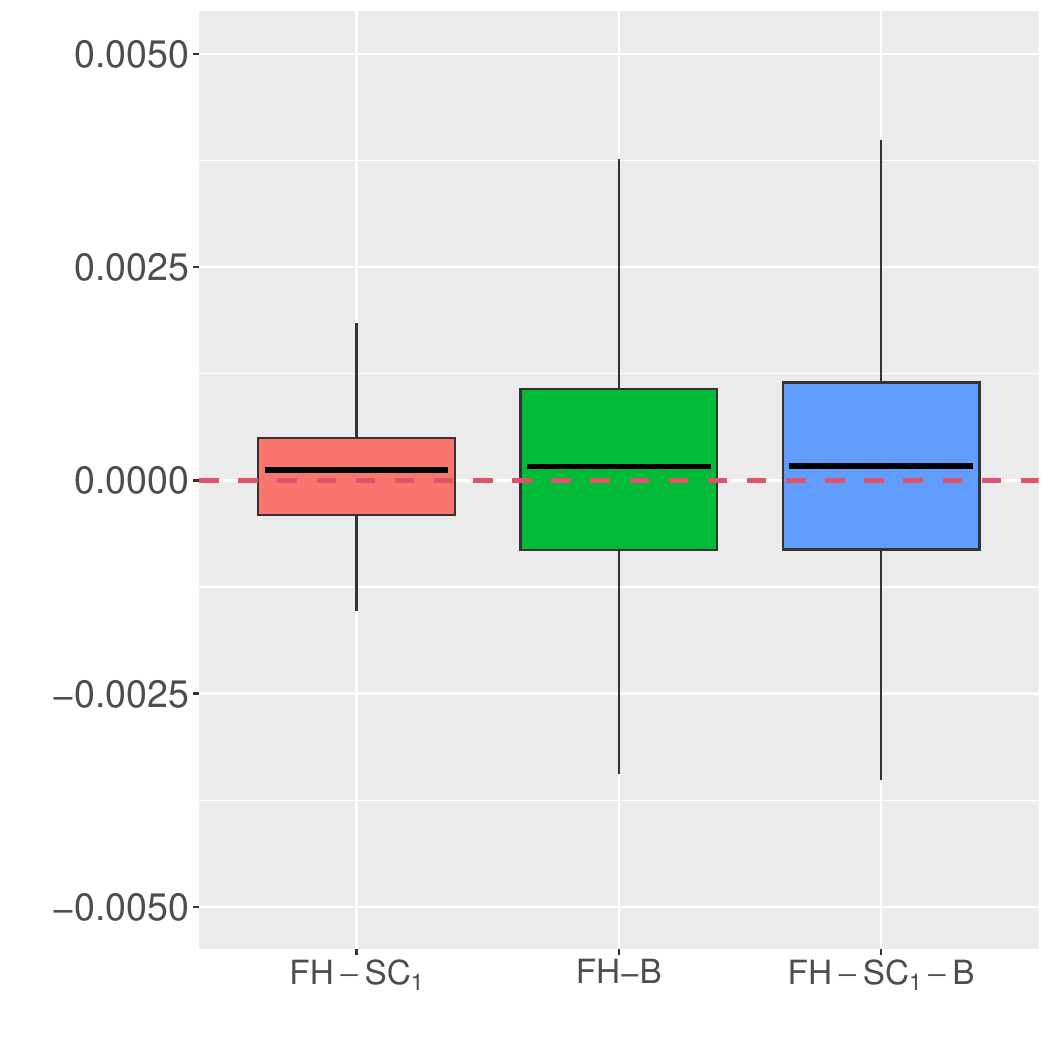} \\
\end{tabular}
\end{center}
\vspace{-0.5cm}
\caption{$\widehat{\text{CPMSE}}(\hat{\theta}_{j,c}^{\M\textsf{-B}}) - \widehat{\text{MSE}}(\hat{\theta}_{j,c}^{\M\textsf{-B}}) $
where  $\rho=\{0.01,0.1,0.2\}$. The red lines display the average of the differences, $\widehat{\text{CPMSE}}(\hat{\theta}_{j,c}^{\M\textsf{-B}}) - \widehat{\text{MSE}}(\hat{\theta}_{j,c}^{\M\textsf{-B}}) $, across the small areas.
\label{fig:CPMSE1}}
\end{figure}

\clearpage

\section{Supplemental material for the Case Study}
\label{applied}

In this section, we present the additional results of the Case Study.  Section \ref{sub_fig} contains the Supplemental Figures and Tables. Figure \ref{fig:figcve_ratio} shows the ratios between the coefficients of variation
produced by RB estimators under FH and  FH-SC$_{2}$  models and RB benchmarked estimators under FH and  FH-SC$_{2}$ models, respectively.  Table \ref{table:conf_cve} shows the results for the main capital cities and some relevant municipalities with high values of poverty and/or education deficit indexes in 2018 at the municipality level in Colombia. Figures \ref{fig:posterior_estimates1} and \ref{fig:posterior_estimates2} show  the posterior distributions of the RB  estimates
and the  95\% confidence intervals  constructed with the direct estimates and the direct variances. Figure \ref{fig:post1} displays  the spatial patterns of the different RB estimates of PHIA.  
Details of the sensitivity analysis are provided in Supplementary Section \ref{sub_applied2}.  Supplementary Section \ref{sub_applied4} is dedicated to discuss the convergence of  Algorithms \ref{alg:MCMC1} and  \ref{alg:MCMC2}. The assumptions on the error and random effects are provided in Supplementary Section \ref{sub_applied3}.

\subsection{Model Selection Criteria and Deviance Measures}
\label{sub_dic}


The EPD is computed using the average of a discrepancy measure between the direct estimates of PHIA ($y_{j,c}$)
and the posterior benchmarked draws ($\tilde{y}_{j,c}^{\textsf{FH-SC-B}}$) obtained from the posterior predictive distribution.  More specifically, 
{\small
\begin{align*}
\text{EPD}(\boldsymbol{y}, \tilde{\boldsymbol{y}}^{\textsf{FH-SC-B}}) = \dfrac{1}{L-T} \sum_{l=T+1}^{L}  \Delta(\boldsymbol{y}, \tilde{\boldsymbol{y}}^{\textsf{FH-SC-B}(l)}),
\end{align*}}
where $\Delta(\boldsymbol{y}, \tilde{\boldsymbol{y}}^{\textsf{FH-SC-B}(t)})$ is a deviance measure calculated for $l=T+1,...,L$ posterior samples. In this work, we consider the Average Absolute Deviation (AAD) and Average Square Deviation (ASD) as two options for the deviance measure:
{\small
\begin{align}
\Delta^{\text{AAD}} &= \sum_{c=1}^{C}\sum_{j=1}^{n_c}  | y_{j,c} - \tilde{y}_{j,c}^{\textsf{FH-SC-B}(l)}|,&
\Delta^{\text{ASD}} &=\sum_{c=1}^{C}\sum_{j=1}^{n_c}   (y_{j,c} - \tilde{y}_{j,c}^{\textsf{FH-SC-B}(l)})^{2}. \notag
\end{align}
In a similar fashion, the EPD can be applied to posterior predictive samples generated using RB small area estimates.
}

\clearpage

\subsection{Supplementary Figures and Tables}
\label{sub_fig}

\begin{figure}[ht]
\begin{center}
\begin{tabular}{ccc}
a)  $(\text{CV}(\hat{\theta}_{j,c}^{\textsf{FH-SC$_{2}$}})/\text{CV}(\hat{\theta}_{j,c}^{\textsf{FH}})<1)$\%=92.17 & \hspace{1cm} b)  $(\text{CV}(\hat{\theta}_{j,c}^{\textsf{FH-SC$_{2}$}\text{-B}})/\text{CV}(\hat{\theta}_{j,c}^{\textsf{FH}\text{-B}})<1)$\%= 85.03 \\
\includegraphics[width=0.50\textwidth]{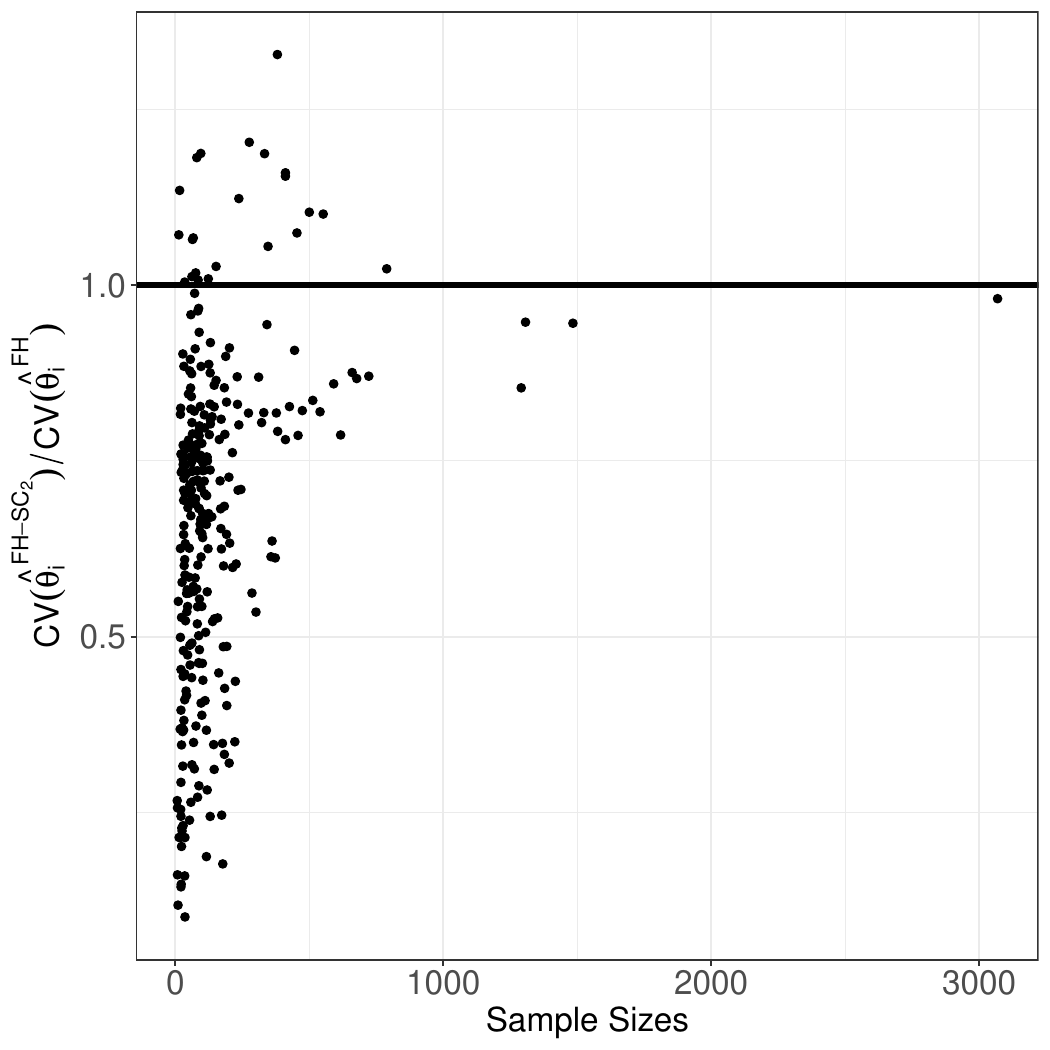} & \hspace{0.2cm}
\includegraphics[width=0.50\textwidth]{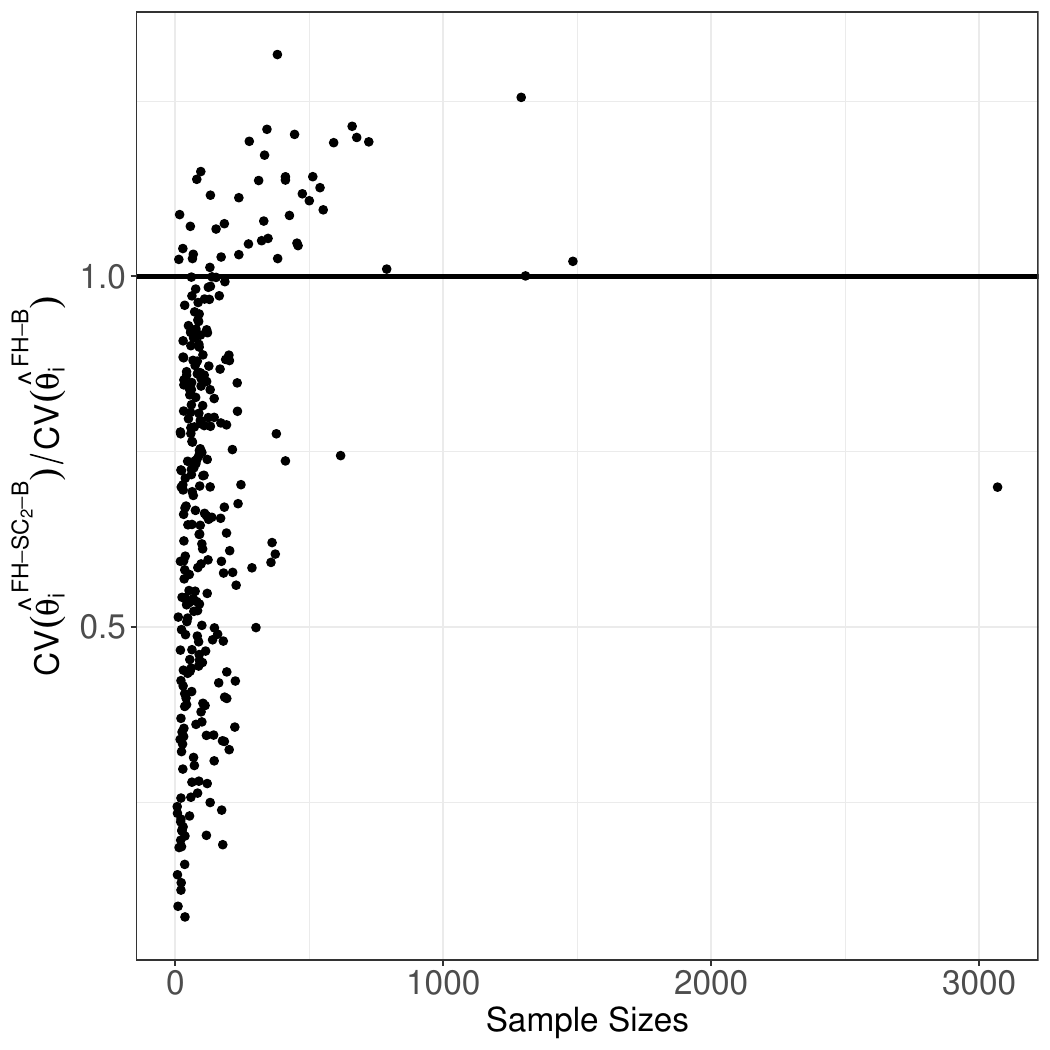} 
\end{tabular}
\end{center}
\caption{Ratios of Coefficient of Variations (CVs) produced by RB estimators under FH-SC$_{2}$ and FH models, respectively, and 
Ratios of CVs produced by RB benchmarked estimators  under FH-SC$_{2}$ and FH models, respectively, where, a)  $\text{CV}(\hat{\theta}_{j,c}^{\textsf{FH-SC$_{2}$}})/\text{CV}(\hat{\theta}_{j,c}^{\textsf{FH}})$ and b)  $\text{CV}(\hat{\theta}_{j,c}^{\textsf{FH-SC$_{2}$}\text{-B}})/\text{CV}(\hat{\theta}_{j,c}^{\textsf{FH}\text{-B}})$.}
\label{fig:figcve_ratio}
\end{figure}

\newpage

{\renewcommand{\arraystretch}{2.8}
\begin{table}[h]  
\begin{center}
\resizebox{1.0\textwidth}{!}{
\begin{tabular}{cccc>{\columncolor[gray]{0.6}}cc>{\columncolor[gray]{0.6}}cc>{\columncolor[gray]{0.9}}cc>{\columncolor[gray]{0.6}}cc>{\columncolor[gray]{0.9}}ccccc} 
    \hline
 & Deparment & Municipality & $\hat{\theta}_{j,c}^{\textsf{Direct}}$ &  $\text{CV}(\hat{\theta}_{j,c}^{\textsf{Direct}})$ & $\hat{\theta}_{j,c}^{\textsf{FH}}$  &  $\text{CV}(\hat{\theta}_{j,c}^{\textsf{FH}})$ & $\hat{\theta}_{j,c}^{\textsf{FH}\text{-B}}$ & $\text{CV}(\hat{\theta}_{j,c}^{\textsf{FH}\text{-B}})$ &  $\hat{\theta}_{j,c}^{\textsf{FH-SC$_{2}$}}$  & $\text{CV}(\hat{\theta}_{j,c}^{\textsf{FH-SC$_{2}$}})$ &  $\hat{\theta}_{j,c}^{\textsf{FH-SC$_{2}$}\text{-B}}$ & $\text{CV}(\hat{\theta}_{j,c}^{\textsf{FH-SC$_{2}$}\text{-B}})$ & MPI & EDI\\
  \hline  
   & Antioquia & Medellín & 0.60 & 8.98 & 0.56 & 9.09 & 0.66 & 16.40 & 0.56 & 8.61 & 0.65 & 16.41 & 12.80 & 35.20 \\ 
   & Bogotá, D.C. & Bogotá, D.C. &  0.59 & 5.89 & 0.58 & 5.90 & 0.83 & 30.47 & 0.57 & 5.79 & 0.72 & 21.31 & 9.00 & 26.20 \\ 
   & Bolívar & Cartagena & 0.48 & 13.38 & 0.44 & 13.51 & 0.47 & 14.18 & 0.46 & 11.71 & 0.53 & 16.98 & 19.90 & 33.50 \\  
   & La Guajira & Riohacha & 0.36 & 20.11 & 0.31 & 21.72 & 0.31 & 21.33 & 0.19 & 25.78 & 0.20 & 25.02 & 45.10 & 48.30 \\ 
   & Amazonas & Leticia & 0.26 & 24.40 & 0.26 & 23.28 & 0.26 & 23.21 & 0.18 & 25.01 & 0.19 & 24.31 & 48.40 & 45.10 \\ 
   & Valle del Cauca & Cali & 0.57 & 8.88 & 0.54 & 8.91 & 0.63 & 16.17 & 0.54 & 8.42 & 0.63 & 16.52 & 11.90 & 33.00 \\
   & Guaviare & San José Del Guaviare & 0.22 & 24.30 & 0.22 & 23.24 & 0.22 & 23.05 & 0.17 & 23.78 & 0.18 & 23.29 & 42.10 & 55.80 \\ 
   & Amazonas & Leticia & 0.26 & 24.40 & 0.26 & 23.28 & 0.26 & 23.21 & 0.18 & 25.01 & 0.19 & 24.31 & 48.40 & 45.10 \\ 
   & Valle del Cauca & Jamundí & 0.45 & 20.45 & 0.40 & 20.14 & 0.41 & 19.97 & 0.44 & 15.46 & 0.50 & 18.47 & 14.90 & 38.60 \\
   & Caquetá & Florencia &  0.26 & 23.79 & 0.25 & 22.52 & 0.26 & 22.08 & 0.16 & 24.80 & 0.17 & 24.18 & 29.60 & 47.80 \\ 
   & Cesar & Valledupar & 0.29 & 21.93 & 0.27 & 22.46 & 0.28 & 21.91 & 0.18 & 24.79 & 0.20 & 24.28 & 30.50 & 37.80 \\ 
   & Valle del Cauca & Buenaventura & 0.32 & 23.82 & 0.29 & 23.61 & 0.30 & 23.10 & 0.18 & 26.51 & 0.20 & 25.69 & 41.00 & 48.30 \\
   & Antioquia & Apartadó & 0.31 & 27.24 & 0.28 & 27.34 & 0.28 & 26.89 & 0.35 & 18.46 & 0.41 & 21.93 & 28.00 & 47.10 \\ 
   & Boyacá & Chiquinquirá & 0.31 & 28.19 & 0.29 & 26.46 & 0.29 & 26.20 & 0.36 & 18.42 & 0.42 & 21.67 & 17.30 & 44.70 \\ 
   & Caldas & Chinchiná & 0.37 & 24.14 & 0.33 & 23.50 & 0.34 & 23.37 & 0.39 & 16.99 & 0.45 & 20.12 & 21.80 & 56.90 \\ 
   & Chocó & Quibdó & 0.33 & 20.50 & 0.29 & 21.55 & 0.30 & 21.29 & 0.19 & 25.00 & 0.20 & 24.31 & 44.40 & 39.80 \\ 
   & Quindío & Calarca & 0.29 & 25.96 & 0.28 & 24.58 & 0.28 & 24.38 & 0.34 & 17.86 & 0.41 & 21.64 & 19.90 & 50.50 \\ 
     \hline     \hline
       \end{tabular}
}   \caption{Direct estimates, $\hat{\theta}_{j,c}^{\textsf{Direct}}$, and RB  estimates of PHIA according to $\hat{\theta}_{j,c}^{\textsf{FH}}$, $\hat{\theta}_{j,c}^{\textsf{FH}\text{-B}}$, $\hat{\theta}_{j,c}^{\textsf{FH-SC$_{2}$}}$ and $\hat{\theta}_{j,c}^{\textsf{FH-SC$_{2}$}\text{-B}}$, and the estimates of Coefficient of Variations (CVs) for each estimator. 
The Table displays large capital cities and some municipalities in Colombia with high levels of poverty and deficit of education according to the Multidimensional Poverty Index 2018 (MPI) and the
Education Deficit Index  2018 (EDI), respectively.}
\label{table:conf_cve}
\end{center}
\end{table}

\clearpage

\begin{figure}[ht]{
\begin{center}
\begin{tabular}{ccc}
 \hspace{-0.6cm} \includegraphics[width=0.48\textwidth]{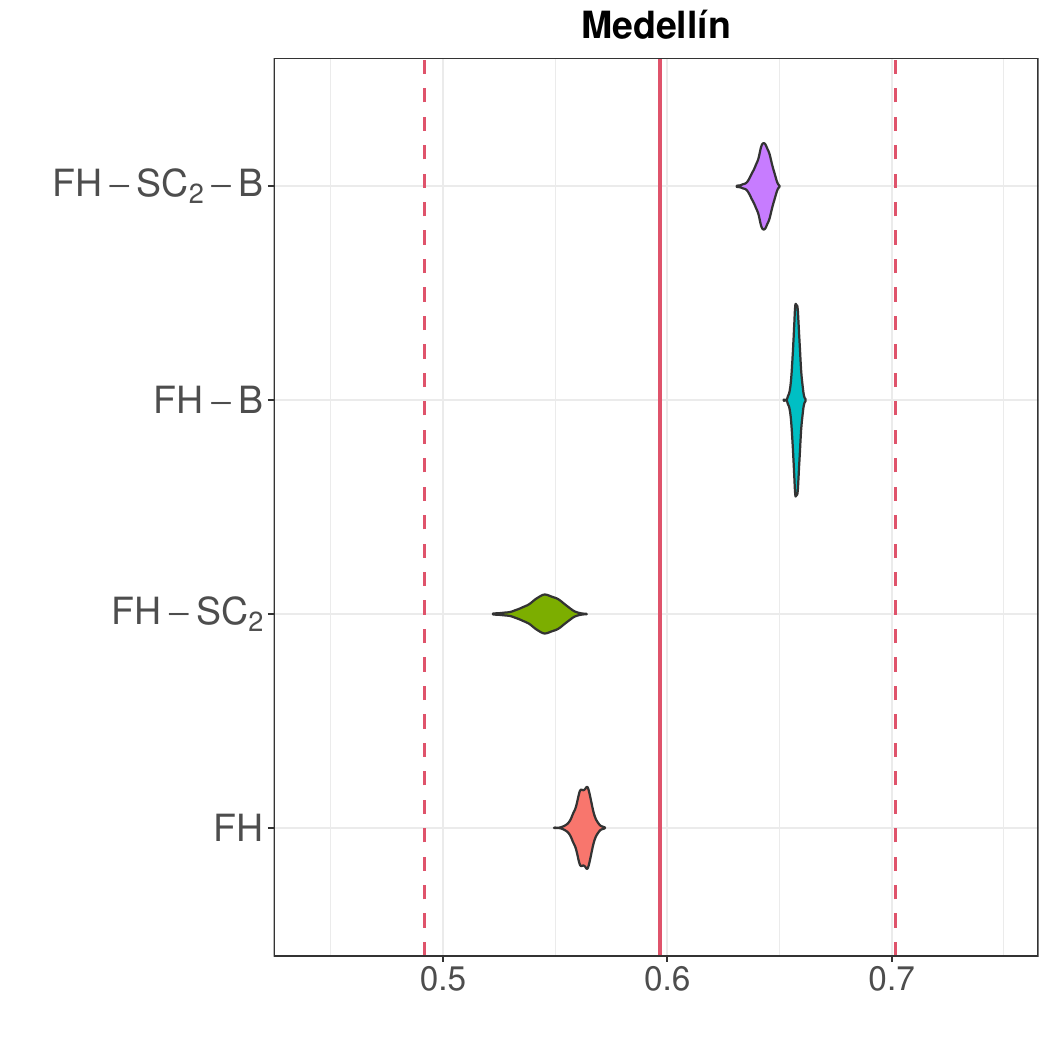} & \hspace{0.2cm}
\includegraphics[width=0.48\textwidth]{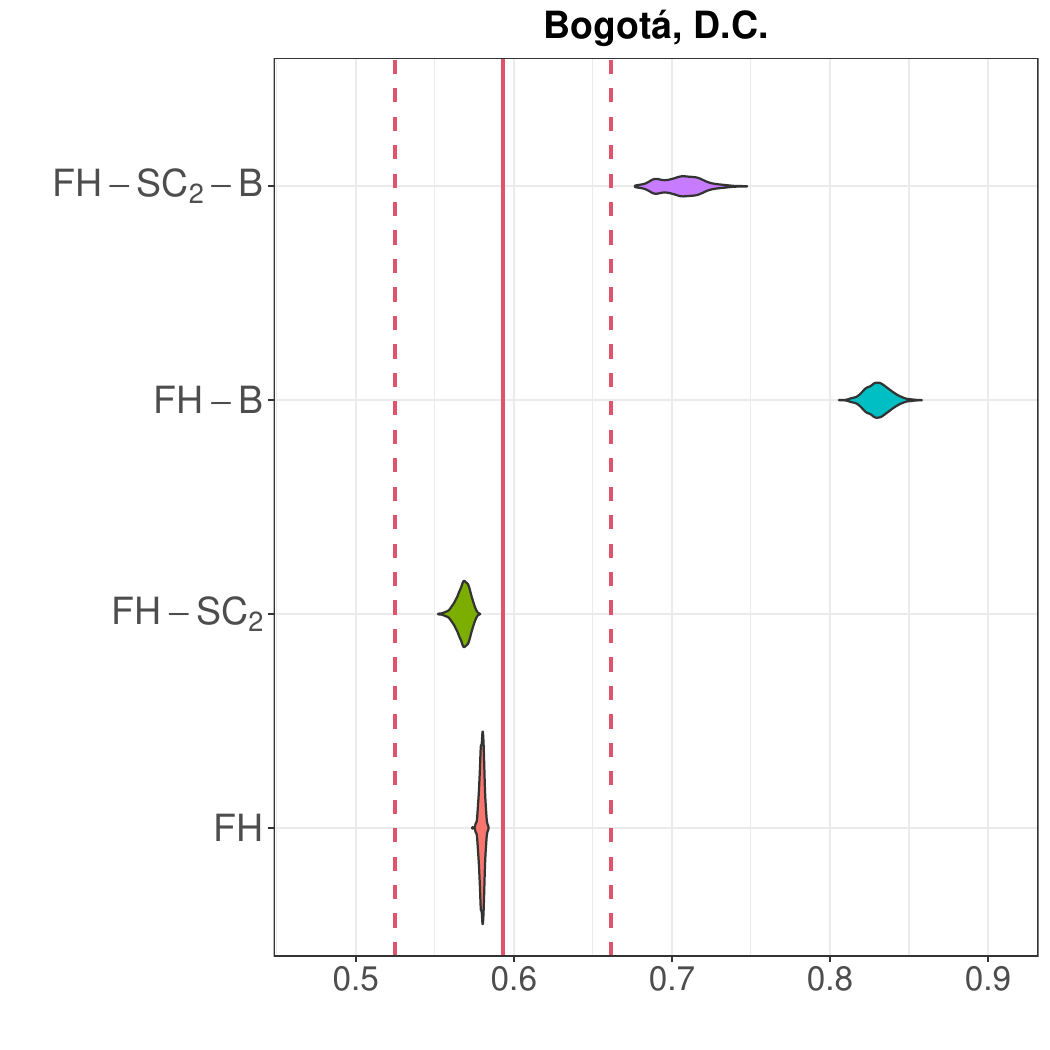} \\  \vspace{0.2cm}
 \hspace{-0.6cm} \includegraphics[width=0.48\textwidth]{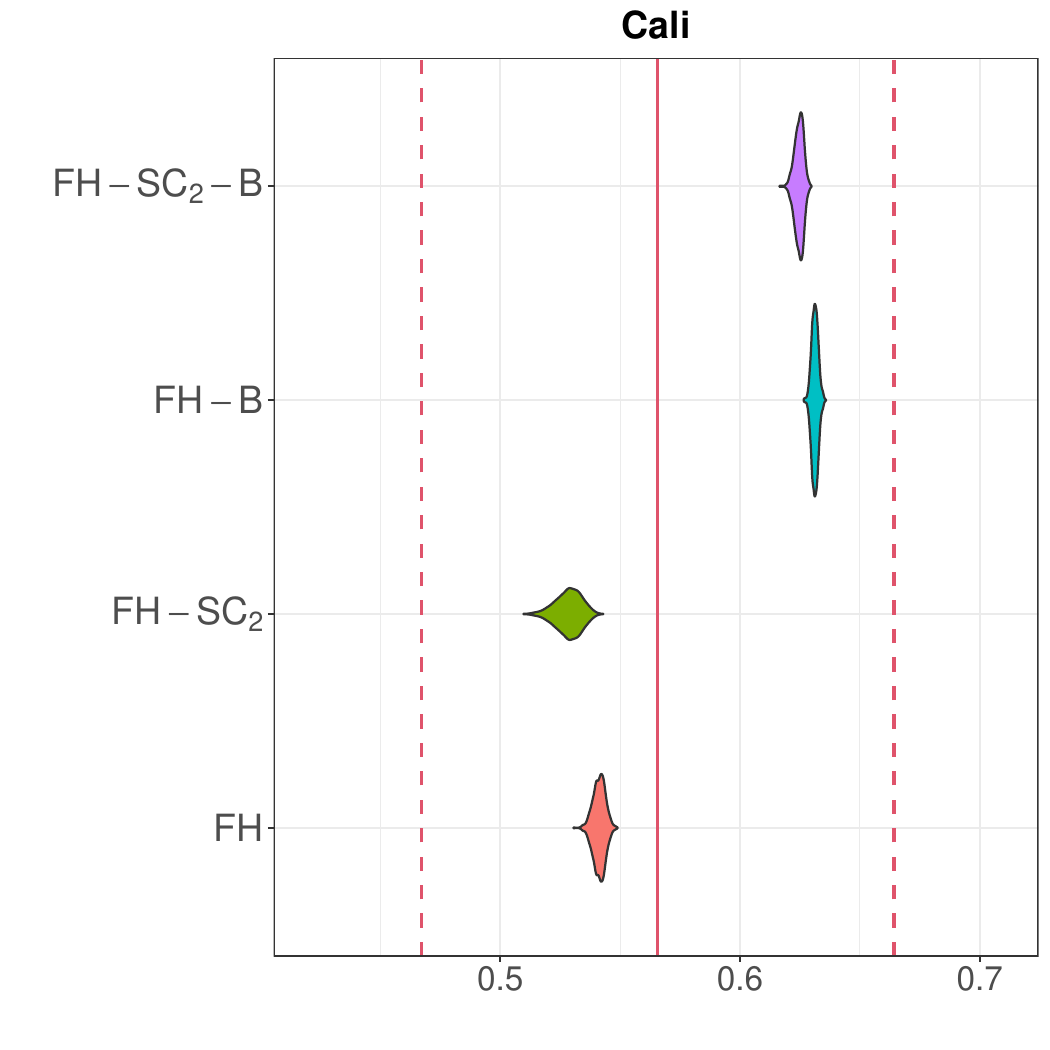} & \hspace{0.2cm}
\includegraphics[width=0.48\textwidth]{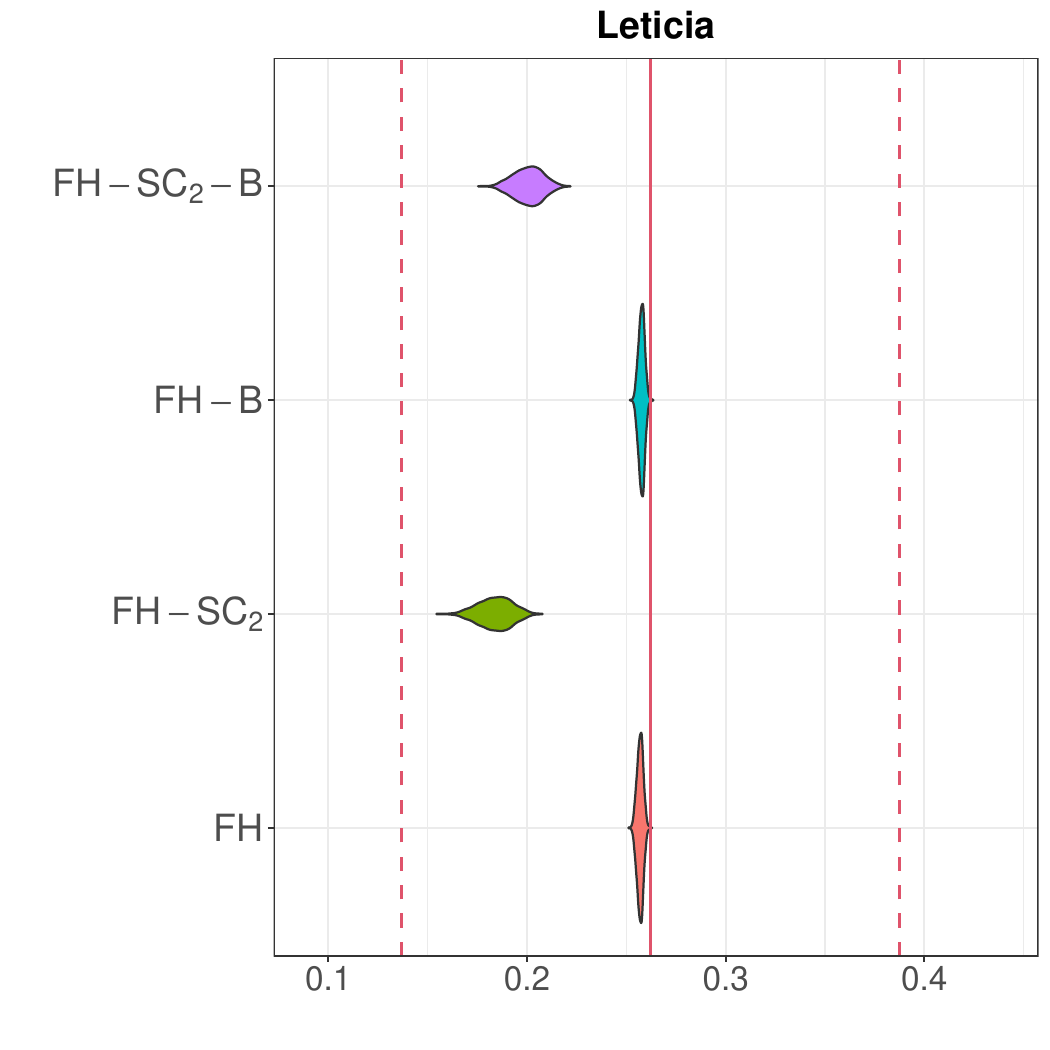} \\  
\end{tabular}
\end{center}
\caption{Posterior distribution of  RB  estimates and RB benchmarked estimates of PHIA for the capital cities of Medellín, Bogotá, D.C.  and Cali, and the municipality of Leticia under the different estimators, $\hat{\theta}_{j,c}^{\textsf{FH}}$,  $\hat{\theta}_{j,c}^{\textsf{FH-SC$_{2}$}}$, $\hat{\theta}_{j,c}^{\textsf{FH}\text{-B}}$ and $\hat{\theta}_{j,c}^{\textsf{FH-SC$_{2}$}\text{-B}}$, denoted as FH, FH-SC$_{2}$, FH-B and FH-SC$_{2}$-B, respectively. The solid red line illustrates the corresponding direct estimate of PHIA  and the dashed red lines display the 95\% confidence interval
considering the direct estimate and direct variance of PHIA.}
\label{fig:posterior_estimates1}}
\end{figure}

\clearpage

\begin{figure}[ht]
\begin{center}
\begin{tabular}{ccc}
 \hspace{-0.6cm} \includegraphics[width=0.48\textwidth]{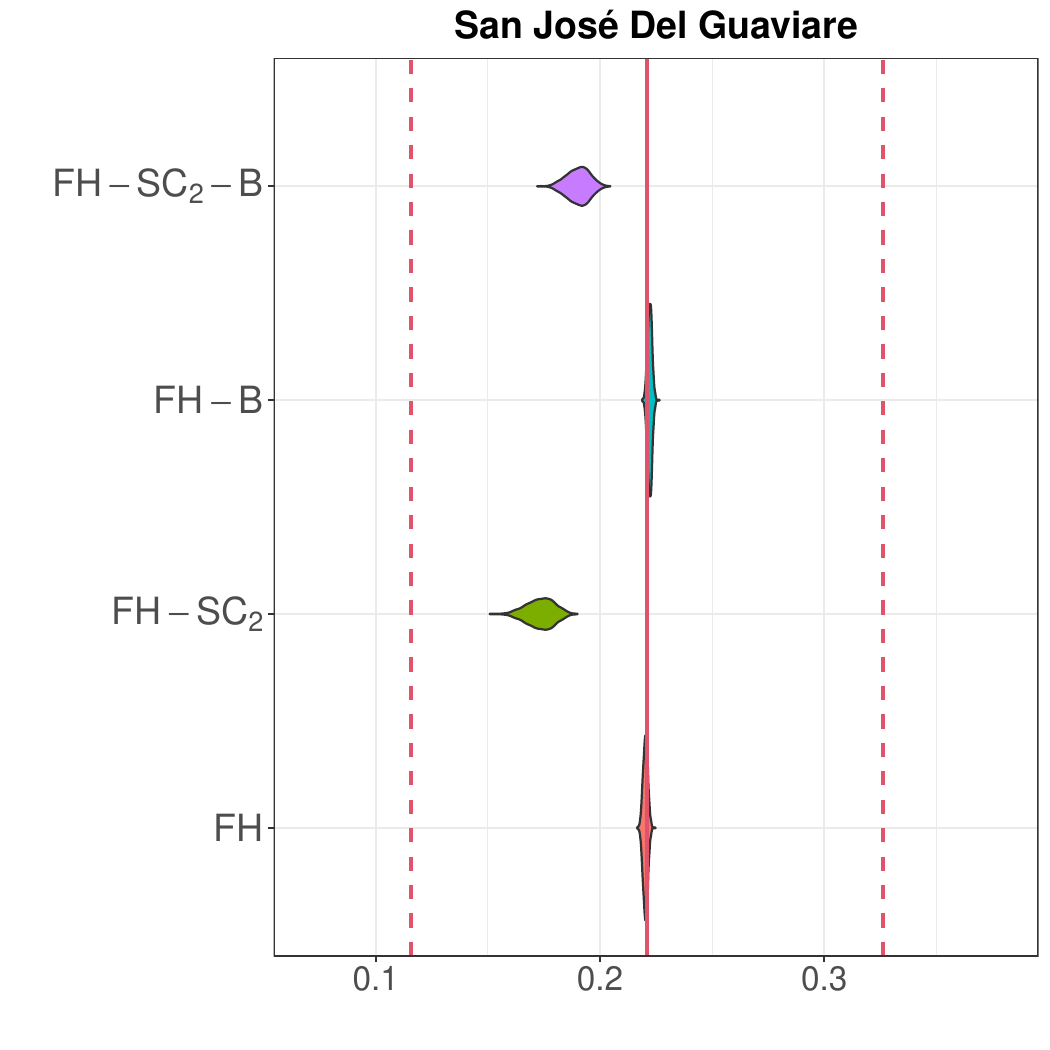} & \hspace{0.2cm}
\includegraphics[width=0.48\textwidth]{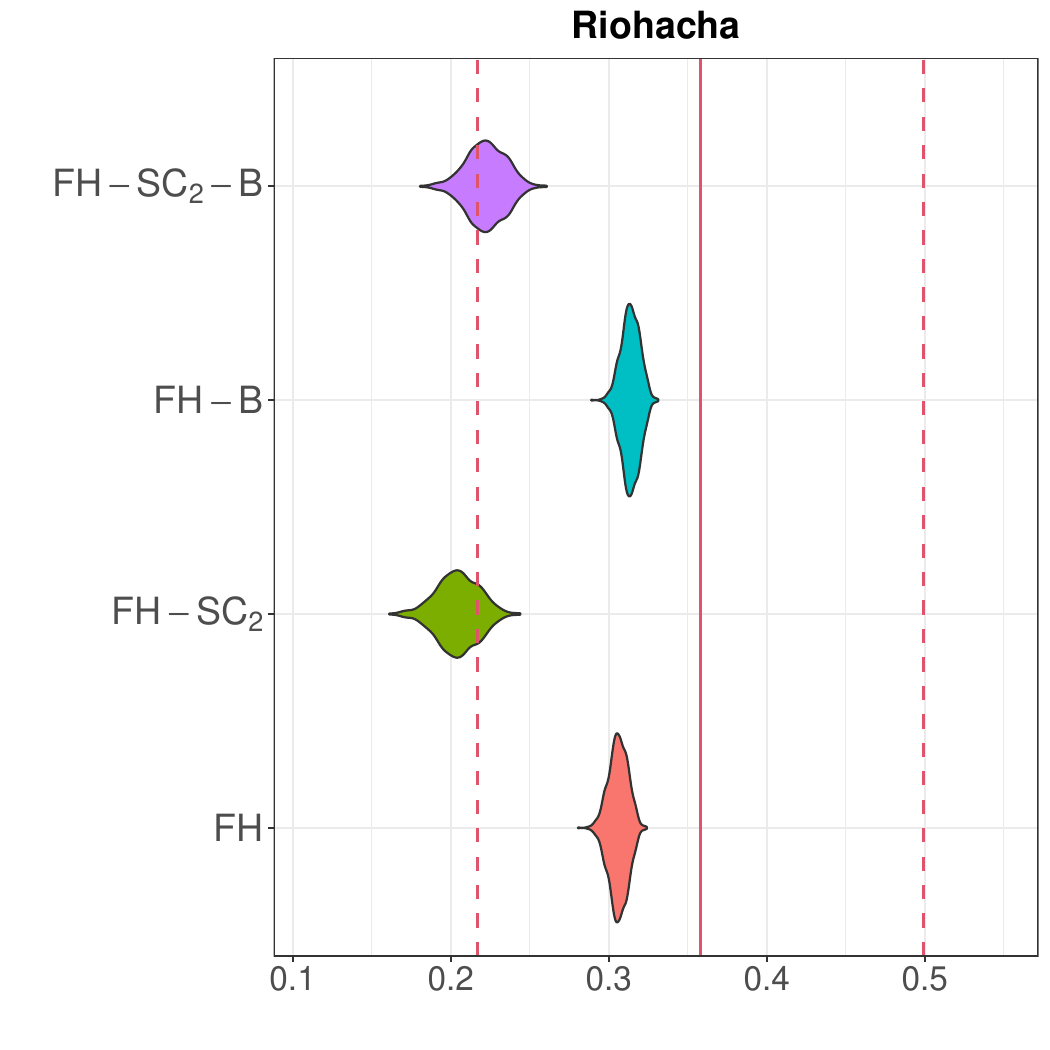} \\  \vspace{0.2cm}
 \hspace{-0.6cm} \includegraphics[width=0.48\textwidth]{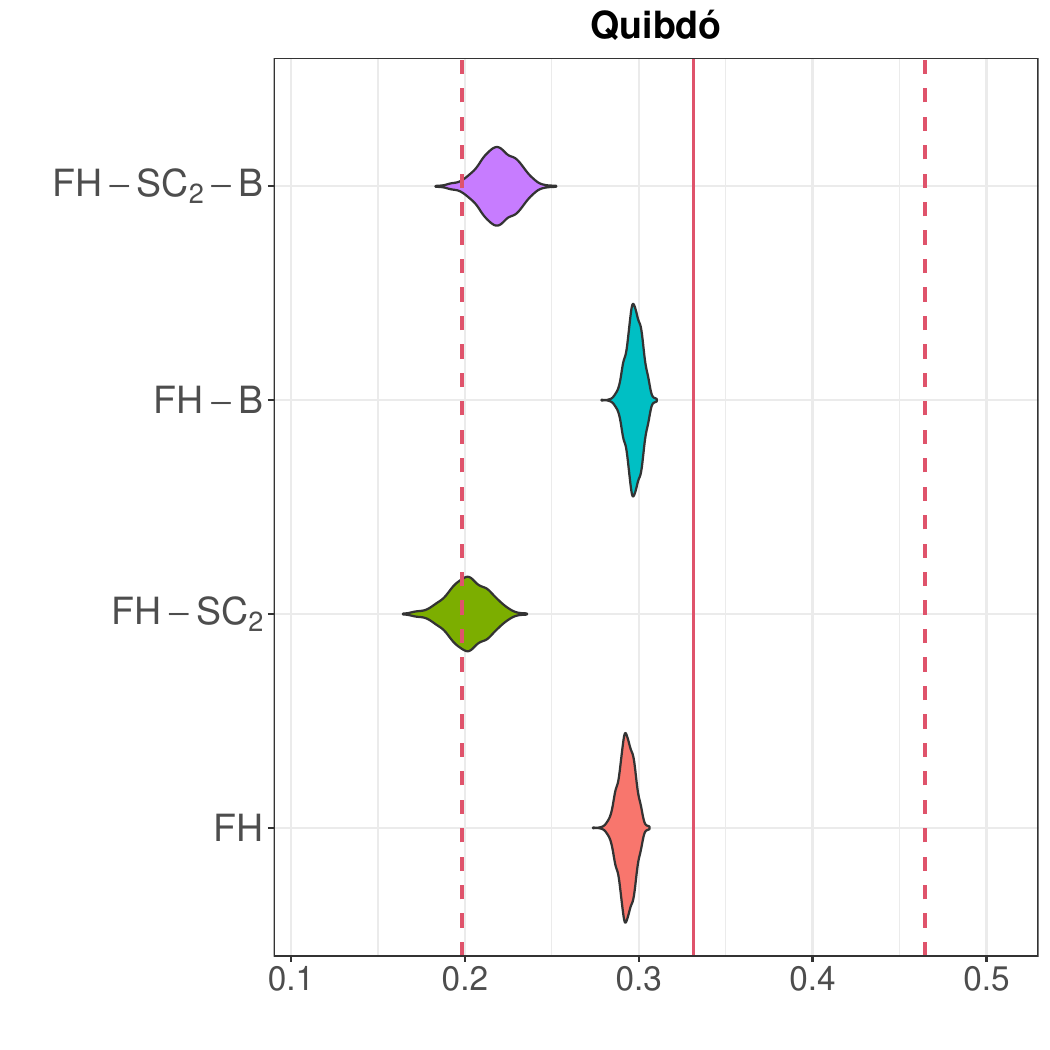} & \hspace{0.2cm}
\includegraphics[width=0.48\textwidth]{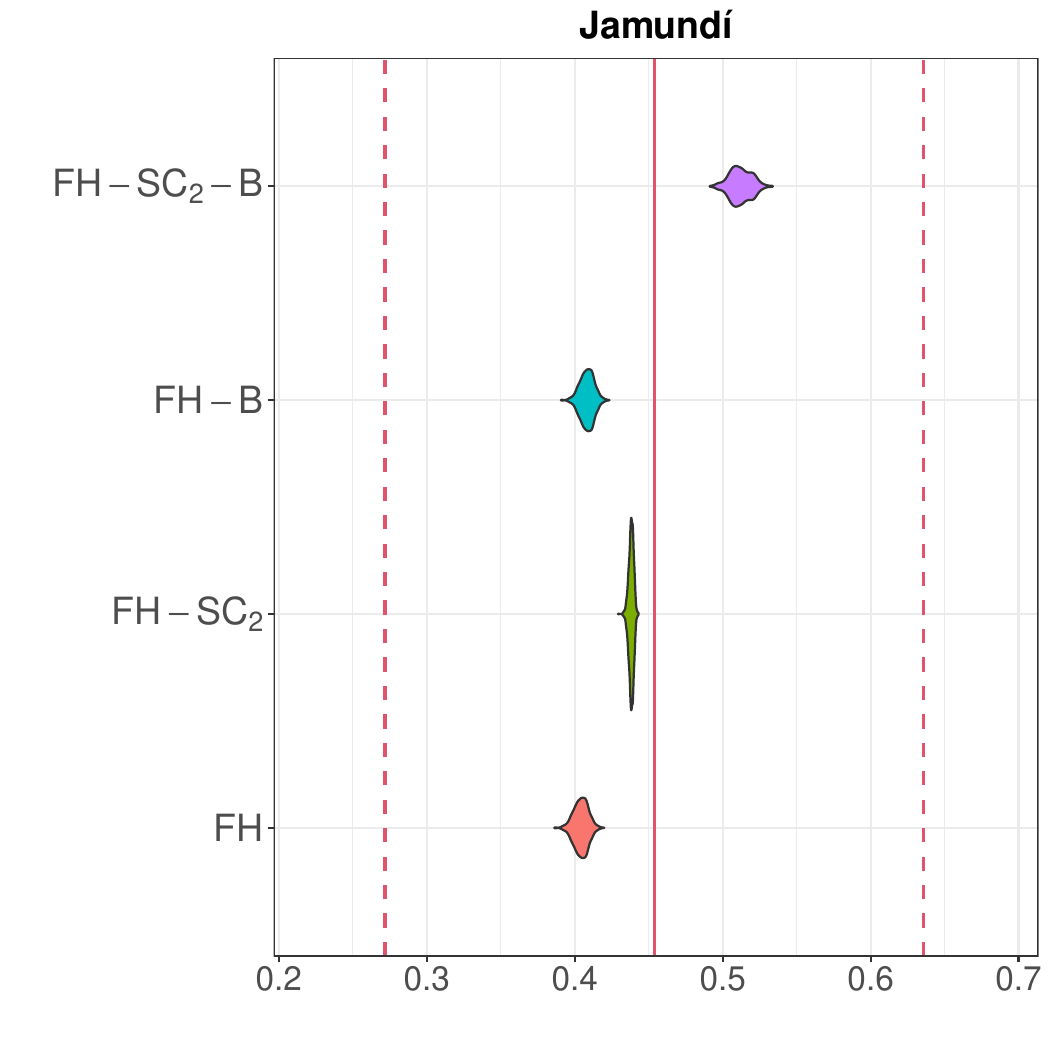} \\  
\end{tabular}
\end{center}
\caption{Posterior distribution of  RB  estimates and RB benchmarked estimates of PHIA  for the municipalities of San José Del Guaviare, Riohacha, Quibdó and Jamundí in Colombia under the different estimators, $\hat{\theta}_{j,c}^{\textsf{FH}}$,  $\hat{\theta}_{j,c}^{\textsf{FH-SC$_{2}$}}$, $\hat{\theta}_{j,c}^{\textsf{FH}\text{-B}}$ and $\hat{\theta}_{j,c}^{\textsf{FH-SC$_{2}$}\text{-B}}$, denoted as FH, FH-SC$_{2}$, FH-B and FH-SC$_{2}$-B, respectively. The solid red line illustrates the corresponding direct estimate of PHIA  and the dashed red lines display the 95\% confidence interval
considering the direct estimate and direct variance of PHIA.}
\label{fig:posterior_estimates2}
\end{figure}

\clearpage

\begin{figure}[h!]
\begin{center}
\begin{tabular}{ccc}
FH &  \hspace{2.5cm}  FH-SC$_{2}$ \vspace{-1.0cm}  \\
\hspace{-1cm}  \includegraphics[width=0.6\textwidth]{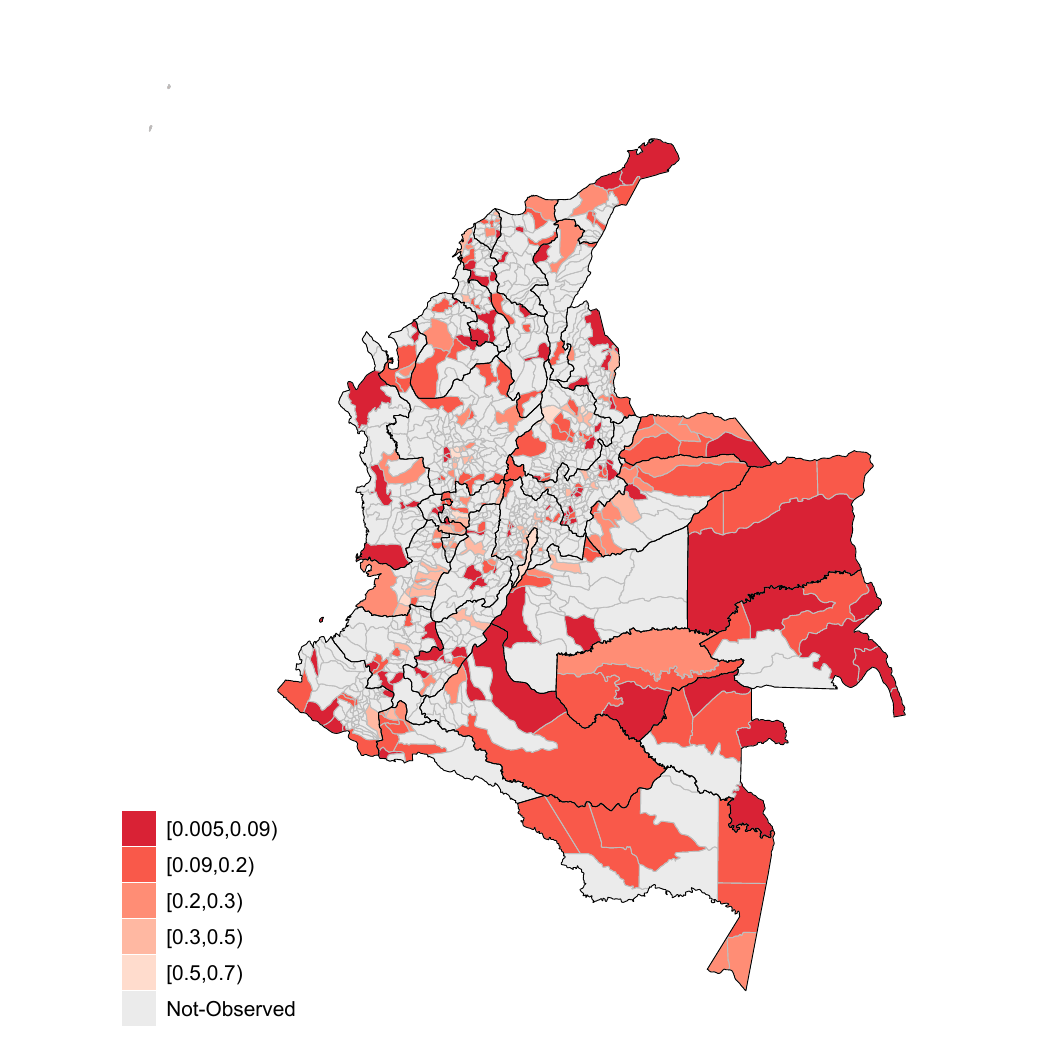} \hspace{-2cm} & \includegraphics[width=0.6\textwidth]{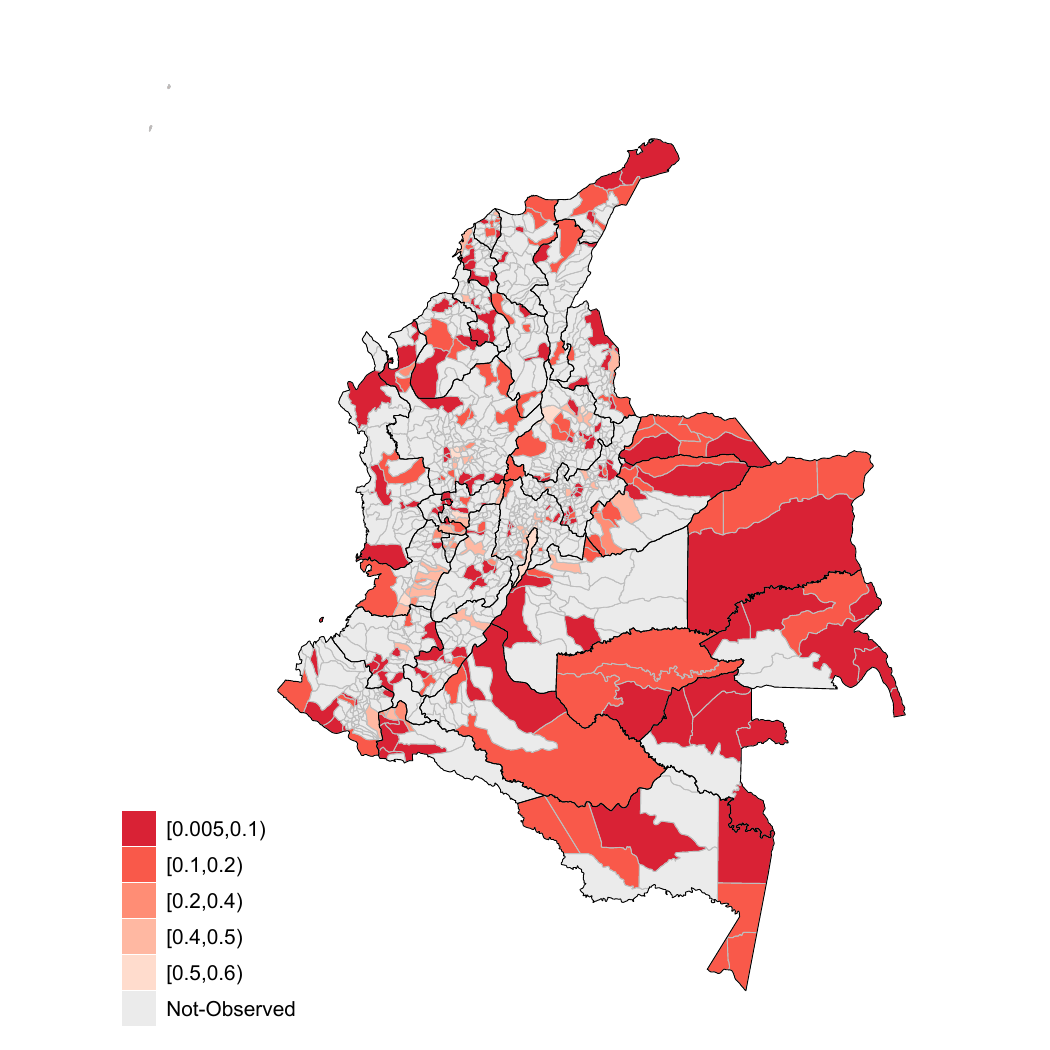}  \hspace{-2cm}  \\
FH-B &  \hspace{2.5cm}  FH-SC$_{2}$-B \vspace{-1.0cm}  \\
\hspace{-1cm}  \includegraphics[width=0.6\textwidth]{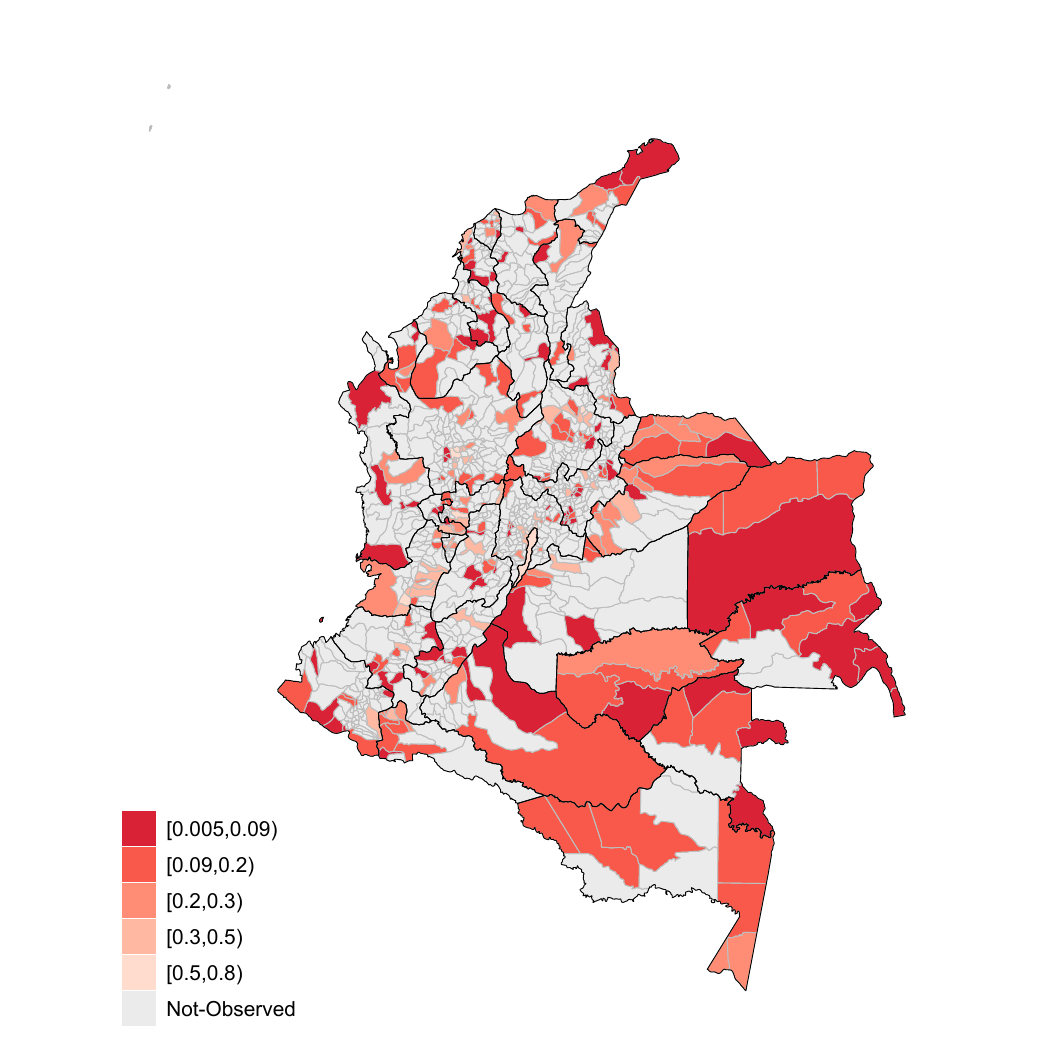} \hspace{-2cm} & \includegraphics[width=0.6\textwidth]{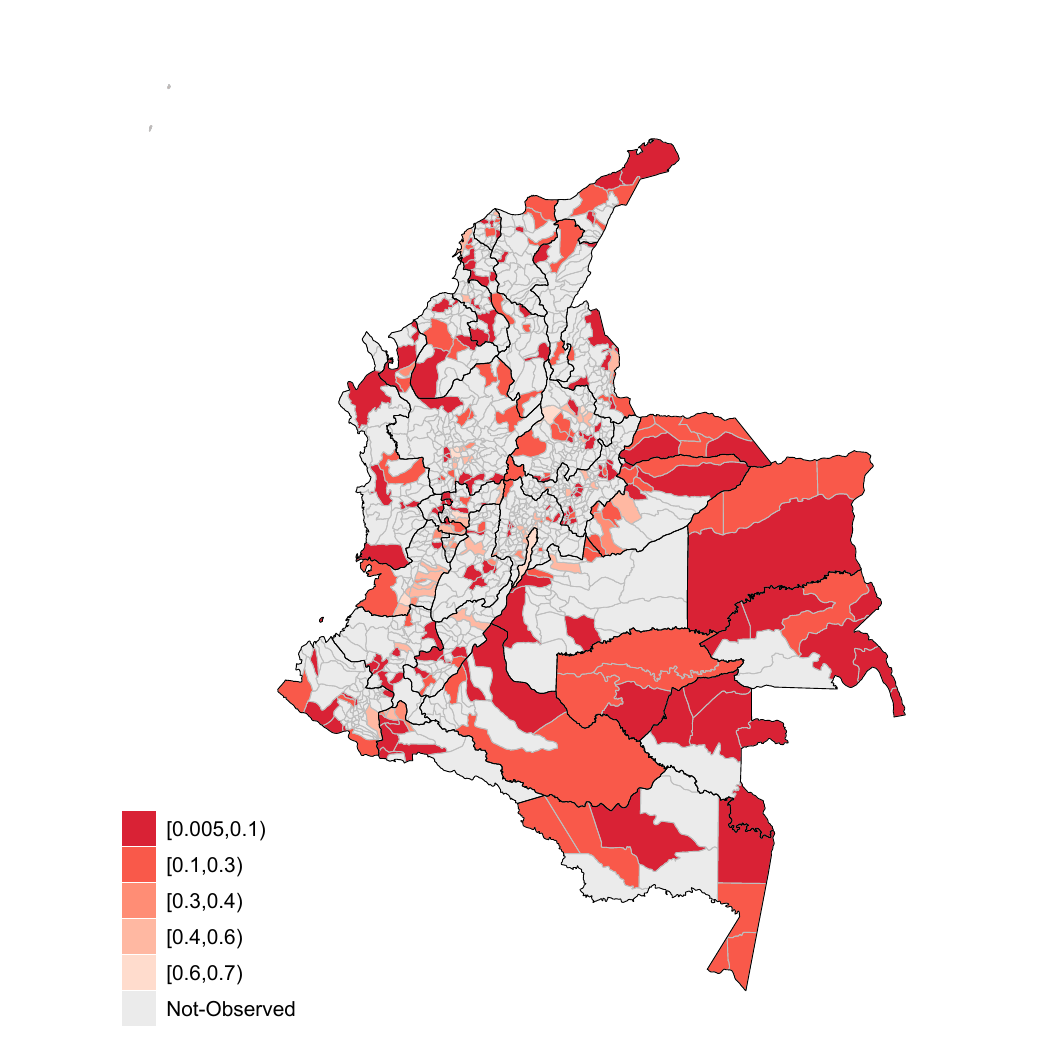}  \hspace{-2cm}  \\
\end{tabular}
\end{center}
\vspace{-0.5cm}
\caption{
RB estimates of PHIA  produced by  a) $\hat{\theta}_{j,c}^{\textsf{FH}}$, b) $\hat{\theta}_{j,c}^{\textsf{FH-SC$_{2}$}}$, c) $\hat{\theta}_{j,c}^{\textsf{FH}\text{-B}}$  and, d) $\hat{\theta}_{j,c}^{\textsf{FH-SC$_{2}$}\text{-B}}$. 
The larger values of RB estimates and RB benchmarked estimates of PHIA  are  concentrated in some of  the main capital cities, such as Bogotá, D.C., Medellín, and Cali. The RB estimator and RB benchmarked estimator under the FHSC model produces
the most conservative estimates of PHIA.
\label{fig:post1}}
\end{figure}

\clearpage

\subsection{Sensitivity of results for different amounts of smoothing via GVF}
\label{sub_applied2}

We compute the direct estimates and their variances using the estimators given by \cite{hajek1971comment}. Specifically, the variable of interest is defined as $y_{ih}=1$ if household $h$ in municipality $i$  has internet access, and $y_{ih}=0$ otherwise.

\begin{figure}[ht]
\begin{center}
\begin{tabular}{ccc}
\includegraphics[width=0.75\textwidth]{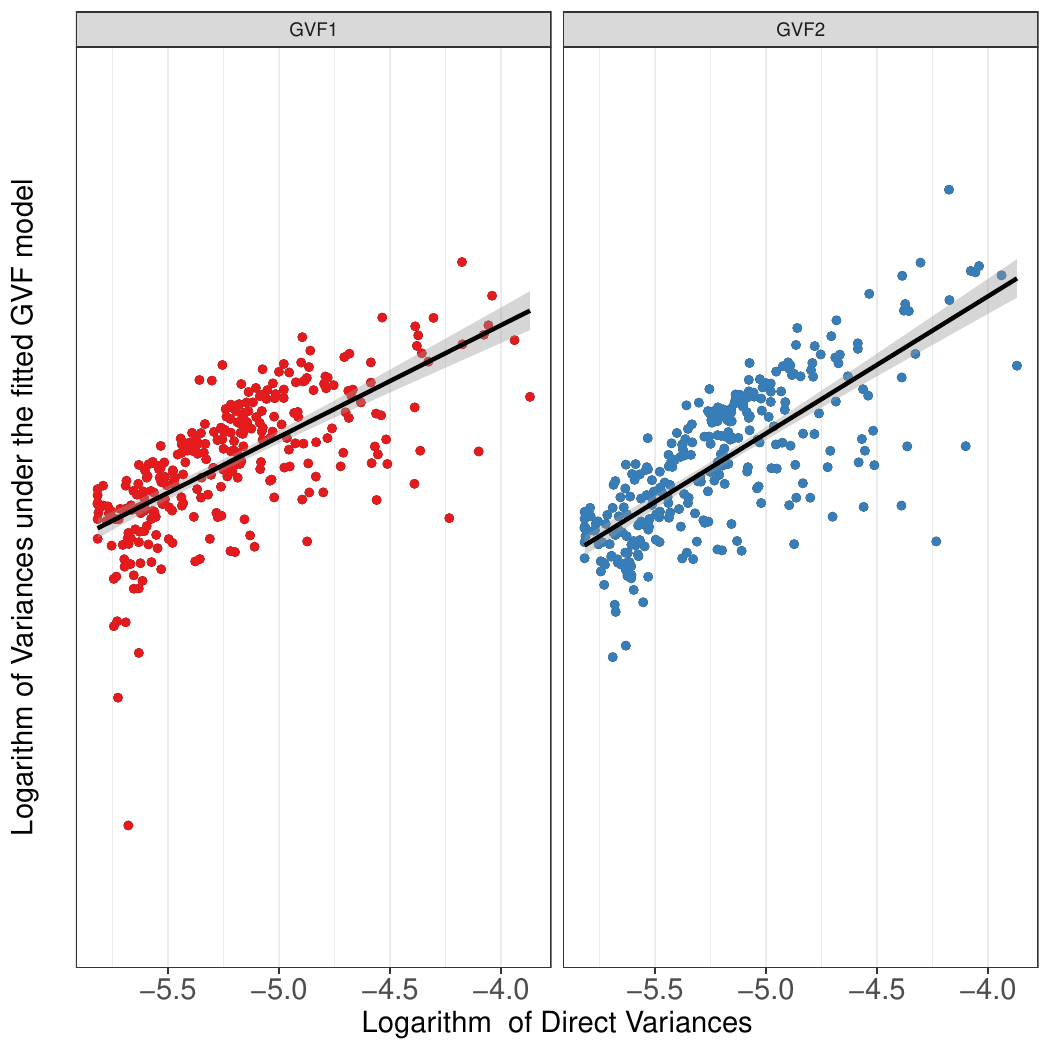} & \hspace{0.2cm}
\end{tabular}
\end{center}
\caption{GVF models to smooth the variance estimates. Left: 
Logarithm  of  variance estimates, $\log(\text{var}(y_{i}))$, versus logarithm of the obtained variances under the GVF$_{1}$  model. Right: 
Logarithm  of  variance estimates, $\log(\text{var}(y_{i}))$, versus logarithm of the obtained variances  under the GVF$_{2}$ model.}
\label{fig:GVF_models}
\end{figure}

We use $y_{ih}$ and its corresponding sample weights $w_{ih}$ for municipality $i$ and household $h$ to compute the direct estimator $y_{i}$ and the variance estimates for PHIA in municipality $i$, as follows:
\begin{align}
y_{i}&=\dfrac{1}{\hat{n}_{i}}\sum_{h=1}^{n_{i}}w_{ih} y_{ih}, & \text{var}(y_{i})&=\dfrac{1}{\hat{n}_{i}^{2}}\sum_{h=1}^{n_{i}}w_{ih}(w_{ih}-1)(y_{ih}-y_{i})^{2},   & i=1, \ldots, m,
\label{eq:estimator}
\end{align}
where $n_{i}$ denotes the sample size of the $i$-th municipality, and $\hat{n}_{i}=\sum_{j=1}^{n_{i}}w_{ih} $ is the direct estimator of the population size in municipality $i$.  The variances of the direct estimates are smoothed using the Generalized Variance Function (GVF)  \citep{wolter1985introduction}.
Specifically, we consider two GVF models to smooth the variance estimates,  $\text{var}(y_{i})$  in (\ref{eq:estimator}), to obtain $D_i$. The covariates in the first GVF model (GVF$_{1}$)  are the PHIA direct estimates, the square root of small area sample sizes and the interaction between them, whereas in the second GVF model   (GVF$_{2}$) the covariates are the direct estimates and the square root of small area sample sizes. 
Figure \ref{fig:GVF_models} displays the logarithm of the variance estimates versus the logarithm of the smoothed variances, $D_i$, 
obtained with the GVF$_{1}$ and GVF$_{2}$ models, respectively. 
 We found that GVF$_{2}$ provides a better fit of the logarithm of the variance estimates, and the variances under the GVF$_{2}$ lead to a smaller mean square error (0.075) compared to GVF$_{1}$ (0.094).  
 In order to conduct a sensitivity analysis  of the different levels of smoothing via GVF$_{1}$ and GVF$_{2}$, 
 we compare the small area posterior estimates obtained under the FH and FH-SC$_{2}$ models. As is illustrated in Figure \ref{fig:differences}, the  posterior estimates using the variances produced by GVF$_{1}$ and GVF$_{2}$ are similar. However, we consider the smoothed variances obtained with GVF$_{2}$ for the case study and simulations based on its lower MSE value. 

\begin{figure}[ht]
\begin{center}
\begin{tabular}{ccc}
\hspace{-2.0cm} \includegraphics[width=0.70\textwidth]{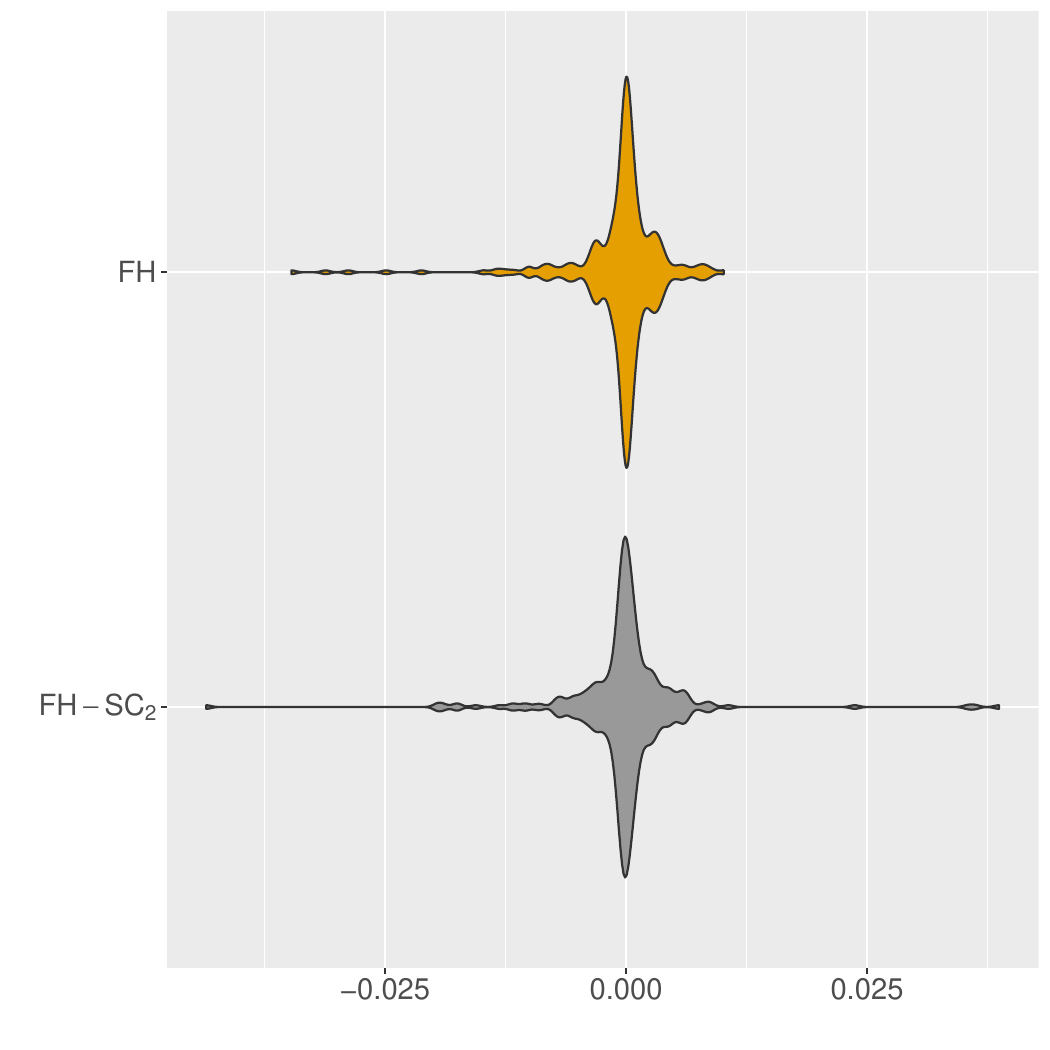} 
\end{tabular}
\end{center}
\caption{Differences of the small area posterior estimates under the FH and FH-SC$_{2}$ models using the smoothed  variances computed with the GVF$_{1}$ and GVF$_{2}$ models.}
\label{fig:differences}
\end{figure}


\clearpage

\subsection{Convergence of the MCMC algorithms}
\label{sub_applied4}

We found convergence of the model parameters under the existing FH and FH-C and the proposed FH-SC models when Algorithms  \ref{alg:MCMC1} and \ref{alg:MCMC2} are implemented. We consider the results under FH-SC$_{2}$ in our case study to illustrate that convergence is achieved for the model parameters. Figures \ref{fig:convergence_fixed_effects}-\ref{fig:penalty} illustrate the posterior samples in sequential order using traceplots and ergodic mean plots for the regression parameters, variances and cluster regularization penalty, respectively. As indicated by Figures \ref{fig:convergence_fixed_effects}-\ref{fig:penalty} convergence and a good mixing is achieved. Importantly, as mentioned 
in Section \ref{MCM_algh1}, we use an adaptive proposal to sample the cluster regularization penalty adjusted to hold acceptance rates between 40\% and 60\%. We found that the acceptance rates for sampling $\rho$ under the FH-SC$_{1}$,  FH-SC$_{2}$ and FH-SC$_{3}$ models were 51.41\%, 50.85\% and 50.03\%, respectively, showing a good performance of the adaptive proposal in the Metropolis-Hasting algorithm 
for $\rho$.

 \vspace{0.1cm}

\begin{figure}[ht]
\begin{center}
\begin{tabular}{ccc}
Traceplot for $\beta_{0}$ & \hspace{1cm} Ergodic mean plot for $\beta_{0}$ \\
\hspace{-0.5cm} \includegraphics[width=0.40\textwidth]{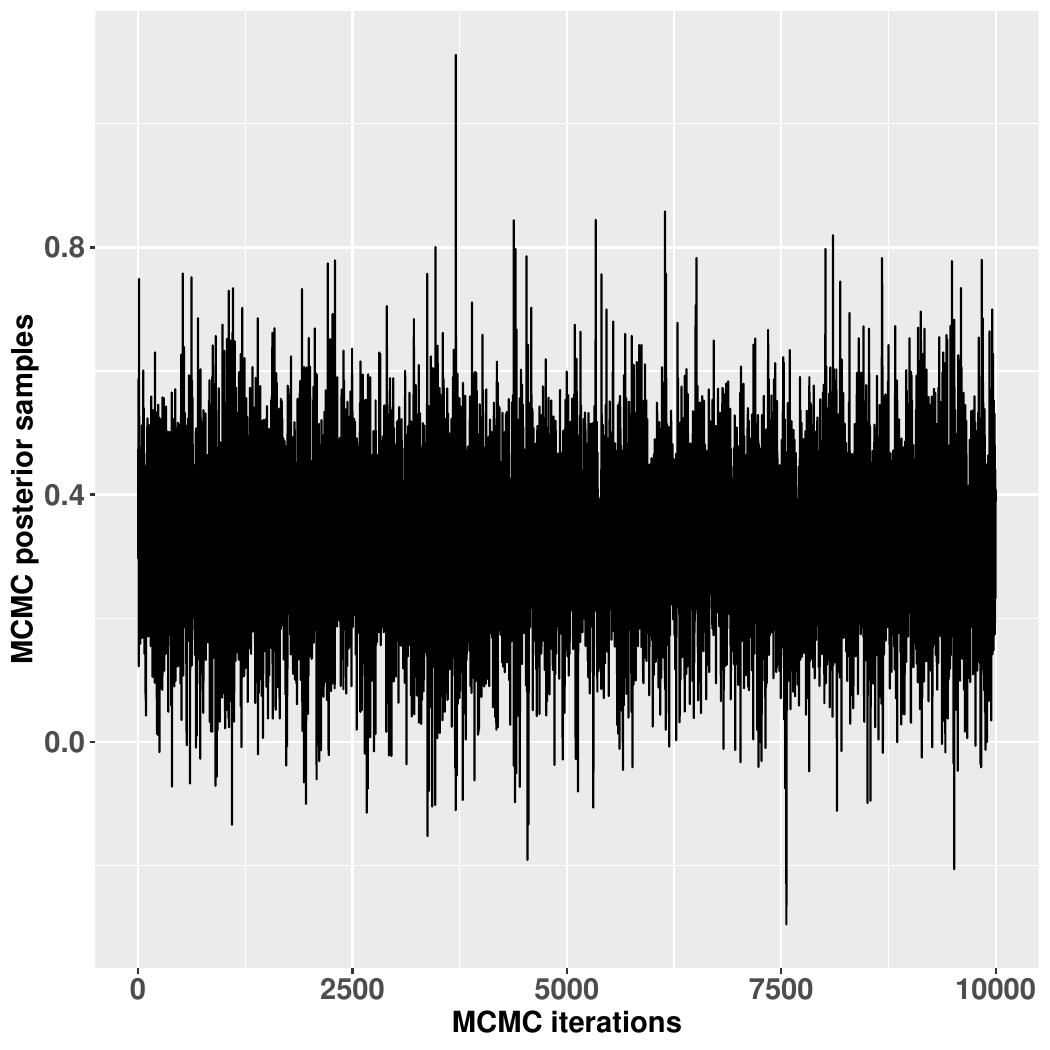} & \hspace{-0.2cm}
\includegraphics[width=0.40\textwidth]{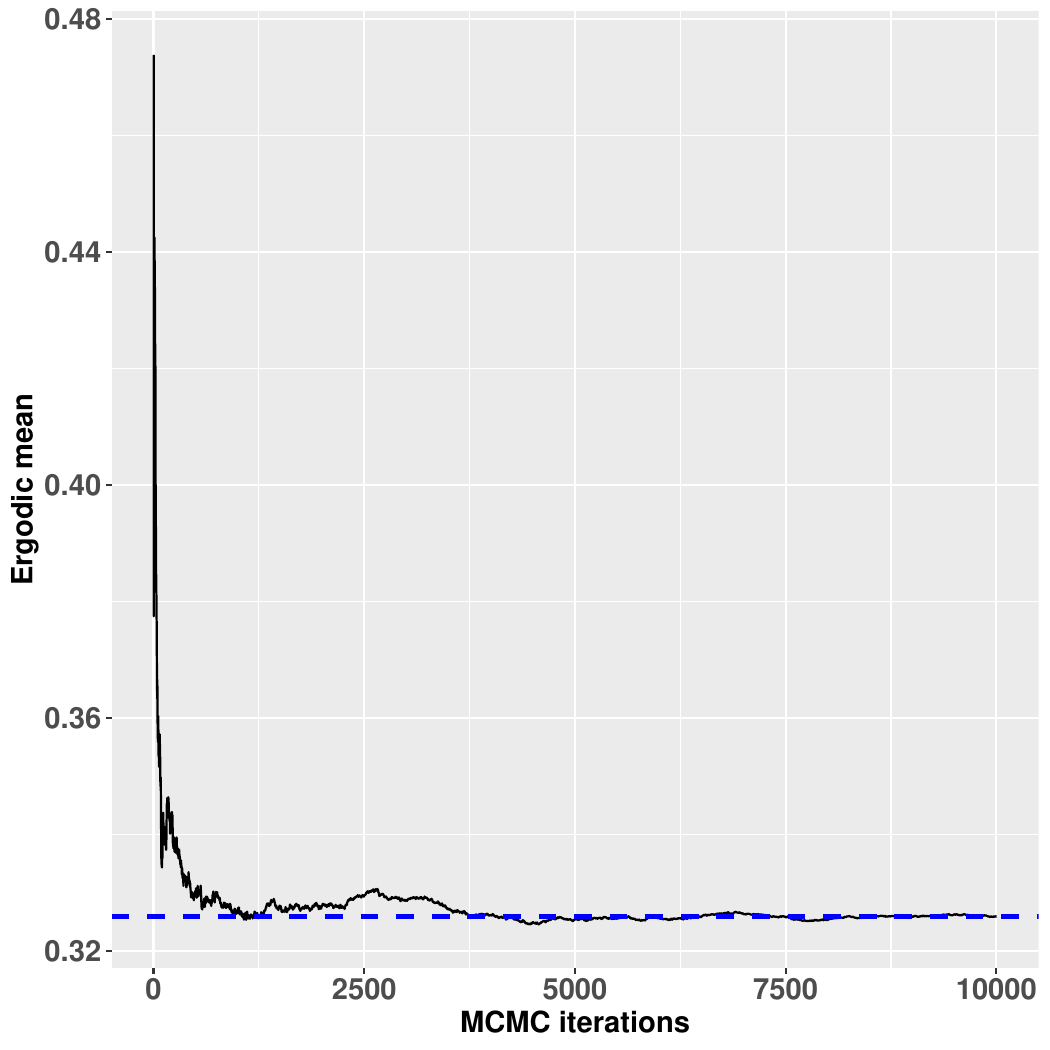} 
\end{tabular}
\end{center}
\begin{center}
\begin{tabular}{ccc}
Traceplot for $\beta_{1}$ & \hspace{1cm} Ergodic mean plot for $\beta_{1}$ \\
\hspace{-1.0cm} \includegraphics[width=0.40\textwidth]{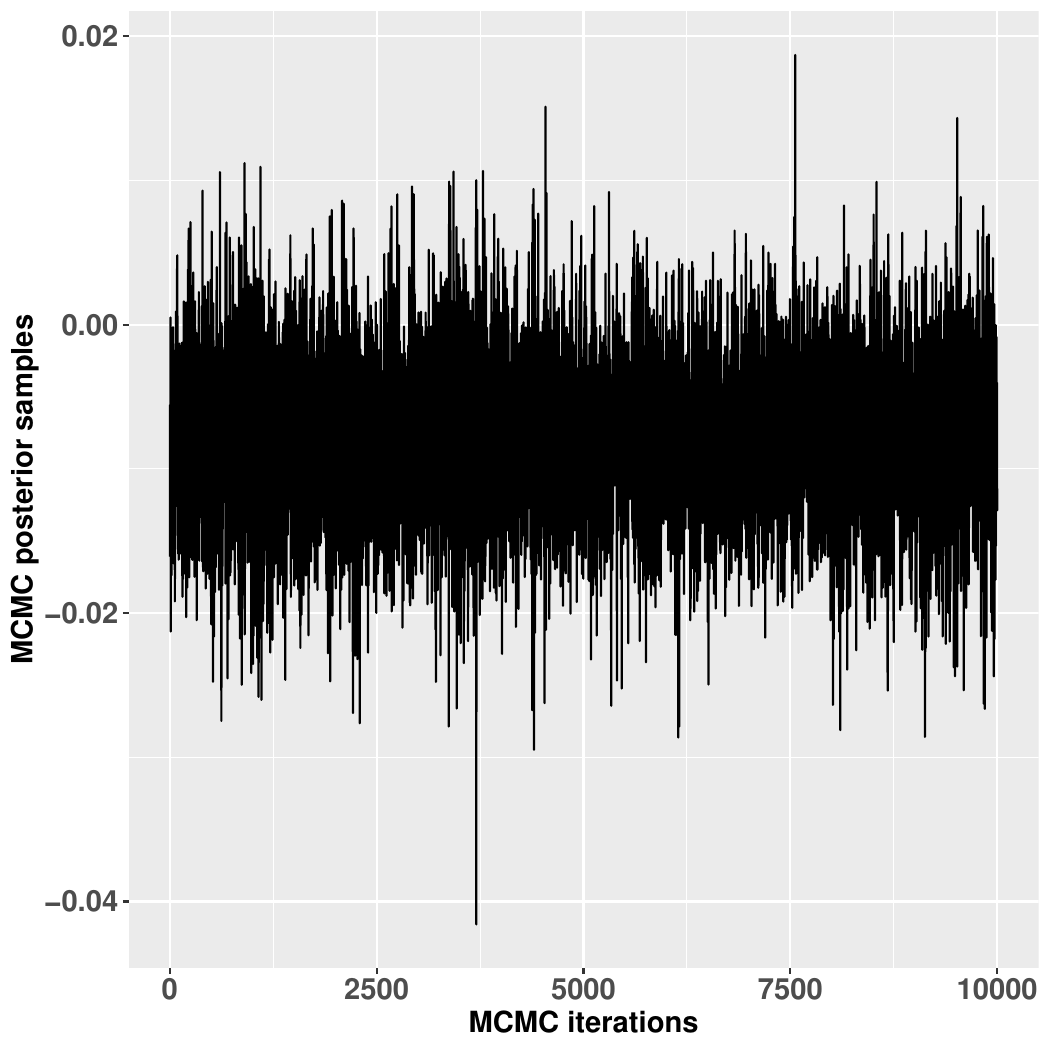} & \hspace{-0.2cm}
\includegraphics[width=0.40\textwidth]{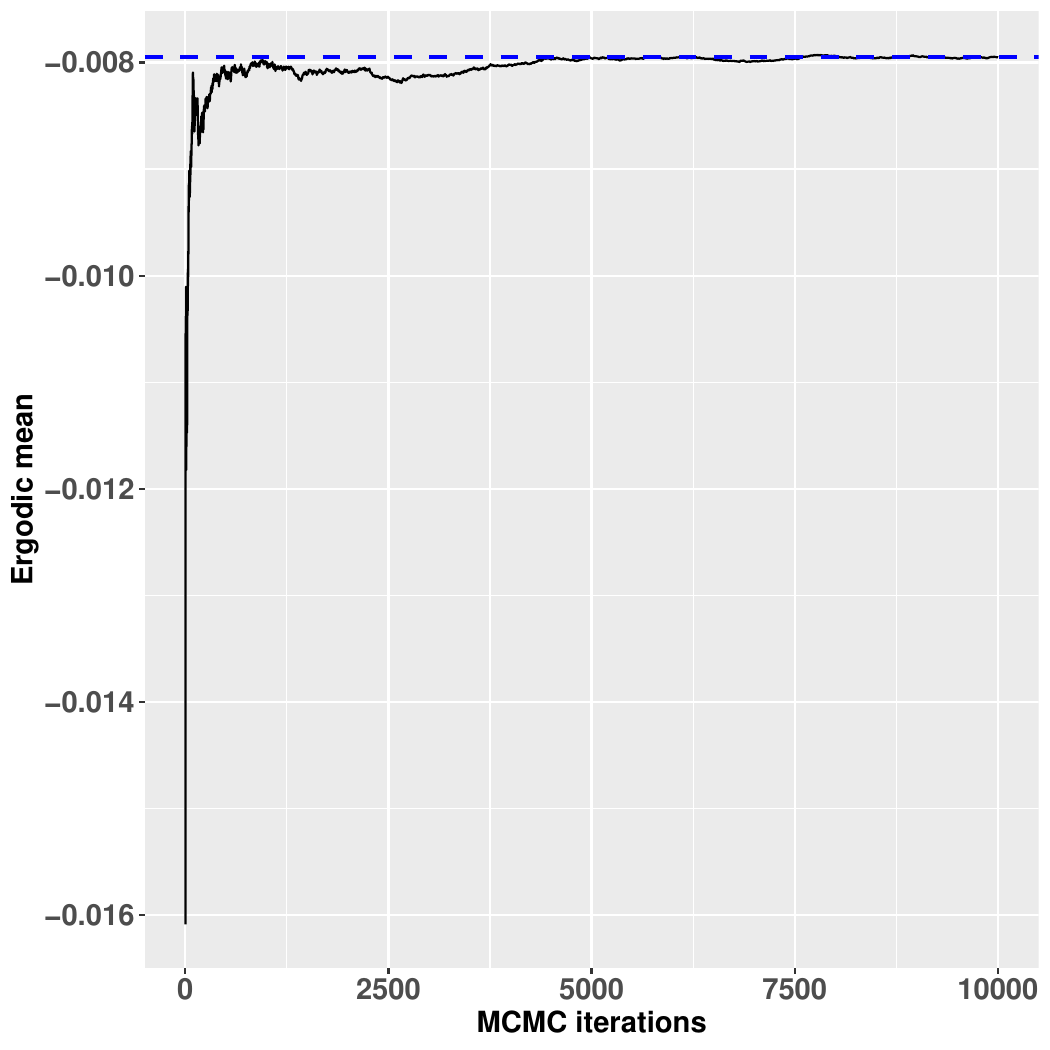} 
\end{tabular}
\end{center}
\caption{Traceplots  and ergodic mean plots for the regression parameter $\beta_{0}$ and $\beta_{1}$ under the FH-SC$_{2}$ model.}
\label{fig:convergence_fixed_effects}
\end{figure}

\clearpage

\begin{figure}[ht]
\begin{center}
\begin{tabular}{ccc}
Traceplot for $\sigma^2_{1}$ & Traceplot for $\sigma^2_{2}$  & Traceplot for $\sigma^2_{3}$ \\
\hspace{-1.9cm} \includegraphics[width=0.40\textwidth]{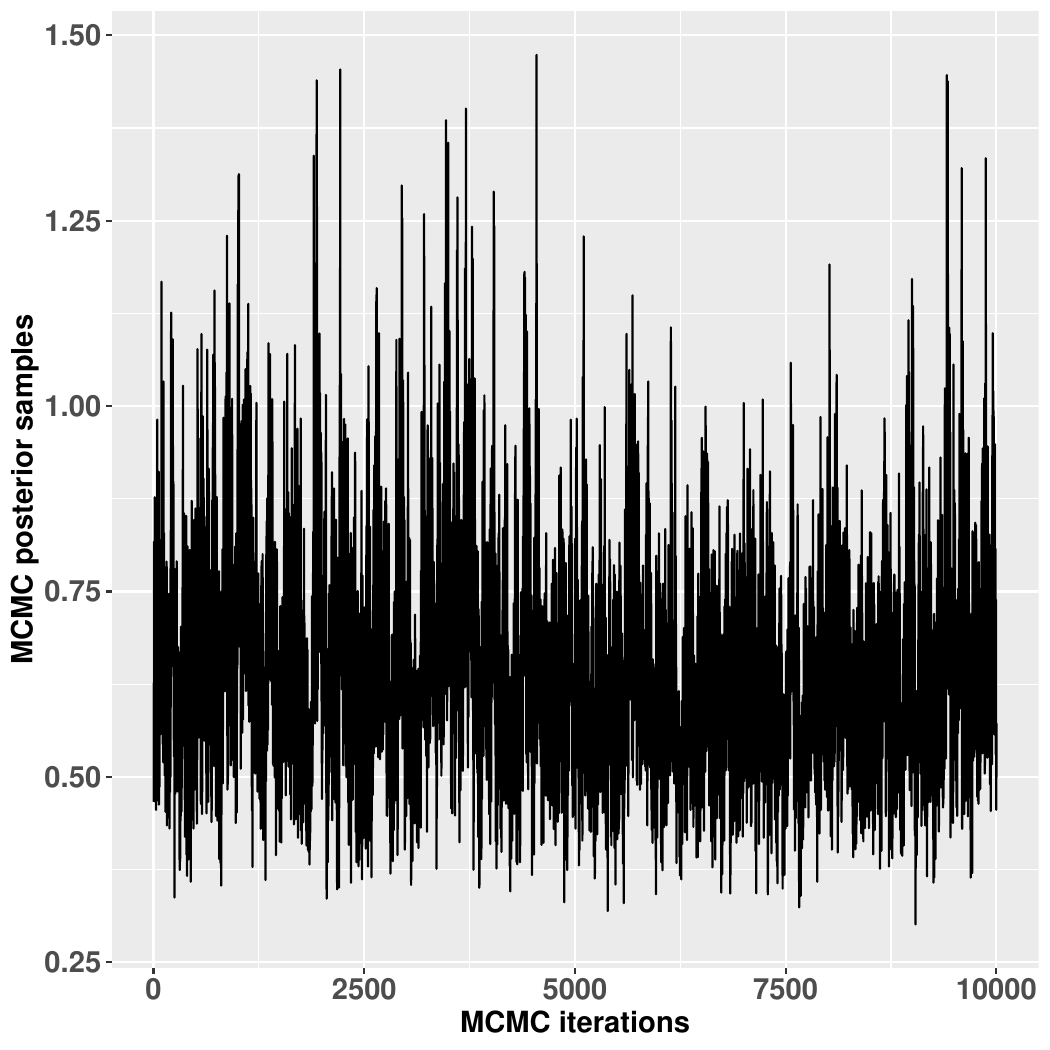}  & \hspace{-0.2cm}
\includegraphics[width=0.40\textwidth]{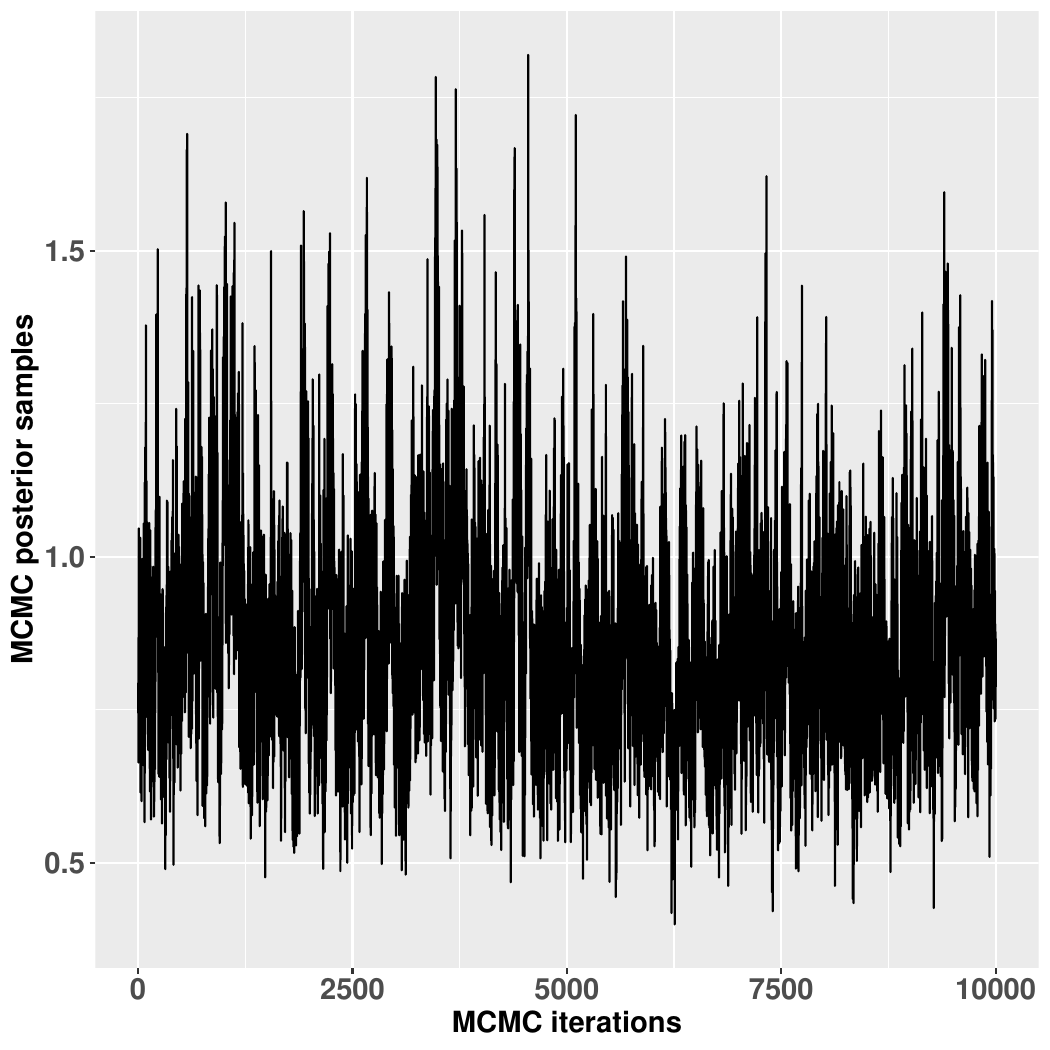} & \hspace{-0.2cm}
\includegraphics[width=0.40\textwidth]{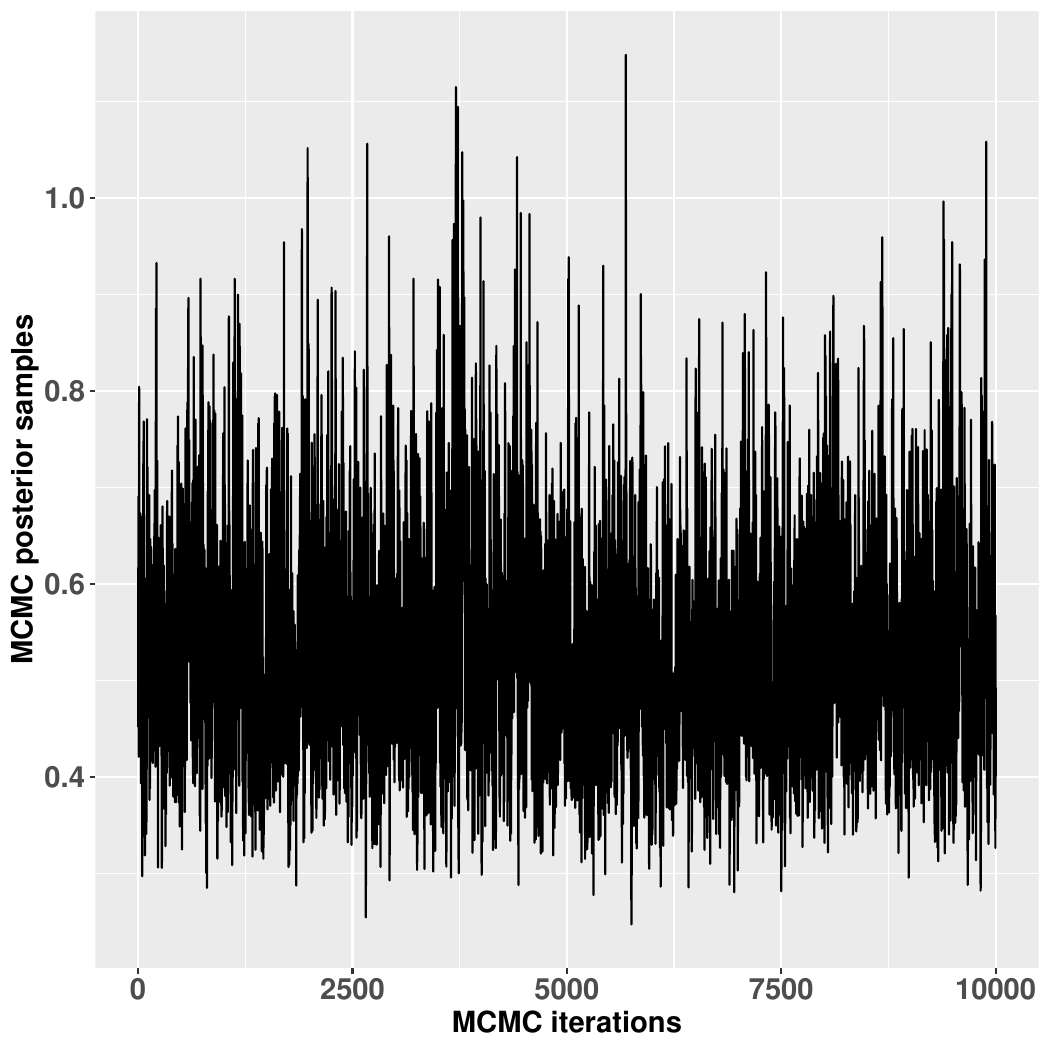}  \vspace{1.5cm} \\
Ergodic mean plot for $\sigma^2_{1}$ & Ergodic mean plot  for $\sigma^2_{2}$  & Ergodic mean plot for $\sigma^2_{3}$ \\
\hspace{-1.9cm} \includegraphics[width=0.40\textwidth]{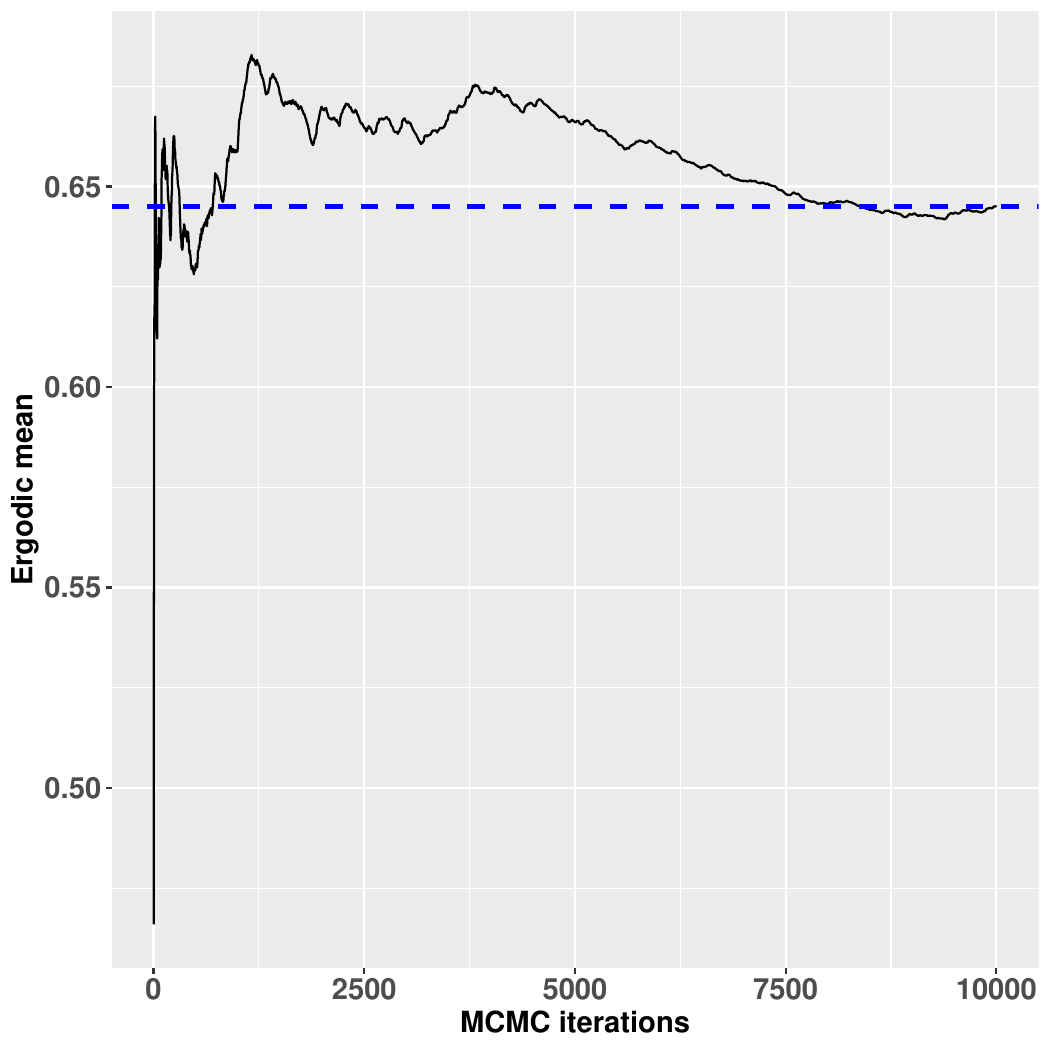}  & \hspace{-0.2cm}
\includegraphics[width=0.40\textwidth]{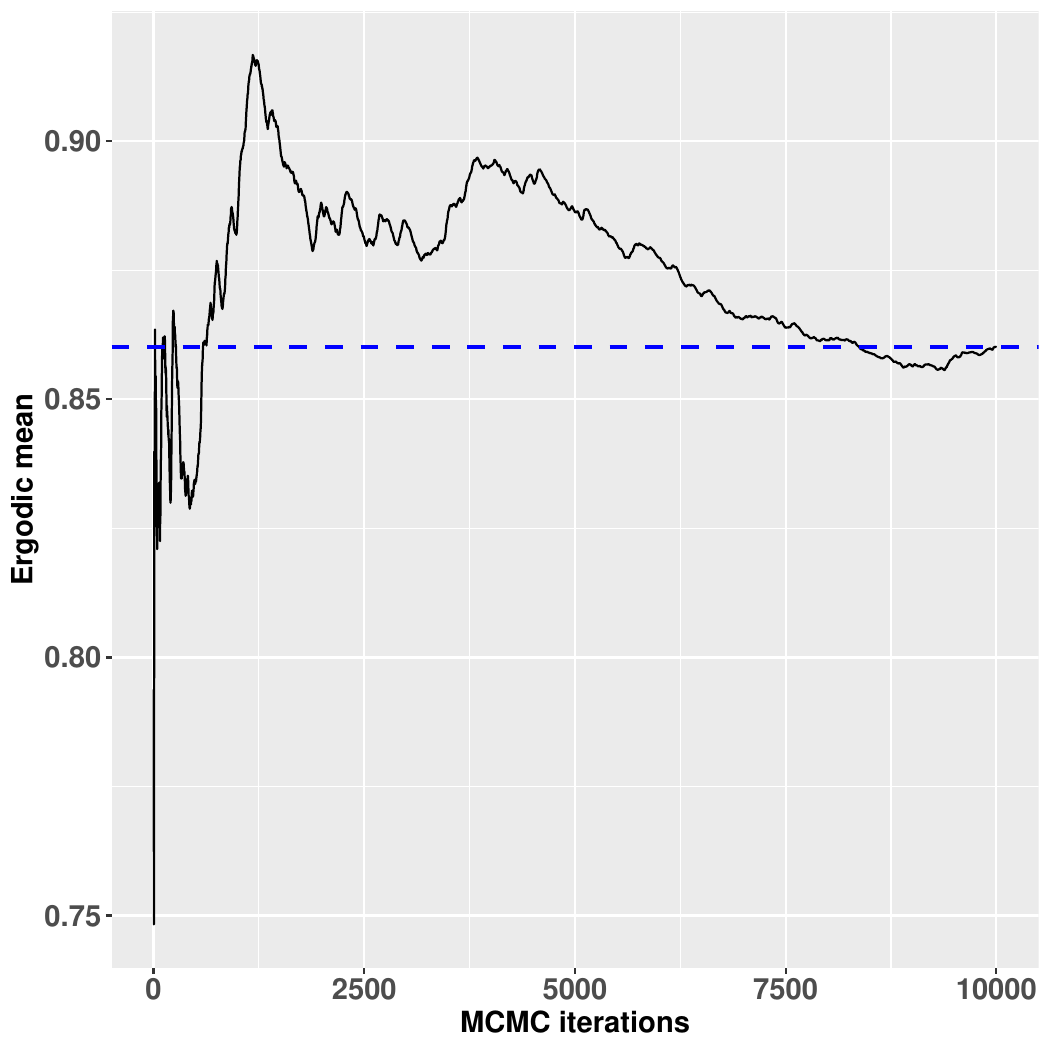} & \hspace{-0.2cm}
\includegraphics[width=0.40\textwidth]{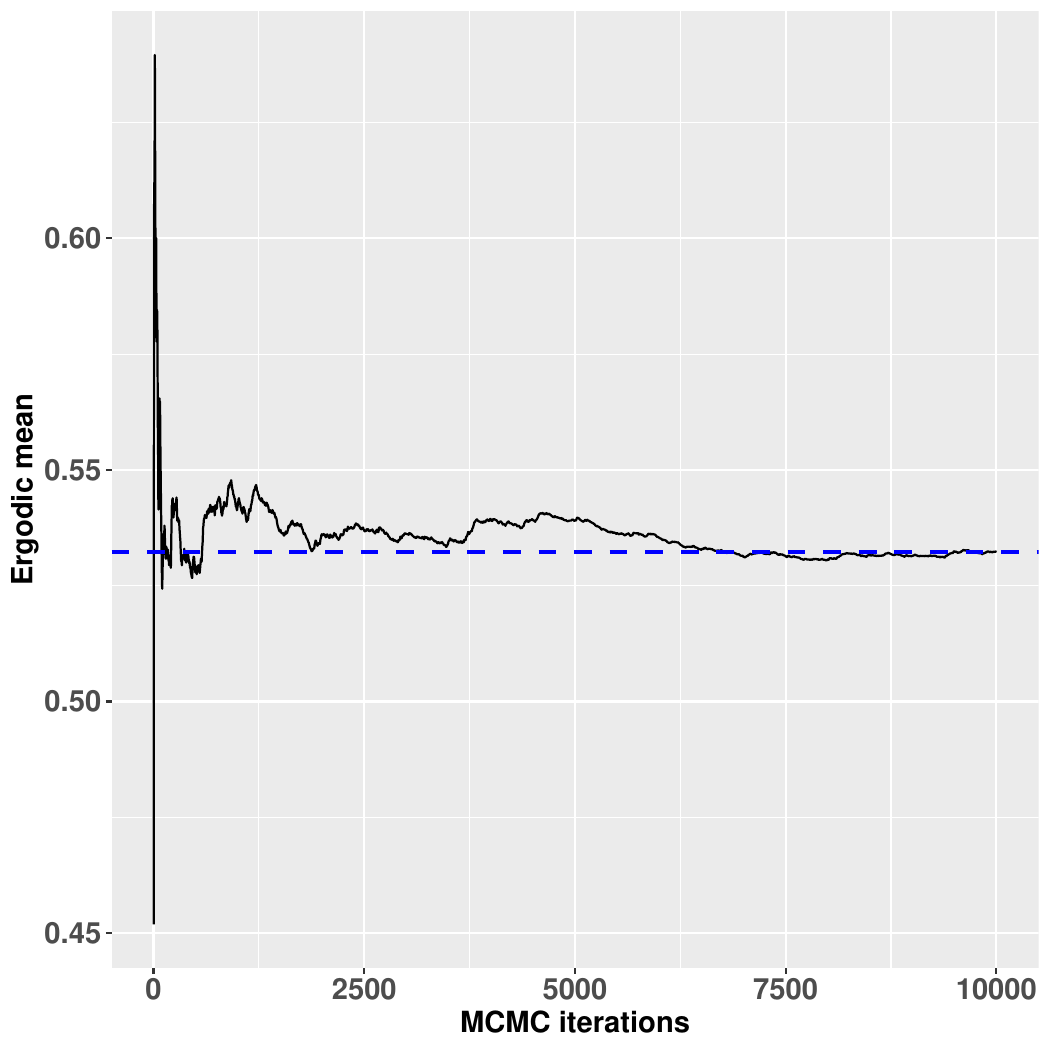} \\
\end{tabular}
\end{center}
\caption{Traceplots  and ergodic mean plots for the variance  of the random effects in cluster $c$,   $\sigma^{2}_c$, under the FH-SC$_{2}$ model.}
\label{fig:variances}
\end{figure}

\clearpage

\begin{figure}[ht]
\begin{center}
\begin{tabular}{ccc}
Traceplot for $\rho$ & \hspace{1cm} Ergodic mean plot for $\rho$ \\
\hspace{-1.0cm} \includegraphics[width=0.50\textwidth]{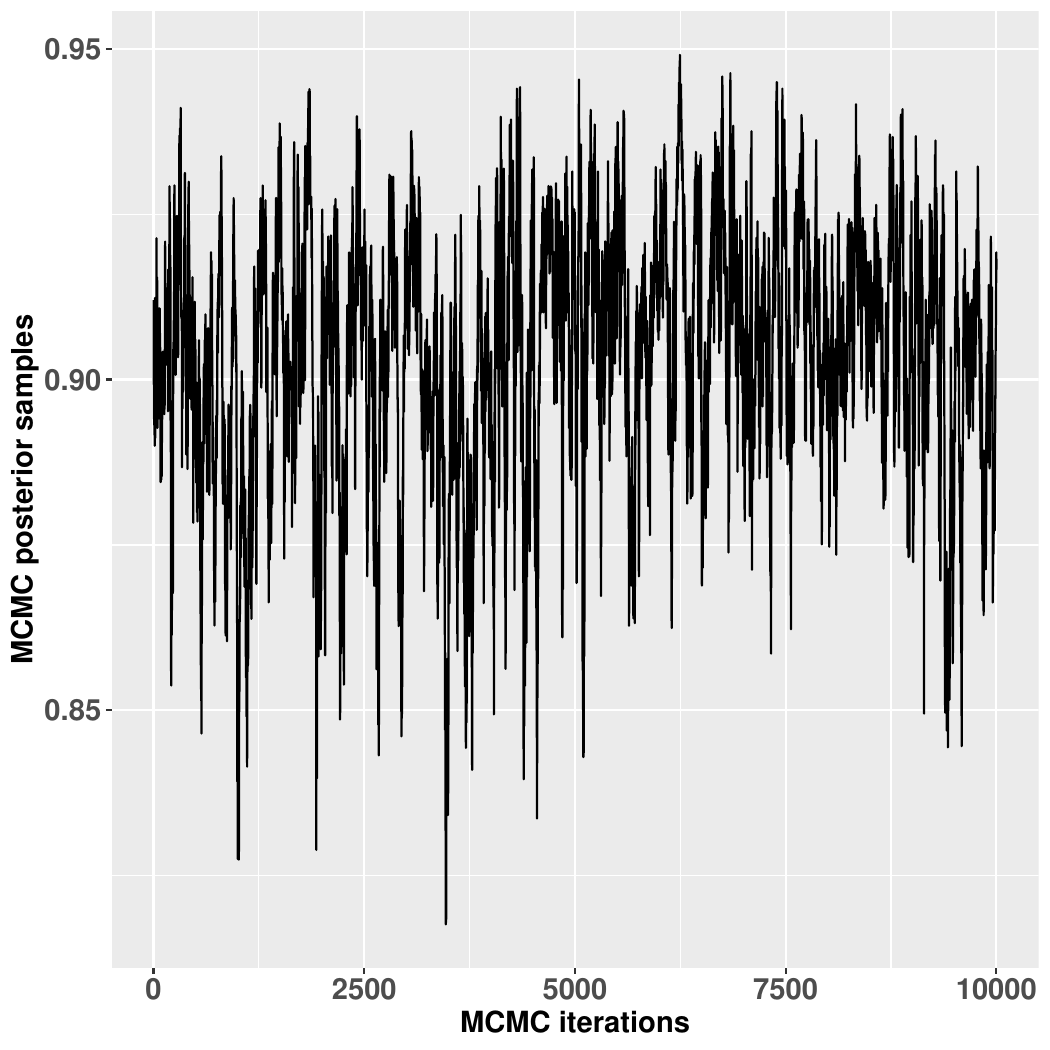} & \hspace{-0.2cm}
\includegraphics[width=0.50\textwidth]{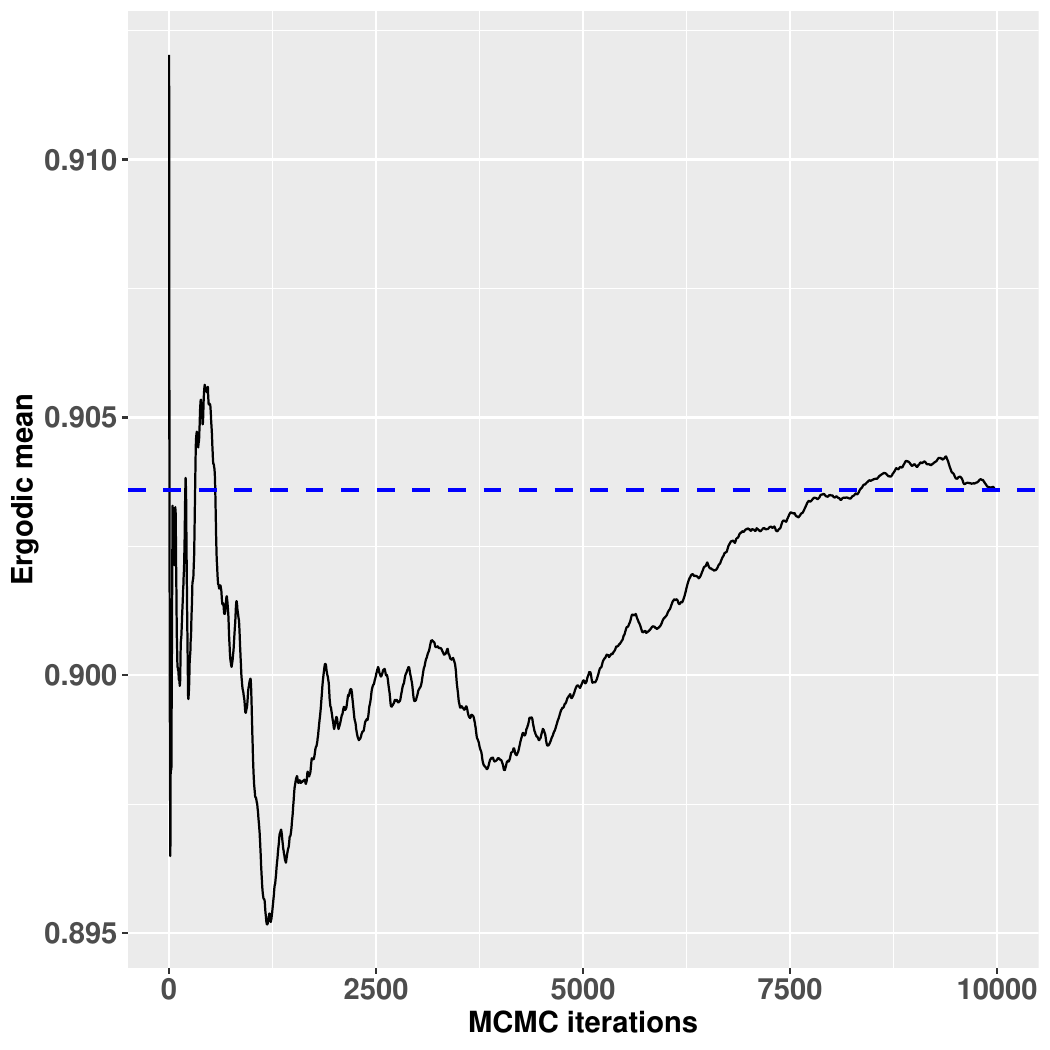} 
\end{tabular}
\end{center}
\caption{Traceplots  and ergodic mean plots for the  cluster regularization penalty,  $\rho$, under the FH-SC$_{2}$ model.}
\label{fig:penalty}
\end{figure}

\clearpage

\subsection{Diagnostic  plots to evaluate assumptions on the errors and random effects }
\label{sub_applied3}

To evaluate the assumptions imposed on the errors and random effects for the existing and proposed models, we consider the 
posterior samples of the standardized random effects and compute the standardized realized residuals produced under the FH and selected FH-SC$_{2}$ models.  
The posterior samples of the random effects produced under the FH and FH-SC$_{2}$ models are obtained using the output of Algorithms  \ref{alg:MCMC1} and \ref{alg:MCMC2}. 
The standardized realized residuals, $r_{j,c}$, under the FH and FH-SC$_{2}$ are computed using $r_{j,c}^{\textsf{FH}}= ( y_{j,c} - \hat{\theta}_{j,c}^{\textsf{FH}} )/\sqrt{D_{j,c}}$ and $r_{j,c}^{\textsf{FH-SC$_{2}$}}= ( y_{j,c} -  \hat{\theta}_{j,c}^{\textsf{FH-SC$_{2}$}} )/\sqrt{D_{j,c}}$, respectively. 

The FH model and FH-SC$_{2}$ account for clustering and non-clustering effects, respectively.  Figure \ref{fig:random_effects}  (left) illustrates that the posterior samples of the standardized random effects under FH-SC$_{2}$ are close to zero  and that the posterior densities for each cluster  are symmetric. In contrast, Figure \ref{fig:random_effects}  (right)  also shows that the density of the posterior means across small areas under the FH model exhibits three distinct peaks suggesting the presence of three distinct group patterns (clusters) that should be considered when estimating the PHIA. 

Furthermore, the residuals under FH-SC$_{2}$ are close to zero, with only a few outliers, as displayed in Figure \ref{fig:residuals}. In comparison, FH produces residuals with both positive and negative deviations from zero. This behavior of the standardized random effects and the standardized realized residuals indicate that the proposed FH-SC$_{2}$ may be more appropriate for estimating the PHIA
at the municipality level.

\clearpage

\begin{figure}[ht]
\begin{center}
\begin{tabular}{ccc}
\hspace{-0.6cm} \includegraphics[width=0.50\textwidth]{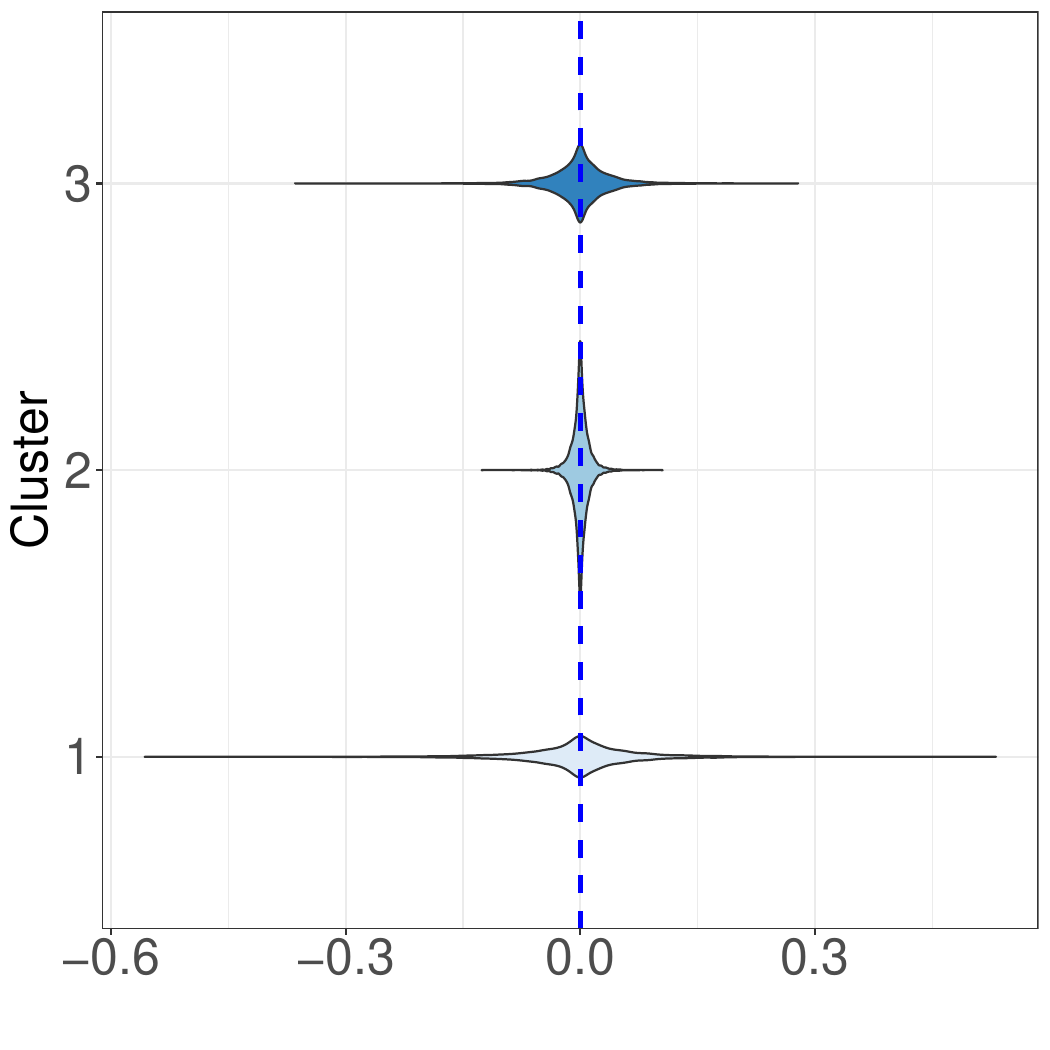} 
& \hspace{-0.2cm}
\includegraphics[width=0.50\textwidth]{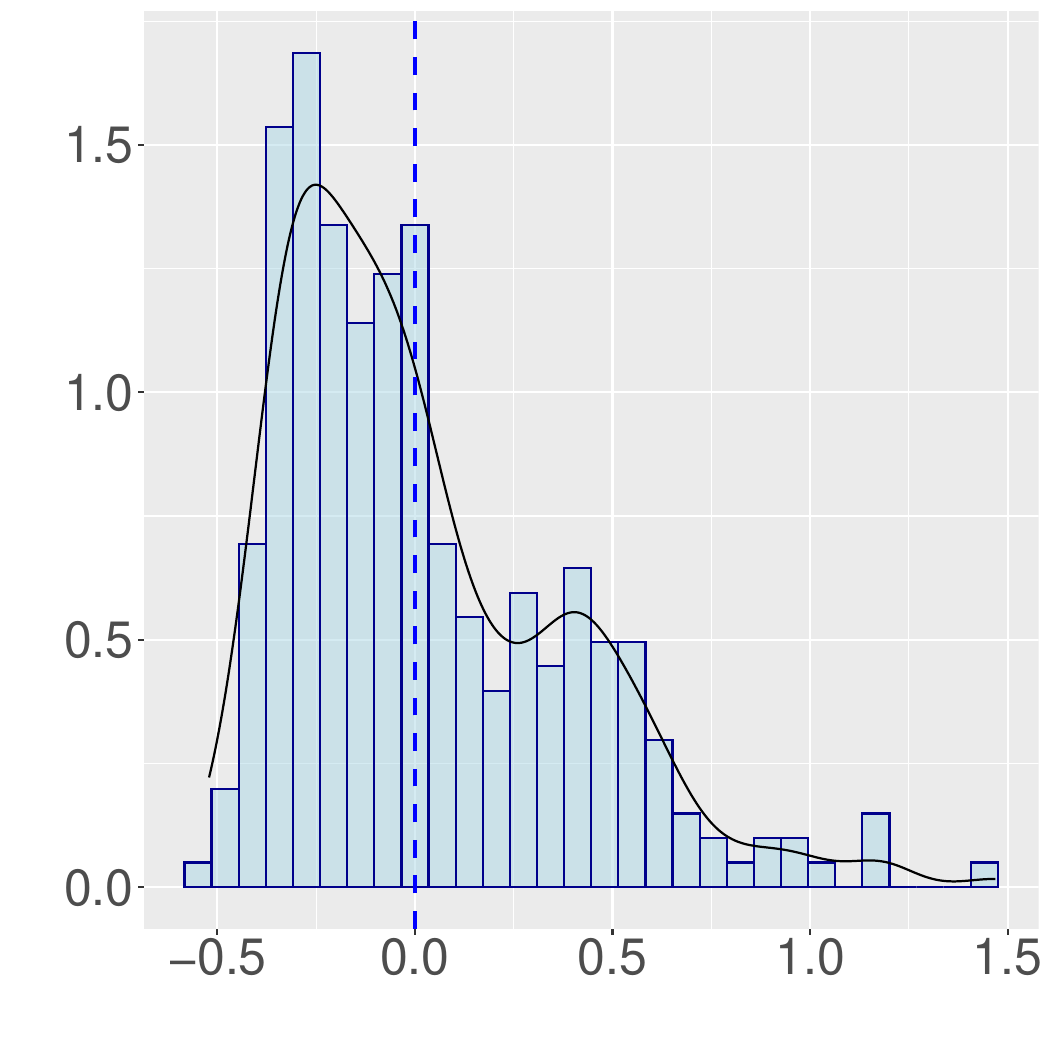} 
\end{tabular}
\end{center}
\caption{Standardized random effects under the FH and FH-SC$_{2}$ models.  Left: Posterior samples of the standardized random effects under FH-SC$_{2}$  for each cluster $c$. Right: Posterior means of standardized random effects across small areas under the FH. }
\label{fig:random_effects}
\begin{center}
\begin{tabular}{ccc}
(a)  Cluster 1 & \hspace{1cm} b)  Cluster 2  & \hspace{1cm}  (c)  Cluster 3  \\
\hspace{-1.9cm} \includegraphics[width=0.40\textwidth]{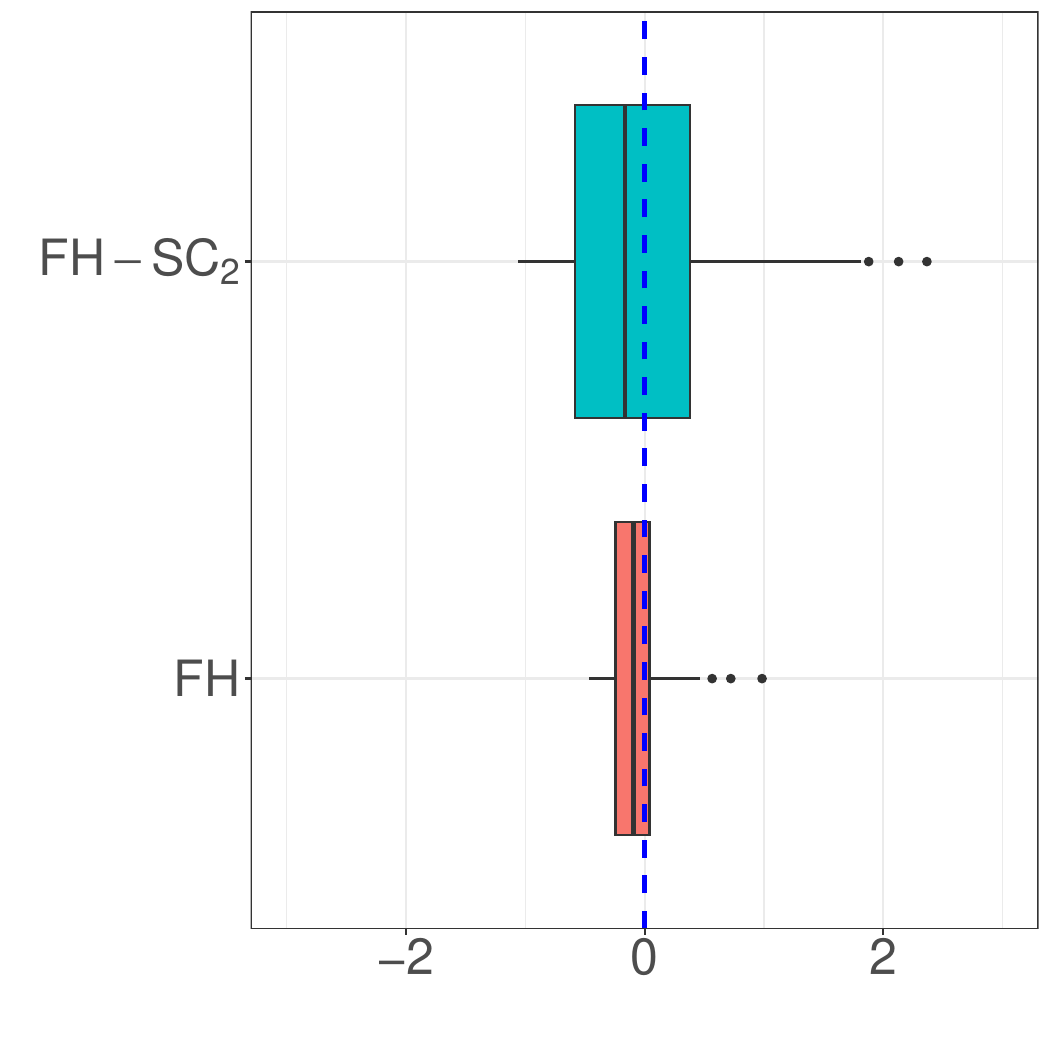}  & \hspace{-0.6cm}
\includegraphics[width=0.40\textwidth]{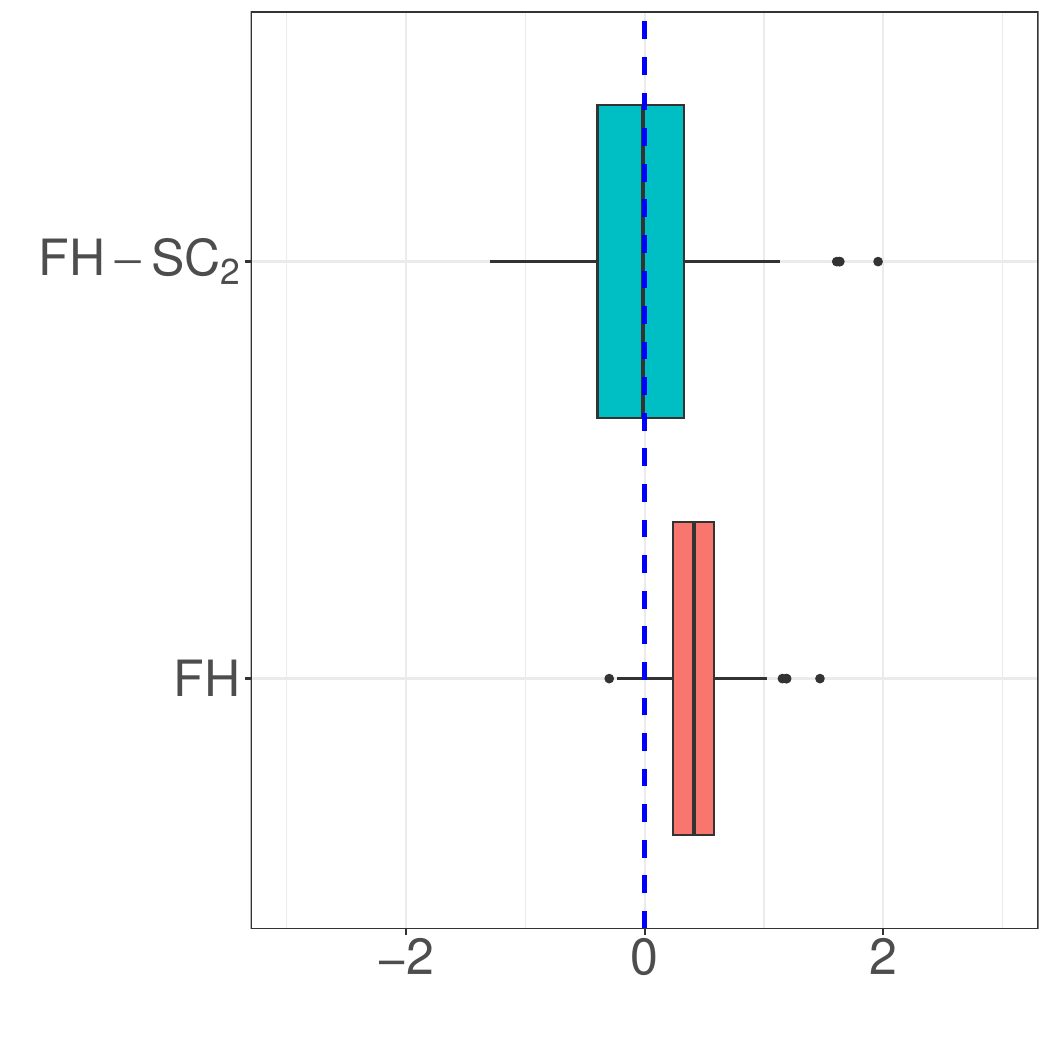} & \hspace{-0.6cm}
\includegraphics[width=0.40\textwidth]{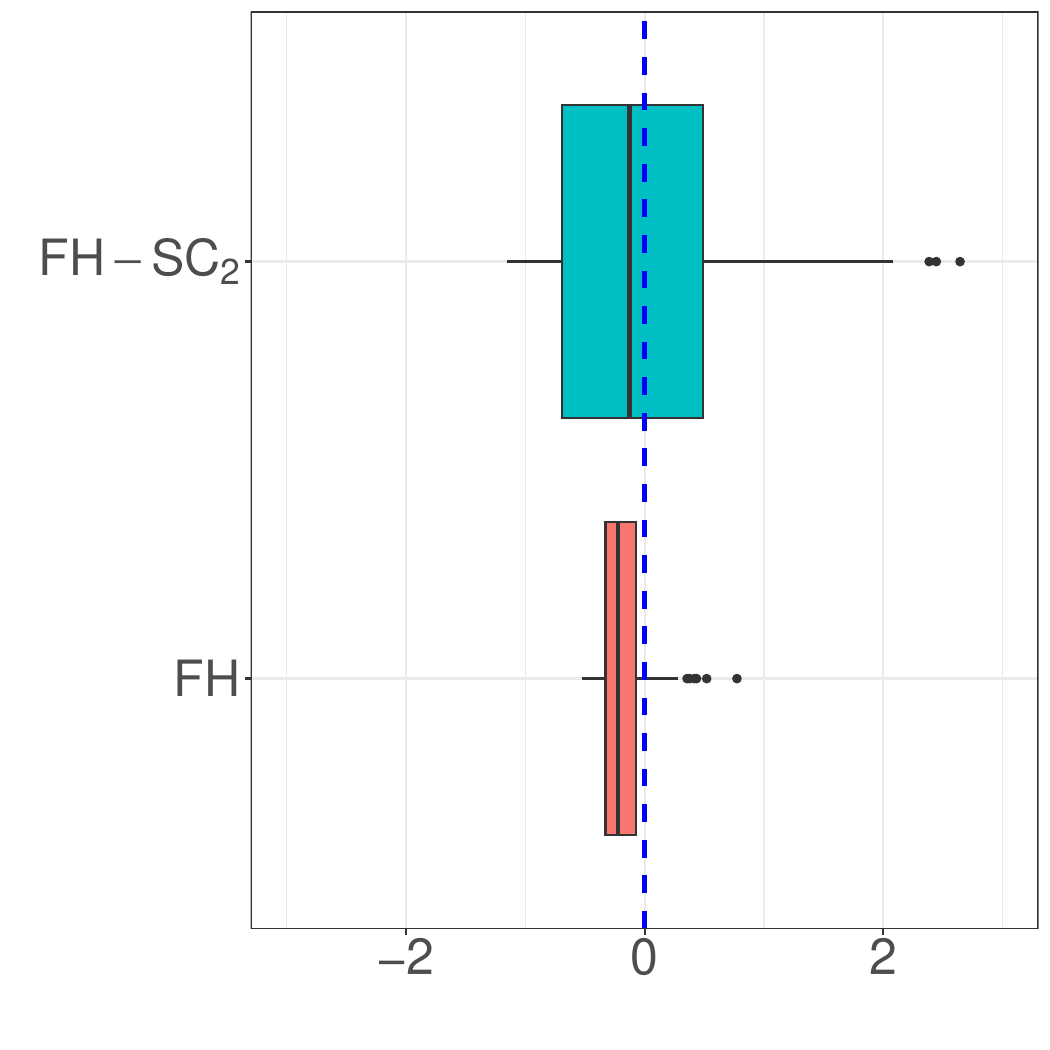} \\
\end{tabular}
\end{center}
\caption{Standardized realized residuals for the FH and FH-SC$_{2}$ models in each cluster.}
\label{fig:residuals}
\end{figure}


}

\clearpage

\bibliographystyle{chicago}







\bibliography{new_bibliof2.bib}


\end{document}